\documentclass[11pt]{article}

\usepackage[utf8]{inputenc}
\usepackage{authblk}

\usepackage{amsfonts,amsmath,amssymb,mathtools,amsthm}
\usepackage{graphicx}
\usepackage{stmaryrd}
\usepackage{thmtools}
\DeclareGraphicsExtensions{.eps}
\usepackage{float}

\usepackage{xcolor}
\definecolor{linkcolour}{rgb}{0.15,0.15,0.55}
\definecolor{urlcolour}{rgb}{0.15,0.15,0.55}
\definecolor{citecolour}{rgb}{0.15,0.15,0.55}

\usepackage[linktoc=all,backref=page]{hyperref}
	\hypersetup{
		colorlinks = true,
		linkcolor = linkcolour,
		urlcolor = citecolour,
		citecolor = citecolour,
		linktoc = all,
		hypertexnames = false,
		unicode = true,
		bookmarksnumbered = false,
		pdfmenubar = true,
		pdftoolbar = true}

\usepackage
	[top=2cm,
	bottom=2cm,
	left=2cm,
	right=2cm,
	a4paper]
	{geometry}
		\usepackage[margin=2cm]{caption}
\renewcommand{\theequation}{\thesection.\arabic{equation}}

\newcommand\encadremath[1]{\vbox{\hrule\hbox{\vrule\kern8pt
\vbox{\kern8pt \hbox{$\displaystyle #1$}\kern8pt}
\kern8pt\vrule}\hrule}}
\def\enca#1{\vbox{\hrule\hbox{
\vrule\kern8pt\vbox{\kern8pt \hbox{$\displaystyle #1$}
\kern8pt} \kern8pt\vrule}\hrule}}

\newcommand\framefig[1]{
\begin{figure}[bth]
\hrule\hbox{\vrule\kern8pt
\vbox{\kern8pt \vbox{
\begin{center}
{#1}
\end{center}
}\kern8pt}
\kern8pt\vrule}\hrule
\end{figure}
}

\newcommand\figureframex[3]{
\begin{figure}[bth]
\hrule\hbox{\vrule\kern8pt
\vbox{\kern8pt \vbox{
\begin{center}
{\mbox{\epsfxsize=#1.truecm\epsfbox{#2}}}
\end{center}
\caption{#3}
}\kern8pt}
\kern8pt\vrule}\hrule
\end{figure}
}
\newcommand\figureframey[3]{
\begin{figure}[bth]
\hrule\hbox{\vrule\kern8pt
\vbox{\kern8pt \vbox{
\begin{center}
{\mbox{\epsfysize=#1.truecm\epsfbox{#2}}}
\end{center}
\caption{#3}
}\kern8pt}
\kern8pt\vrule}\hrule
\end{figure}
}
 \makeatother

 \usepackage{mathdots}
 \usepackage{xcolor}
 \usepackage{amsfonts,amsmath,amssymb,mathtools,amsthm}
 \usepackage{graphicx}
 \usepackage{stmaryrd}
 \usepackage{thmtools}
 \DeclareGraphicsExtensions{.eps}
\renewcommand{\subsectionautorefname}{Subsection}
\renewcommand{\sectionautorefname}{Section}

\renewcommand{\thesection}{\arabic{section}}
\renewcommand{\theequation}{\arabic{section}-\arabic{equation}}
\makeatletter
\@addtoreset{equation}{section}
\makeatother
\newtheorem{theorem}{Theorem}[section]

\newtheorem{proposition}{Proposition}[section]
\newtheorem{lemma}{Lemma}[section]
\newtheorem{corollary}{Corollary}[section]

\theoremstyle{definition}
\newtheorem{remark}{Remark}[section]
\newtheorem{definition}{Definition}[section]

\def\br{\begin{remark}\rm\small}
\def\er{\end{remark}}
\def\bt{\begin{theorem}}
\def\et{\end{theorem}}
\def\bd{\begin{definition}}
\def\ed{\end{definition}}
\def\bp{\begin{proposition}}
\def\ep{\end{proposition}}
\def\bl{\begin{lemma}}
\def\el{\end{lemma}}
\def\bc{\begin{corollary}}
\def\ec{\end{corollary}}
\def\beaq{\begin{eqnarray}}
\def\eeaq{\end{eqnarray}}

\theoremstyle{definition}

\newcommand{\be}{\begin{equation}}
\newcommand{\ee}{\end{equation}}
\newcommand{\beq}{\begin{equation}}
\newcommand{\eeq}{\end{equation}}
\newcommand{\bea}{\begin{eqnarray}}
\newcommand{\eea}{\end{eqnarray}}
\newcommand{\beqq}{\begin{equation*}}
\newcommand{\eeqq}{\end{equation*}}
\newcommand{\beaa}{\begin{eqnarray*}}
\newcommand{\eeaa}{\end{eqnarray*}}

\newcommand{\diag}{{\operatorname{diag}}}

\newcommand{\td}{\tilde}

\newcommand\blfootnote[1]{%
  \begingroup
  \renewcommand\thefootnote{}\footnote{#1}%
  \addtocounter{footnote}{-1}%
  \endgroup
}

\newcommand{\Res}{\mathop{\,\rm Res\,}}
\title{\bf{Explicit Hamiltonian structure of the Flaschka–Newell Painlev\'{e} II hierarchy via symmetry reduction of the Painlev\'{e} IV hierarchy
}}

\author{$_{1}$Olivier Marchal\footnote{Universit\'{e} Jean Monnet Saint-\'{E}tienne, CNRS, Institut Camille Jordan UMR 5208, Institut Universitaire de France, Les Forges 2, 20 Rue du Dr Annino, 42000 Saint-\'{E}tienne, France.}\,\,,
$_{2}$Mohamad Alameddine\footnote{Section de math\'{e}matiques, Universit\'{e} de Gen\`{e}ve, Rue du conseil-g\'{e}n\'{e}ral 7-9, 1205 Gen\`{e}ve Switzerland.}\,\,,
$_{3}$Sergej Laub\footnote{Universit\'{e} Jean Monnet Saint-\'{E}tienne, CNRS, Institut Camille Jordan UMR 5208, Les Forges 2, 20 Rue du Dr Annino, 42000 Saint-\'{E}tienne, France.}
}

\date{\vspace{-5ex}}

\begin{document}

\maketitle

\textbf{Abstract}: 
We study the Hamiltonian structure of the Flaschka–Newell Painlev\'{e} II hierarchy via symmetry reduction of the associated space of meromorphic connections. Building on the realization of this hierarchy as a reduction of the Painlev\'{e} IV hierarchy via a $\mathbb{Z}_2$-symmetry, we construct a set of Darboux coordinates adapted to the involution. After a suitable change of trivialization, the symmetry acts diagonally in these coordinates, allowing the fixed-point locus to be explicitly described as a symplectic submanifold. This enables us to derive the reduced Hamiltonians after symmetry, thereby obtaining explicit expressions for the Hamiltonians and the Lax matrices of the Flaschka–Newell Painlev\'{e} II hierarchy. This strategy also illustrates how symmetry-adapted canonical Darboux coordinates enable explicit reductions of isomonodromic systems at the level of their underlying symplectic geometry.

\blfootnote{\textit{Email Addresses:}
$_{1}$\textsf{olivier.marchal@univ-st-etienne.fr}
$_{2}$\textsf{mohamad.alameddine@unige.ch}
$_{3}$\textsf{sergej.laub@etu.unige.ch}
}

\tableofcontents

\section{Introduction}\label{SecIntro}

The theory of isomonodromic deformations (i.e., monodromy-preserving deformations) is part of differential algebraic geometry and associates a nonlinear differential equation with a meromorphic linear differential system. The six Painlev\'{e} equations discovered through the classification of a class of second order differential equations are the simplest manifestations of this perspective \cite{Painleve,Picard,Fuchs,Garnier,Garnier1912,schlesinger1912klasse}. While the main idea was due to B. Riemann in 1857, the theory has seen recent developments bridging different topics such as moduli spaces of connections, representation theory, symplectic geometry, integrable systems, and analysis of nonlinear differential equations. The regular (Fuchsian) case on the Riemann sphere is well understood, leading to the Schlesinger equations \cite{schlesinger1912klasse,Fuchs}. However, the irregular case, consisting of deformations of meromorphic connections defined on the trivial vector bundle of arbitrary rank over the Riemann sphere with an arbitrary pole structure, still presents an active area of research. This requires a notion of generalized monodromy data, which is obtained by considering Stokes data around the irregular singularities in addition to the usual monodromies. The space of deformations itself needs to be extended beyond the space of complex structures of the Riemann surface by including the irregular type of the monodromy \cite{BertolaMo2005,BoalchThesis,Boalch2001,Boalch2014,JimboMiwa,JimboMiwaUeno}. By considering admissible deformations, one obtains a symplectic fibration over the base of deformation parameters. We adopt this viewpoint throughout the present work. If one successfully trivializes this fibration, i.e., defines canonical coordinates on the fibres, the evolutions of the chosen coordinates (the deformation equations) associated with vector fields provide time-dependent Hamiltonians. These are the non-autonomous Hamiltonians of the Painlev\'{e} equations, as derived by J. Malmquist \cite{Malmquist1922} for the Painlev\'{e} cases. This geometric picture, whose origin lies in the works of the Japanese school \cite{JimboMiwa,JimboMiwaUeno}, has been extended by many authors and provides one of the fundamental links between algebraic geometry and integrable systems.   

\medskip

The problem of obtaining Hamiltonian formulations of the isomonodromy equations has been actively studied from many perspectives, revealing rich structures and opening numerous research directions across several areas of mathematics. For instance, one of the most famous examples is the JMMS \cite{JMMS} equations. These equations were the starting point for J. Harnad \cite{Harnad_1994} who established a duality between meromorphic differential systems defined on different rank bundles admitting different polar structures, obtaining two realizations of the JMMS Hamiltonian structure related by a symmetry. This Harnad duality was generalized by P. Boalch \cite{Boalch2012} (along with the JMMS Hamiltonian equations) to a continuous $\text{SL}_2(\mathbb{C})$ action initiating the systematic classification of isomonodromic systems. This provided evidence for the existence of various isomorphisms between different moduli spaces of connections generalizing Okamoto's observation  \cite{Okamoto1986,zbMATH00929752} that links the Painlev\'{e} equations to affine Dynkin diagrams. Different Lax representations leading to the same Painlev\'{e} equations have been established for several examples \cite{MazzoccoP6}. The current understanding is that moduli spaces of meromorphic connections provide a geometric framework for isomonodromic deformations and Painlev\'{e}-type equations arise from different classes of connections. Understanding the relations between these classes is a fundamental problem in the geometry of isomonodromic systems. Recently, two of the authors \cite{alameddine2026symmetryreductionpainleveiv} demonstrated a different kind of morphism between isomonodromic systems. More precisely, they constructed an embedding via a $\mathbb{Z}_2$-reduction, which yields the Flaschka–Newell Painlev\'{e} II (FN PII) hierarchy as a sub-hierarchy of the Painlev\'{e} IV (PIV) hierarchy.

\medskip

The FN PII hierarchy is the hierarchy whose first member is the FN Lax pair giving the PII equation discovered in \cite{Flaschka1980}. The higher members of the hierarchy are nonlinear differential equations of order $2d$ obtained as a reduction of the mKdV hierarchy, so that its origin is an isospectral integrable system rather than an isomonodromic one \cite{articleCJM,mazzocco2007hamiltonian}. The Lax matrix of the mKdV  hierarchy itself is of particular interest for this work, it admits a pole structure given by one irregular pole at infinity of order $r_\infty=2d+2$ (here $d\in\mathbb{N}^*$) and another regular finite pole $r_0=1$ at zero. Note that this FN PII hierarchy is different from the Jimbo-Miwa-Ueno PII hierarchy (JMU PII), described in \cite{JimboMiwa} (explicitly characterized in \cite{marchal2024hamiltonianrepresentationisomonodromicdeformations}) that corresponds to only a single irregular pole at infinity. Indeed, even if the first member of the JMU PII hierarchy and the FN PII hierarchy provide the same Painlev\'{e} II equation, the corresponding Lax pairs differ and are related by a complicated Laplace transform. Moreover, even the symplectic structures differ for higher members of the hierarchies and thus one needs to be precise in the terminology associated with the Painlev\'{e} II hierarchy. In this article, we shall only consider the FN Painlev\'{e} II hierarchy.  As mentioned above, the space of connections for the FN PII hierarchy corresponds to a subset of connections with a regular finite pole and an irregular untwisted pole at infinity of even degree. Such connections correspond to the PIV hierarchy studied in \cite{marchal2024hamiltonianrepresentationisomonodromicdeformations,MarchalAlameddineIsospectralIsomono2023}. Thus, the starting point of \cite{alameddine2026symmetryreductionpainleveiv} was the space of such connections ultimately leading to a symmetry reduction providing the FN PII space of connections. The key point is that the connection matrices of the FN PII hierarchy admit an additional symmetry seen as a grading on the loop algebra on which the connection is defined and decomposed into even and odd entries (with a relation between the entries). The moduli space itself admits an involution imposed by an involution on the horizontal sections of the meromorphic differential system. This idea motivated a $\mathbb{Z}_2$-symmetry imposed on the space of connections of the PIV hierarchy, and provided a reduction of the symplectic fibration, of the base of deformation parameters, and the time-dependent Hamiltonians. This enables us to identify the space of connections of the FN PII hierarchy with an invariant subspace of the PIV hierarchy. While this observation has been made explicit for the first cases of the PIV hierarchy in \cite{alameddine2026symmetryreductionpainleveiv}, the authors left open the question of finding trivializations of the fibration that are compatible with the symmetry reduction. In other words, is there a set of Darboux coordinates $(\mathbf{Q},\mathbf{P})$ for the PIV hierarchy for which the symmetry is more transparent and allows one to obtain the Hamiltonian structure explicitly after symmetry reduction? Although the fixed-point submanifold is known, and the symmetry is explicitly given, canonical coordinates rarely restrict to canonical coordinates on the symmetry-reduced space. Indeed, the usual Darboux coordinates discussed in \cite{alameddine2026symmetryreductionpainleveiv} mix the invariant and anti-invariant evolution directions under the involution preventing a simple reduction for the Hamiltonian structure.

\medskip

This paper answers this question in the affirmative. One of the main results of the paper is to provide a set of Darboux coordinates that we refer to as the symmetric geometric set of Darboux coordinates trivializing the symplectic structure and through which the $\mathbb{Z}_2$-reduction becomes straightforward. This result allows us to identify the Lax matrices with those proposed in \cite{mazzocco2007hamiltonian}, and through the resolution of the compatibility equations, it enables us to obtain the full Hamiltonian evolution for the FN PII hierarchy in an explicit way. This question was previously left open in \cite{mazzocco2007hamiltonian} and remained unresolved in \cite{alameddine2026symmetryreductionpainleveiv}. In order to do so, we will use another representation of the $\mathbb{Z}_2$ group action on the Lax matrices. In fact, we will define a representation automorphism $S$ that induces a different normalization for the isomonodromic system. The matrix will appear as a conjugation action on the horizontal sections and this decomposes the action into a gauge transformation (left multiplication) and a change of basis of fundamental solutions (right multiplication). However, this automorphism will be crucial when imposing the symmetry in which the Pauli matrix $\sigma_1$ imposed previously will become $\sigma_3$ which is a diagonal matrix (thus, acting diagonally rather than swapping the eigenspaces). In other words, we will diagonalize the symmetry reduction itself. The advantage of this change is that the entries of the Lax matrices in this gauge have definite parity, with no mixing between them, thus emphasizing the parity, and no swapping occurs between entries when imposing the symmetry. This allows the symmetry to be imposed in a transparent and effective way. We believe that this idea is more general than the case studied in this article and we elaborate on this in more details in the final section. Furthermore, the construction produces explicit expressions for the Lax representation, including the deformation matrices, in canonical Darboux coordinates adapted to the symmetry, thereby completing the full analysis of the problem.

\medskip

From a broader perspective, this work illustrates that symmetry reductions of isomonodromic systems are best performed at the level of the underlying symplectic geometry rather than at the level of differential equations. Once the symmetry is made compatible with the normalization and the trivialization of the symplectic structure, an explicit Hamiltonian description becomes accessible on the reduced space. 

\medskip

The article is organized as follows. \autoref{SecP4hierarchy} is devoted to the study of the even PIV hierarchy: we will establish in this section the new normalization of the isomonodromy problem and introduce the new set of canonical Darboux coordinates $(\mathbf{Q_\infty}, \mathbf{P_\infty})$ adapted to the upcoming symmetry (\autoref{DefNewcoor2}). This is one of the main results of the paper answering the question left open in \cite{alameddine2026symmetryreductionpainleveiv}. In addition to this definition, we will obtain the Hamiltonian structure of this new set of Darboux coordinates $(\mathbf{Q_\infty}, \mathbf{P_\infty})$ (\autoref{TheoLaxMatricesQinfty} and \autoref{HamTheorem2}). This Hamiltonian structure will be given in terms of the coordinates $(\mathbf{Q_\infty}, \mathbf{P_\infty})$ after solving the compatibility equations in the corresponding oper gauge with the corresponding oper Darboux coordinates $(\mathbf{q},\mathbf{p})$ following the strategy of  \cite{marchal2024hamiltonianrepresentationisomonodromicdeformations} but with a different normalization of the Lax matrix (\autoref{HamTheorem}). We will then briefly review the setup of \cite{mazzocco2007hamiltonian} giving the Lax pair of the FN PII hierarchy and canonical Darboux coordinates on the fibration in \autoref{SecFNP2} together with its link to the reduction of the mKdV hierarchy. Ultimately, in \autoref{secSymmetry} we introduce the symmetry (\autoref{DefSymmetry}) that reduces the PIV hierarchy to the FN PII hierarchy. Since the Darboux coordinates $(\mathbf{Q_\infty}, \mathbf{P_\infty})$ are adapted to the symmetry, we obtain explicit formulas for the Lax matrices and the Hamiltonians of the FN PII hierarchy (\autoref{TheoReducedLaxHamSym}) and we can provide an explicit identification of our Darboux coordinates and Lax matrices with those of \cite{mazzocco2007hamiltonian} in \autoref{PropIdentification}. We illustrate the main results in \autoref{secExamples} on the first two elements of the hierarchy and devote \autoref{secConclusion} for the concluding remarks and the discussion of the applicability of our methods in a more general context.

\subsection*{Statements and declaration}
The authors declare that they have no conflict of interest. Data sharing is not applicable to this article as no datasets were generated or analyzed during the current study.

\subsection*{Acknowledgments}
M. A. was partially funded by the Swiss National Science Foundation (SNSF) under the grant agreement PZ00P2 223297 and by the DFG -- project ID 551478549. O. M. was supported by the fundamental junior IUF grant G752IUFMAR.

\section{The even Painlev\'{e} IV hierarchy}\label{SecP4hierarchy}

The main goal of this section is to introduce the even PIV hierarchy, a sub-hierarchy of the PIV hierarchy whose infinite pole order is taken to be even. This hierarchy is thus naturally contained in the generic isomonodromic reformulation of \cite{MarchalAlameddineIsospectralIsomono2023,marchal2024hamiltonianrepresentationisomonodromicdeformations}. However, we will introduce a new normalization of the Lax matrix that yields a different set of Darboux coordinates. In particular, this choice of normalization and the associated Darboux coordinates is made so that they are adapted to the upcoming symmetry reduction developed in \autoref{secSymmetry}.

\subsection{Isomonodromic formulation and normalization of Lax matrices}
The even PIV hierarchy governs the isomonodromic deformations of rank 2 generic connections defined on the trivial vector bundle over the Riemann sphere with an effective divisor having one regular pole ($n=1$ and $r_0=1$ in the notation of \cite{marchal2024hamiltonianrepresentationisomonodromicdeformations}) located at $X_0$, and one irregular pole of even order $r_\infty=2d+2\geq 4$ located at $ \{ \infty \}$. The space of connections is defined as follows:

\begin{definition}[Space of connections of the even PIV hierarchy]\label{DefP4hierarchy} Considering a meromorphic connection defined on the trivial rank-two vector bundle over the Riemann sphere, we define the following Poisson space:
\bea
F_{d} &\coloneqq& \Big\{\hat{L}(\lambda) = \sum_{k=1}^{2d+1} \hat{L}^{[\infty,k]} \lambda^{k-1} + \frac{\hat{L}^{[X_0,0]}}{\lambda-X_0}  \text{ with }\,\, \left\{\hat{L}^{[X_0,0]}\right\}\cup \left\{\hat{L}^{[\infty,k]}\right\}_{k\in \llbracket 1,2d+1\rrbracket} \in \left(\mathfrak{sl}_2(\mathbb{C})\right)^{2d+2}\cr
&& \text{ and }  \left(\hat{L}^{[\infty,2d+1]}, \hat{L}^{[X_0,0]}\right)  \text{ have distinct eigenvalues}
\Big\}/\text{SL}_2(\mathbb{C}),
\eea
where the $\text{SL}_2(\mathbb{C})$ action is the conjugation action of the reductive group. After trivializing the bundle, one obtains the meromorphic differential system
\beq \label{eqhatL}\partial_\lambda \hat{\Psi}(\lambda)= \hat{L}(\lambda)\hat{\Psi}(\lambda)\eeq
defined in the usual complex-analytic charts of the Riemann sphere with $\hat{\Psi}(\lambda)$ the matrix of horizontal sections solving the system. 
\end{definition}

From the theory of meromorphic differential systems, it is always possible to (locally) transform a generic (leading parts at each pole having distinct eigenvalues) connection on the formal punctured disk to a diagonal normal form. This is the classical Hukuhara-Levelt-Turrittin theorem \cite{Birkhoff,Wasowbook}. 

\begin{proposition}[Local diagonalization and irregular times] There exists a local analytic gauge transformation, giving the Birkhoff factorization or formal Turrittin-Levelt form in which the singular part of the connection matrix is diagonalizable:
\beq \hat{L}(\lambda) \overset{\lambda\to \infty}{\sim} \diag\left(-\sum_{k=0}^{2d+1} t_{\infty,k}\lambda^{k-1},\sum_{k=0}^{2d+1} t_{\infty,k}\lambda^{k-1} \right)+ O(\lambda^{-2}).\eeq
Similarly, at $\lambda=X_0$, one has
\beq \hat{L}^{[X_0,0]}\overset{\lambda\to X_0}{\sim} \diag(-t_{X_0,0},t_{X_0,0}).\eeq
The parameters $\mathbf{t}\coloneqq(t_{\infty,k})_{1\leq k\leq 2d+1}$ appearing in the formal normal forms are called the irregular times: they form the base manifold of parameters for isomonodromic deformations together with the location of the pole $X_0$. On the contrary, the parameters $(t_{\infty,0},t_{X_0,0})$ are called the monodromies at the poles, they constitute the exponent of formal monodromy of the connection. They are treated as constants while performing isomonodromic deformations.
\end{proposition}

The geometric picture is then captured by the definition of a symplectic fibration over the base $\mathbb{B}=\{X_0,t_{\infty,1},\dots,t_{\infty,2d+1}\}$ of deformation parameters, each vector field of the tangent space to the base will provide a time-dependent Hamiltonian function on the fibration. This fibration is the framework for the Hamiltonian and symplectic nature of the isomonodromy deformation equations. It is a symplectic variety of dimension $2 (r_\infty+1-3)= 4d$. Therefore, in our setting, a general isomonodromic deformation corresponds to a vector field 
\beq \label{GeneralDef}\mathcal{L}_{\boldsymbol{\hat{\alpha}}}\coloneqq\sum_{k=1}^{2d+1}\hat{\alpha}_{\infty,k}\partial_{t_{\infty,k}} + \hat{\alpha}_{X_0}\partial_{X_0},\eeq
where $\hat{\boldsymbol{\alpha}} \in \mathbb{C}^{2d+2}$ is a complex vector of the tangent space. This gives rise to a space of deformations of dimension $2d+2$. To a general deformation operator \eqref{GeneralDef}, we associate an auxiliary matrix $\hat{A}_{\hat{\boldsymbol{\alpha}}}(\lambda)$ seen also as an auxiliary connection over the base of times:
\beq \label{AuxEq}\mathcal{L}_{\hat{\boldsymbol{\alpha}}}[\hat{\Psi}(\lambda)]=\hat{A}_{\hat{\boldsymbol{\alpha}}}(\lambda)\hat{\Psi}(\lambda),\eeq
which is a meromorphic function in $\lambda$ with poles bounded by those of $\hat{L}(\lambda)$. The Hamiltonian structure obtained from the compatibility equations of \eqref{eqhatL} and \eqref{AuxEq} is given by the compatibility equations:
\beq \label{compatibility} \partial_\lambda \hat{A}_{\hat{\boldsymbol{\alpha}}}(\lambda)- \mathcal{L}_{\hat{\boldsymbol{\alpha}}}\left[\hat{L}(\lambda)\right] + \left[\hat{A}_{\hat{\boldsymbol{\alpha}}}(\lambda),\hat{L}(\lambda) \right]=0.
\eeq

It admits various symmetries, one of which is the set of M\"obius transformations of $\lambda$ \cite{Boalch2012,marchal2024hamiltonianrepresentationisomonodromicdeformations}. In practice, since we have fixed a pole at infinity, there remain two degrees of freedom from the action of M\"obius transformations that allow us to choose two time parameters. In \cite{marchal2024hamiltonianrepresentationisomonodromicdeformations}, the canonical choice taken was $(t_{\infty,2d+1},t_{\infty,2d})=(1,0)$ and simplified versions of the Hamiltonians and Lax matrices were derived under this canonical choice. However, since we want to identify with the FN PII hierarchy, we will make another choice that is more compatible with the upcoming symmetry. 
\begin{definition}[Normalization from M\"obius transformations]\label{NormalizationMobius} We use the M\"obius transformations of $\lambda$ to set:
\begin{align}
    X_0=0, \qquad \qquad  \text{ and }\qquad \qquad  t_{\infty,2d+1}=-4^d.
\end{align} 
\end{definition}

\begin{remark}The value of $t_{\infty,2d+1}$ is chosen to match with \cite{mazzocco2007hamiltonian} but most formulas derived in this article will be given for an arbitrary value of this parameter. However, we will never consider variations with respect to $t_{\infty, 2d+1}$ and consider it fixed.
\end{remark}

In particular, note that this reduces the dimension of the deformation space to $2d$.
\beq \label{GeneralDef2}\mathcal{L}_{\boldsymbol{\alpha}}\coloneqq\sum_{k=1}^{2d}\alpha_{\infty,k}\partial_{t_{\infty,k}} \eeq
the corresponding general deformation operator with $\boldsymbol{\alpha}\coloneqq\left(\alpha_{\infty,1},\dots,\alpha_{\infty,2d}\right)\in \mathbb{C}^{2d}$ the corresponding vector in the reduced tangent space. It induces the auxiliary matrix $A_{\boldsymbol{\alpha}}(\lambda)$ meromorphic in $\lambda$ defined by
\beq \label{AuxEq2}\mathcal{L}_{\boldsymbol{\alpha}}[\hat{\Psi}(\lambda)]=\hat{A}_{\boldsymbol{\alpha}}(\lambda)\hat{\Psi}(\lambda).
\eeq

Another choice to be made is offered by the $\text{SL}_2(\mathbb{C})$ conjugation action. In fact, in order to obtain explicit formulas for the Lax matrices, one needs to choose a specific element in the orbit and this choice provides natural Darboux coordinates. Equivalently, the choice of the element in the orbit can be seen as fixing a normalization on the Lax matrix at infinity (or at one of its pole). The choice that we shall use in this paper is given by the following definition. We stress that it is not the same as the one used in the literature (usually taken to be that the leading order at infinity is diagonal) and in particular differs from the one used in \cite{marchal2024hamiltonianrepresentationisomonodromicdeformations}.

\begin{definition}[Choices of normalization of the Lax matrix]\label{DefGauges} We denote $\td{\Psi}(\lambda)$, $\td{L}(\lambda)$ and $\td{A}_{\boldsymbol{\alpha}}(\lambda)$ the representatives of $F_d$ such that
\beq \td{L}(\lambda)\overset{\lambda\to \infty}{=}\begin{pmatrix}-t_{\infty,2d+1}&0\\0&t_{\infty, 2d+1}\end{pmatrix}\lambda^{2d} +\begin{pmatrix} X& \omega\\ X&X\end{pmatrix}\lambda^{2d-1} +O(\lambda^{2d-2}).\eeq
The parameter $\omega$ is arbitrary and corresponds to the remaining degree of freedom of the $\text{SL}_2(\mathbb{C})$ conjugation action by diagonal matrices.
We will denote $\check{\Psi}(\lambda)$, $\check{L}(\lambda)$ and $\check{A}_{\boldsymbol{\alpha}}(\lambda)$ the matrices corresponding to the conjugation action
\beq \check{\Psi}(\lambda)=\frac{1}{\sqrt{2}}\begin{pmatrix} 1&-1\\ 1&1\end{pmatrix} \td{\Psi}(\lambda)\frac{1}{\sqrt{2}}\begin{pmatrix} 1&1\\ -1&1\end{pmatrix} \coloneqq S\td{\Psi}(\lambda)S^{-1} \,\,\text{ with }\, S\coloneqq\frac{1}{\sqrt{2}}\begin{pmatrix}
    1&-1\\ 1&1
\end{pmatrix}
\eeq 
and corresponding to the normalization of the Lax matrix: 
\beq \label{NormalizationcheckL}\check{L}(\lambda)\overset{\lambda\to \infty}{=}\begin{pmatrix}0&-t_{\infty,2d+1}\\-t_{\infty, 2d+1}&0\end{pmatrix}\lambda^{2d} +O(\lambda^{2d-1}).
\eeq
\end{definition}

The new normalization $\check{\Psi}(\lambda)$ will play a central role for the identification with the results of \cite{mazzocco2007hamiltonian}. It is the gauge that is naturally compatible with the upcoming symmetry of \autoref{secSymmetry}. In this gauge, one has:  
\begin{align}
 \check{L}(\lambda)=S \td{L}(\lambda) S^{-1}\,\, \qquad 
 \text{and} \qquad 
 \check{A}_{\boldsymbol{\alpha}}(\lambda)=S \td{A}_{\boldsymbol{\alpha}}(\lambda) S^{-1}.
\end{align}
Equivalently, in terms of the entries,
\begin{align}\label{checkLtdL}
 \check{L}_{1,1}(\lambda)&=-\frac{1}{2}\left(\td{L}_{1,2}(\lambda)+\td{L}_{2,1}(\lambda)\right), \qquad \qquad
\check{L}_{1,2}(\lambda)=\td{L}_{1,1}(\lambda)+\frac{1}{2}\left(\td{L}_{1,2}(\lambda)-\td{L}_{2,1}(\lambda)\right),\cr
\check{L}_{2,1}(\lambda)&=\td{L}_{1,1}(\lambda)-\frac{1}{2}\left(\td{L}_{1,2}(\lambda)-\td{L}_{2,1}(\lambda)\right), \qquad 
\qquad \check{L}_{2,2}(\lambda)=-\check{L}_{1,1}(\lambda),
\end{align}
or
\begin{align}\label{checkLtdL2}
 \td{L}_{1,1}(\lambda)&=\frac{1}{2}\left(\check{L}_{1,2}(\lambda)+\check{L}_{2,1}(\lambda)\right), \qquad \qquad
\td{L}_{1,2}(\lambda)=-\check{L}_{1,1}(\lambda)+\frac{1}{2}\left(\check{L}_{1,2}(\lambda)-\check{L}_{2,1}(\lambda)\right),\cr
\td{L}_{2,1}(\lambda)&=-\check{L}_{1,1}(\lambda)-\frac{1}{2}\left(\check{L}_{1,2}(\lambda)-\check{L}_{2,1}(\lambda)\right), \qquad 
\qquad \td{L}_{2,2}(\lambda)=-\td{L}_{1,1}(\lambda).
\end{align}

Note that if $\td{L}(\lambda)$ is the usual normalization used in the literature, $\check{L}(\lambda)$ just corresponds to another choice of representative in the orbit defining $F_d$. In particular, the underlying symplectic structure is independent of the choice of representative.

\begin{remark}
    We will see that the matrix $S$ is a representation automorphism of the imposed $\mathbb{Z}_2$-symmetry in \autoref{secSymmetry}. Apart from this observation, we use it here as an orthogonal conjugation that rotates the bases of the loop algebra on which the Lax matrices are defined.   
\end{remark}

\subsection{Parameterization using oper and geometric Darboux coordinates}\label{Sec22}
In this section, we discuss some choices of Darboux coordinates that were used previously to obtain explicit parameterizations of the Lax matrices and the underlying Hamiltonian structure reached through the compatibility of the system. These coordinates will be useful in adapting our formalism to the symmetry reduction in the next subsection. 

\medskip

The first set of canonical Darboux coordinates used to solve the compatibility equations is the set of “oper Darboux coordinates" denoted $(\mathbf{u},\mathbf{v})\coloneqq(u_1,\dots,u_{2d},v_1,\dots,v_{2d})$ in this article.

\begin{definition}[Oper Darboux coordinates]We define the set of oper Darboux coordinates $(\mathbf{u},\mathbf{v})\coloneqq(u_1,\dots,u_{2d},v_1,\dots,v_{2d})$ by
  \beq\label{DeftdL12}\td{L}_{1,2}(\lambda)=\frac{\omega \underset{i=1}{\overset{2d}{\prod}}(\lambda-u_i)}{\lambda} \qquad \text{ and } \qquad  
v_i=\td{L}_{1,1}(u_i) \,\,, \qquad  \forall \, i\in \llbracket 1, 2d\rrbracket.
\eeq  
\end{definition}

These oper Darboux coordinates define $2d$ points on the spectral curve $\det( yI_2-\td{L}(\lambda))=0$, since $\det(v_i I_2-\td{L}(u_i))=0$ for all $i\in \llbracket 1,2d\rrbracket$. Moreover, they are canonical Darboux coordinates and the expression for the Hamiltonians and Lax matrices were given in \cite{marchal2024hamiltonianrepresentationisomonodromicdeformations}, we will briefly recall them below. The coefficient $\omega$ is an arbitrary parameter, it is the last remaining degree of freedom for the $\text{SL}_2(\mathbb{C})$ action to select a representative of $F_d$ and could be fixed by this action. It corresponds to the action of the centralizer that allows to fix one off-diagonal entry of the sub-leading coefficient without changing the leading order (conjugation by matrices of the form $\diag(1,\varphi)$). As we will see below, the symmetry imposed on the system determines a specific choice of $\omega$. In this setting, the entry $\td{L}_{1,1}(\lambda)$ is given by:
\beq \label{DefCheckL11}\td{L}_{1,1}(\lambda)=\frac{-O_{2d}(\lambda)- (t_{\infty,2d+1}\lambda+t_{\infty,2d})\underset{i=1}{\overset{2d}{\prod}}(\lambda-u_i)}{\lambda} \,, \,\text{  with }\,\, \,
O_{2d}(\lambda)\coloneqq-\sum_{i=1}^{2d} v_iu_i\prod_{j\neq i}\frac{\lambda-u_j}{u_i-u_j}
\eeq
a Lagrange interpolating term such that $O_{2d}(u_i)=-u_iv_i$ for all $i\in \llbracket 1, 2d\rrbracket$. As suggested by their name, the oper Darboux coordinates are particularly convenient for expressing the Lax matrices in the oper gauge. The oper gauge is reached through the transformation:
\beq \td{\Psi}_{\text{oper}}(\lambda) =\begin{pmatrix}1 &0\\
\td{L}_{1,1}(\lambda)& \td{L}_{1,2}(\lambda)\end{pmatrix} \td{\Psi}(\lambda),\eeq
for which the corresponding Lax matrix $\td{L}_{\text{oper}}(\lambda)$ is companion-like:
\begin{align} \left[\td{L}_{\text{oper}}(\lambda)\right]_{1,1}=&0, \qquad \left[\td{L}_{\text{oper}}(\lambda)\right]_{1,2}=1, \qquad 
  \left[\td{L}_{\text{oper}}(\lambda)\right]_{2,2}=\sum_{i=1}^{2d} \frac{1}{\lambda-u_i} -\frac{1}{\lambda},\cr
\left[\td{L}_{\text{oper}}(\lambda)\right]_{2,1}=&-t_{\infty,2d+1}\lambda^{2d-1} -\td{P}_2(\lambda) +\sum_{j=0}^{2d-2} \td{H}_{\infty, j}\lambda^{j}+\frac{\td{H}_{0,1}}{\lambda} -\sum_{j=1}^{2d}\frac{v_j}{\lambda-u_j},
\end{align}
where $\td{P}_2$ is a rational function of $\lambda$ given by the irregular times and monodromies:
\bea \td{P}_2(\lambda)&\coloneqq& -\sum_{k=0}^{2d+1} \left(\sum_{j=0}^k t_{\infty,2d+1-j} t_{\infty,2d+1-(k-j)}\right)\lambda^{4d-k} -\frac{(t_{0,0})^2}{\lambda^2}\cr
&=& -\underset{j= 2 d-1}{\overset{4d}{\sum}}\left(\underset{m=0}{\overset{4d-j}{\sum}} t_{\infty,2 d +1-m}t_{\infty,j+m-2d+1}\right) \lambda^{j}-\frac{(t_{0,0})^2}{\lambda^2},
\eea
and the coefficients $\td{\mathbf{H}}_{\infty}\coloneqq\left(\td{H}_{\infty,0},\dots, \td{H}_{\infty,{2d-2}}\right)$ and $\td{H}_{0,1}$ are the coefficients that will be determined once the compatibility equations are solved. The main advantage of the oper gauge is that one has simpler compatibility equations 
\begin{align}
    \partial_\lambda \td{A}_{ \text{oper}, \boldsymbol{\alpha}}(\lambda) - \partial_t \td{L}_{\text{oper}}(\lambda) + \left[ \td{L}_{\text{oper}}(\lambda) , \td{A}_{\text{oper}, \boldsymbol{\alpha} }(\lambda) \right] = 0.
\end{align}
Solving this compatibility equation provides the Hamiltonian evolutions of the oper Darboux coordinates. This result, derived in \cite{marchal2024hamiltonianrepresentationisomonodromicdeformations}, is recalled in the following proposition.

\begin{proposition}[Oper Hamiltonian evolutions \cite{marchal2024hamiltonianrepresentationisomonodromicdeformations}]\label{PropHamOper}
Define the matrix $M_\infty(\mathbf{t}) \in \mathcal{M}_{2d}(\mathbb{C})$ and the vector $\boldsymbol{\nu}^{(\boldsymbol{\alpha})}_{\infty}\coloneqq \left(\nu^{(\boldsymbol{\alpha})}_{\infty,0},\dots, \nu^{(\boldsymbol{\alpha})}_{\infty,2d-1}\right) \in \mathcal{M}_{2d}(\mathbb{C})$ by 
\beq\label{MatrixMInfty} M_\infty(\mathbf{t})\coloneqq\begin{pmatrix}t_{\infty,2d+1}&0&\dots &\dots &0\\
\vdots &\ddots &\ddots  & &\vdots\\
\vdots &&\ddots&0&0\\
t_{\infty,3} &\dots&& t_{\infty,2d+1}&0\\
t_{\infty,2}& \dots & & t_{\infty,2d}& t_{\infty,2d+1}
 \end{pmatrix} \text{ and }\begin{pmatrix}
\nu^{(\boldsymbol{\alpha})}_{\infty,0}\\ \vdots \\ \nu^{(\boldsymbol{\alpha})}_{\infty,2d-1}\end{pmatrix}\coloneqq (M_\infty(\mathbf{t}))^{-1}\begin{pmatrix} 
\frac{\alpha_{\infty,2d}}{2d}\\ \vdots \\ \frac{\alpha_{\infty,1}}{1}\end{pmatrix}.
\eeq
Then, the Hamiltonians are given by
\small{\beq \mathrm{Ham}_{\boldsymbol{\alpha}}(\mathbf{u},\mathbf{v};\mathbf{t},t_{\infty,0},t_{0,0})=\sum_{k=0}^{2d-2} \nu_{\infty,k+1}^{\boldsymbol{(\alpha)}}\td{H}_{\infty,k}(\mathbf{u},\mathbf{v},\mathbf{t},t_{\infty,0},t_{0,0}) +\nu^{(\boldsymbol{\alpha})}_{\infty,0}\left(\td{H}_{0,1}(\mathbf{u},\mathbf{v},\mathbf{t},t_{\infty,0},t_{0,0})- \sum_{j=0}^{2d}v_j\right),
\eeq}
\normalsize{which} is equivalent to
\beq \label{NewHamReduced}\begin{pmatrix}\mathrm{Ham}_{{t_{\infty,1}}}(\mathbf{u},\mathbf{v};\mathbf{t},t_{\infty,0},t_{0,0})\\ \vdots \\ (2d-1)\mathrm{Ham}_{{t_{\infty,2d-1}}}(\mathbf{u},\mathbf{v};\mathbf{t},t_{\infty,0},t_{0,0})\\
(2d)\mathrm{Ham}_{{t_{\infty,2d}}}(\mathbf{u},\mathbf{v};\mathbf{t},t_{\infty,0},t_{0,0})
\end{pmatrix}=\left(M_{\infty}(\mathbf{t})\right)^{-1}\begin{pmatrix}\td{H}_{\infty,2d-2}(\mathbf{u},\mathbf{v},\mathbf{t},t_{\infty,0},t_{0,0})\\ \vdots\\ \td{H}_{\infty,0}(\mathbf{u},\mathbf{v},\mathbf{t},t_{\infty,0},t_{0,0}) \\ \td{H}_{0,1}(\mathbf{u},\mathbf{v},\mathbf{t},t_{\infty,0},t_{0,0})-\underset{j=1}{\overset{2d}{\sum}}v_j 
\end{pmatrix}. \nonumber
 \eeq
Finally, the coefficients $\td{\mathbf{H}}_{\infty}$ and $\td{H}_{0,1}$ are given in terms of the oper Darboux coordinates by
\beq\label{HCoeffOperDarbouxCoord} \begin{pmatrix} \frac{1}{u_1}&1&u_1 &\dots &\dots &u_1^{2d-2}\\
\frac{1}{u_2}&1& u_2&\dots &\dots& u_{2}^{2d-2} \\
\vdots & & & & \vdots&\vdots\\
\vdots & & & & \vdots&\vdots\\
\frac{1}{u_{2d}}&1& u_{2d} &\dots & \dots& u_{2d}^{2d-2}&  \end{pmatrix}\begin{pmatrix}\td{H}_{0,1}\\ \td{H}_{\infty, 0}\\ \vdots\\ \td{H}_{\infty,2d-2}  \end{pmatrix}=\begin{pmatrix} v_1^2 + \frac{v_1}{u_1}+\td{P}_2(u_1)+\overset{2d}{\underset{i\neq 1}{\sum}}\frac{v_i-v_1}{u_1-u_i}+ t_{\infty,2d+1}u_1^{2d-1}\\
\vdots\\ 
\vdots\\
 v_{2d}^2 + \frac{v_{2d}}{u_{2d}}+\td{P}_2(u_{2d})+\overset{2d}{\underset{i\neq 2d}{\sum}}\frac{v_i-v_{2d}}{u_{2d}-u_i}+ t_{\infty,2d+1}u_{2d}^{2d-1}
\end{pmatrix}.
\eeq
For completeness, the expression of the auxiliary matrix $\td{A}_{\mathrm{oper}, \boldsymbol{\alpha}}(\lambda)$ is given in Appendix \ref{AppendixAuxiliaryMatrixOperGauge}.
\end{proposition}

\sloppy{
We now turn our attention to another set of coordinates, particularly convenient for having the Lax matrix $\td{L}(\lambda)$ expressed in its original geometric gauge giving the name:  “geometric Darboux coordinates" denoted by $\left(\mathbf{U}_\infty, U_{0,1}, \mathbf{V}_{\infty},V_{0,1}\right)\coloneqq\left(U_{\infty, 0},\dots, U_{\infty,2d-2}, U_{0,1}, V_{\infty, 0},\dots, V_{\infty,2d-2}, V_{0,1}\right)$.

\begin{definition}[Geometric Darboux coordinates]\label{DefGeoDarboux}The geometric Darboux coordinates $\left(\mathbf{U}_\infty, U_{0,1}, \mathbf{V}_{\infty},V_{0,1}\right)\coloneqq\left(U_{\infty, 0},\dots, U_{\infty,2d-2}, U_{0,1}, V_{\infty, 0},\dots, V_{\infty,2d-2}, V_{0,1}\right)$ are defined by
  \bea\label{DefNewcoor} \frac{\omega\underset{j=1}{\overset{2d}{\prod}}(\lambda-u_j)}{\lambda}&=& \omega\left(\frac{U_{0,1}}{\lambda}+\sum_{k=0}^{2d-2} U_{\infty,k}\lambda^k +\lambda^{2d-1}\right)\\
v_i&=&\sum_{k=0}^{2d-2} V_{\infty,k}\frac{\partial U_{\infty,k}(u_1,\dots,u_{2d})}{\partial u_i}+ V_{0,1}\frac{\partial U_{0,1}(u_1,\dots,u_{2d})}{\partial u_i} \,\,,\,\, \forall \, i\in \llbracket 1,2d\rrbracket. \nonumber
\eea  
\end{definition}

The geometric Darboux coordinates are also canonical Darboux coordinates with respect to the Poisson structure because they are related to the oper Darboux coordinates by a time-independent symplectic (i.e., preserving the symplectic $2-$form $\Omega$) change of coordinates}. The fact that this change of coordinates is symplectic is a consequence of Lemma $6.3$ of \cite{MarchalP1Hierarchy}. We point out that the geometric Darboux coordinates defined in this paper differ from those of \cite{MarchalAlameddineIsospectralIsomono2023} by a factor $\omega$ for $(\mathbf{U}_{\infty},U_{0,1})$ and $\omega^{-1}$ for $(\mathbf{V}_{\infty},V_{0,1})$. The present choice is motivated by the fact that we want Darboux coordinates to be independent of the choice of $\omega$ so that quantities $(\td{\mathbf{H}}_\infty, \td{H}_{0,1})$ hold no dependence on $\omega$.

Expressions for the Lax matrices and Hamiltonians using the geometric Darboux coordinates can be obtained from adaptations of formulas provided in \cite{MarchalAlameddineIsospectralIsomono2023} taking into account that the normalizations are not exactly the same (See \autoref{NormalizationMobius}).

\begin{proposition}[Lax matrix in terms of geometric Darboux coordinates]\label{GeoLaxMatrices}The Lax matrix $\td{L}(\lambda)$ is parameterized in terms of the geometric coordinates as follows:
\bea\label{L11OldPaper} \td{L}_{1,2}(\lambda)&=&\omega\left(\frac{U_{0,1}}{\lambda}+\sum_{k=0}^{2 d-2} U_{\infty,k}\lambda^k+\lambda^{2d-1}\right),\cr
\td{L}_{1,1}(\lambda)&=&-\sum_{k=0}^{2 d-2} V_{\infty,2 d-2-k}\lambda^k-\sum_{k=0}^{2 d-3}\sum_{m=0}^{2d-3-k}V_{\infty,m}U_{\infty,k+1+m}\lambda^k\cr
&&+ \frac{V_{0,1}U_{0,1}}{\lambda}-\left(t_{\infty,2 d+1}\lambda+t_{\infty, 2d}-t_{\infty,2d+1}U_{\infty,2 d-2}\right)\frac{\td{L}_{1,2}(\lambda)}{\omega},\cr
\td{L}_{2,2}(\lambda)&=&-\td{L}_{1,1}(\lambda),\cr
\td{L}_{2,1}(\lambda)&=& \frac{1}{\omega}\left[\frac{\underset{j= 2 d-1}{\overset{4d}{\sum}}\left(\underset{m=0}{\overset{4d-j}{\sum}} t_{\infty,2 d +1-m}t_{\infty,j+m-2d+1}\right) \lambda^{j} -\td{L}_{1,1}(\lambda)^2}{\frac{\td{L}_{1,2}(\lambda)}{\omega}}\right]_{\infty,+}\cr&&
+\frac{1}{\omega}\frac{\frac{(t_{0,0})^2}{U_{0,1}}-U_{0,1}\left(V_{0,1}-t_{\infty,2d}+t_{\infty,2d+1}U_{\infty,2d-2}\right)^2 }{\lambda}.
\eea
where $\left[f(\lambda)\right]_{\infty,+}$ stands for the polynomial part (including the constant term) of a meromorphic function $f$.
\end{proposition}

Using \autoref{PropHamOper}, one only needs to express the quantities $(\td{\mathbf{H}}_{\infty},\td{H}_{0,1})$ in terms of the geometric Darboux coordinates to obtain their evolutions. This is given by the following proposition.

\begin{proposition}[Hamiltonians for the geometric Darboux coordinates]\label{PropHamGeometricDarboux} We have
\small{\beq \label{NewHamtdReduced2}\begin{pmatrix}\mathrm{Ham}_{{t_{\infty,1}}}(\mathbf{U}_{\infty}, U_{0,1},\mathbf{V}_{\infty},V_{0,1};\mathbf{t},t_{\infty,0},t_{0,0})\\ \vdots \\ (2d-1)\mathrm{Ham}_{{t_{\infty,2d-1}}}(\mathbf{U}_{\infty},U_{0,1},\mathbf{V}_{\infty},V_{0,1};\mathbf{t},t_{\infty,0},t_{0,0})\\
(2d)\mathrm{Ham}_{{t_{\infty,2d}}}(\mathbf{U}_{\infty},U_{0,1},\mathbf{V}_{\infty},V_{0,1};\mathbf{t},t_{\infty,0},t_{0,0})
\end{pmatrix}=\left(M_{\infty}(\mathbf{t})\right)^{-1}\begin{pmatrix}\td{H}_{\infty,2d-2}(\mathbf{U}_{\infty},U_{0,1},\mathbf{V}_{\infty},V_{0,1},\mathbf{t},t_{\infty,0},t_{0,0})\\ \vdots\\ \td{H}_{\infty,0}(\mathbf{U}_{\infty},U_{0,1},\mathbf{V}_{\infty},V_{0,1},\mathbf{t},t_{\infty,0},t_{0,0}) \\ \td{H}_{0,1}(\mathbf{U}_{\infty},U_{0,1},\mathbf{V}_{\infty},V_{0,1},t_{\infty,0},t_{0,0})+\td{\delta}_{0,1} 
\end{pmatrix},\eeq}
\normalsize{where} $\forall\, j\in \llbracket 0, 2d-2\rrbracket$: 
\begin{align}  \label{DeftdHs}
\td{H}_{\infty,j}(\mathbf{U}_{\infty},U_{0,1},\mathbf{V}_{\infty},V_{0,1},\mathbf{t},t_{\infty,0},t_{0,0})= &-\Res_{\lambda\to \infty}\lambda^{-j-1} \left[ (\td{L}_{1,1})^2+\td{L}_{2,1}\td{L}_{1,2} +\td{L}_{1,2}\partial_\lambda\left(\frac{ \td{L}_{1,1}}{\td{L}_{1,2}}\right)\right], \nonumber \\
\td{H}_{0,1} (\mathbf{U}_{\infty},U_{0,1},\mathbf{V}_{\infty},V_{0,1},\mathbf{t},t_{\infty,0},t_{0,0})=&\Res_{\lambda\to 0}(\td{L}_{1,1})^2+\td{L}_{2,1}\td{L}_{1,2} +\td{L}_{1,2} \partial_\lambda\left(\frac{ \td{L}_{1,1}}{\td{L}_{1,2}}\right).
\end{align}
Finally, the additional term is given by
\beq \td{\delta}_{0,1}=-\sum_{j=1}^{2d} v_j=2d\,V_{\infty,2d-2}+\sum_{k=0}^{2d-3} (k+2)U_{\infty,k+1}V_{\infty,k}+ U_{\infty,0}V_{0,1},\eeq
and the entries of the matrix $\td{L}(\lambda)$ are given by \autoref{GeoLaxMatrices}. Note in particular that the Hamiltonians are  independent of the parameter $\omega$.
\end{proposition}

The proof of the above results follows from adaptations of results of \cite{marchal2024hamiltonianrepresentationisomonodromicdeformations} and \cite{MarchalAlameddineIsospectralIsomono2023}. For completeness, we also give the expression of the auxiliary matrix $\td{\mathcal{A}}_{\boldsymbol{\alpha}}(\lambda)$ in Appendix \ref{AppendixAuxiliaryMatrixGeometricDarboux}. 
\medskip

Note that the geometric Darboux coordinates are not adapted for the upcoming symmetry (this point was observed in \cite{alameddine2026symmetryreductionpainleveiv}).

\subsection{Symmetric oper Darboux coordinates and choice of representative in $F_d$}
The previous geometric Darboux coordinates are not well-suited for the symmetry, as proved in \cite{alameddine2026symmetryreductionpainleveiv}. Indeed, the symmetry requires taking a very specific and complicated value of $\omega$, i.e., to use the remaining part of the $\text{SL}_2(\mathbb{C})$ action to select an appropriate representative for the symmetry. Moreover, the symmetry does not eliminate half of the Darboux coordinates but provides complicated relations between them. However, the condition on $\omega$ becomes clearer when considering $\check{L}(\lambda)$ as we explain in the following definition.

\begin{definition}[Choice of representative in $F_d$ compatible with the symmetry]\label{ChoiceRepresentative}We use the remaining part of the $\text{SL}_2(\mathbb{C})$ action (i.e., conjugation of $\td{L}(\lambda)$ by $\diag(1,\alpha)$) to select the unique representative in $F_d$ such that $\check{L}_{1,2}(\lambda)$ is a polynomial of degree $2d$ with leading coefficient $-t_{\infty,2d+1}$:
\beq \check{L}_{1,2}(\lambda)=-t_{\infty,2d+1}\prod_{i=1}^{2d} (\lambda-q_i).\eeq
In other words, we use the remaining part of the $\text{SL}_{2}(\mathbb{C})$ action to select the representative in $F_d$ such that $\check{L}_{1,2}$ has no (simple) pole at zero. This choice will be made for the rest of the paper when considering $\check{L}(\lambda)$ and the associated matrices $\check{A}_{\boldsymbol{\alpha}}(\lambda)$ and $\check{\Psi}(\lambda)$. 
\end{definition}

\begin{remark}\autoref{ChoiceRepresentative} is equivalent to taking
\beq \omega=t_{\infty,2d}-V_{0,1}-t_{\infty,2d+1}U_{\infty,2d-2}-\frac{t_{0,0}}{U_{0,1}},\eeq
which is the specific value of $\omega$ required for the symmetry as observed in \cite{alameddine2026symmetryreductionpainleveiv}. Moreover, it provides a normalization of $\check{L}_{1,1}(\lambda)$ at $\lambda=0$ given by:
\beq \label{NormcheckL11}
\check{L}_{1,1}(\lambda)\overset{\lambda\to 0}{=}\frac{t_{0,0}}{\lambda}+O(1).
\eeq 
Moreover, we get that
\beq \label{ConsequencetdL21pole0} \td{L}_{2,1}(\lambda)\overset{\lambda\to 0}{=}-\frac{t_{0,0}+U_{0,1}\left(t_{\infty,2d}-t_{\infty,2d+1}U_{\infty,2d-2}-V_{0,1}\right) }{\lambda}+O(1)\eeq
\end{remark}

Using this representative in $F_d$, we define the corresponding oper gauge and the associated oper coordinates that we shall refer to as the “symmetric oper Darboux coordinates".
\begin{definition}[Symmetric oper Darboux coordinates]We define $(\mathbf{q},\mathbf{p})\coloneqq(q_1,\dots,q_{2d},p_1,\dots,p_{2d})$ as
    \beq\label{DefCheckL12}\check{L}_{1,2}(\lambda)=-t_{\infty,2d+1}\prod_{i=1}^{2d} (\lambda-q_i) \qquad \text{ and } \qquad  
p_i=\check{L}_{1,1}(q_i) \,\,, \qquad  \forall \, i\in \llbracket 1, 2d\rrbracket.
\eeq
\end{definition}

The terminology “symmetric" indicates that $(q_1,\dots,q_{2d})$ are by \eqref{checkLtdL} the zeros of $\td{L}_{1,2}(\lambda)+\td{L}_{2,1}(\lambda)$, which is symmetric with respect to the exchange $1\leftrightarrow 2$.

\medskip

We stress again that these symmetric oper Darboux coordinates depend on the choice of representative in $F_d$. Indeed, if one does not select $\check{L}$ as defined in \autoref{ChoiceRepresentative}, then $\check{L}_{1,2}(\lambda)$ would have $2d+1$ zeros (and a pole at $\lambda=0$) so that the present definition would not make sense. This phenomenon is absent in the tilde gauge, because the condition in \autoref{ChoiceRepresentative} does not affect the definition of the oper Darboux coordinates for $\td{L}(\lambda)$. Indeed, $\omega$ is the normalizing factor of $\td{L}_{1,2}(\lambda)$ so it affects neither the zeros $(u_i)_{1\leq i\leq 2d}$ of this entry, nor the values of $v_i=\td{L}_{1,1}(u_i)$. However, the symmetry condition for $\td{L}(\lambda)$ fixes the value of $\omega$ to a non-trivial value and this value is equivalent to the absence of a pole at zero of $\check{L}_{1,2}(\lambda)$. Thus, in both cases, the symmetry condition completely fixes the choice of a representative in the orbit. 

\medskip

Note that \autoref{ChoiceRepresentative} implies that we have
\beq \check{L}_{1,1}(\lambda) = \frac{t_{0,0}}{\lambda} +\sum_{i=1}^{2d}\left(p_i-\frac{t_{0,0}}{q_i}\right) \prod_{j\neq i}\frac{\lambda-q_j}{q_i-q_j}.\eeq
Indeed, it obviously satisfies that $\check{L}_{1,1}(\lambda)=\frac{t_{0,0}}{\lambda}+O(1)$ that follows from \eqref{NormcheckL11} and $\check{L}_{1,1}(q_i)=p_i$, for all $i\in \llbracket 1,2d\rrbracket$, and these identities uniquely determine $\check{L}_{1,1}(\lambda)$. As suggested by their name, the symmetric oper Darboux coordinates are particularly convenient for expressing the Lax matrices in the symmetric oper gauge. The symmetric oper gauge is reached through the transformation:
\beq \check{\Psi}_{\text{oper}}(\lambda) \coloneqq\begin{pmatrix}1 &0\\
\check{L}_{1,1}(\lambda)& \check{L}_{1,2}(\lambda)\end{pmatrix} \check{\Psi}(\lambda),\eeq
for which the corresponding Lax matrix $\check{L}_{\text{oper}}(\lambda)$ is companion-like.

\begin{proposition}\label{PropLopersym}The wave matrix $\check{\Psi}(\lambda)$ has the following asymptotics:
 \bea
 \check{\Psi}_{1,1}(\lambda)&\overset{\lambda\to \infty}{=}&\exp\left(-\sum_{k=1}^{2d+1} \frac{t_{\infty,k}}{k}\lambda^k -t_{\infty,0}\ln \lambda+ A_{\infty,0} +O\left(\lambda^{-1}\right)\right),\cr
 \check{\Psi}_{1,2}(\lambda)&\overset{\lambda\to \infty}{=}&\exp\left(\sum_{k=1}^{2d+1} \frac{t_{\infty,k}}{k}\lambda^k +t_{\infty,0}\ln \lambda+ B_{\infty,0} +O\left(\lambda^{-1}\right)\right),\cr
 \check{\Psi}_{1,1}(\lambda)&\overset{\lambda\to 0}{=}&\exp\left(t_{0,0}\ln(\lambda)+ A_{0,0}+O\left(\lambda\right)\right),\cr
\check{\Psi}_{1,2}(\lambda)&\overset{\lambda\to 0}{=}&\exp\left((1-t_{0,0})\ln(\lambda)+ B_{0,0}+O\left(\lambda\right)\right).
\eea
Consequently, the Lax matrix $\check{L}_{\mathrm{oper}}(\lambda)$ is given by
\begin{align} \left[\check{L}_{\mathrm{oper}}(\lambda)\right]_{1,1}=&0, \qquad \left[\check{L}_{\mathrm{oper}}(\lambda)\right]_{1,2}=1, \qquad 
  \left[\check{L}_{\mathrm{oper}}(\lambda)\right]_{2,2}=\sum_{i=1}^{2d} \frac{1}{\lambda-q_i},\cr
\left[\check{L}_{\mathrm{oper}}(\lambda)\right]_{2,1}=& -\check{P}_2(\lambda) +\sum_{j=0}^{2d-2} \check{H}_{\infty, j}\lambda^{j}+\frac{\check{H}_{0,1}}{\lambda} -\sum_{j=1}^{2d}\frac{p_j}{\lambda-q_j},
\end{align}
where $\check{P}_2$ is a rational function of $\lambda$ given by the irregular times and monodromies:
\bea \check{P}_2(\lambda)&\coloneqq& -\sum_{k=0}^{2d+1} \left(\sum_{j=0}^k t_{\infty,2d+1-j} t_{\infty,2d+1-(k-j)}\right)\lambda^{4d-k} +\frac{t_{0,0}(1-t_{0,0})}{\lambda^2}\cr
&=& -\underset{j= 2 d-1}{\overset{4d}{\sum}}\left(\underset{m=0}{\overset{4d-j}{\sum}} t_{\infty,2 d +1-m}t_{\infty,j+m-2d+1}\right) \lambda^{j}+\frac{t_{0,0}(1-t_{0,0})}{\lambda^2}.
\eea
\end{proposition}
\begin{proof}The asymptotic expansions of $\check{\Psi}(\lambda)$ are proved  in \autoref{AppendixAsymptPsi}. The formula for $\check{L}_{\text{oper}}(\lambda)$ then follows from the fact that
\bea \label{LInTermsOfTdL} 
\left[\check{L}_{\text{oper}}(\lambda)\right]_{1,1}&=&0, \qquad \left[\check{L}_{\text{oper}}(\lambda)\right]_{1,2}=1, \cr
\left[\check{L}_{\text{oper}}(\lambda)\right]_{2,1}&=&-\det \check{L}(\lambda)+\partial_\lambda\check{L}_{1,1}(\lambda)-\check{L}_{1,1}(\lambda)\frac{\partial_\lambda\check{L}_{1,2}(\lambda)}{\check{L}_{1,2}(\lambda)},\cr
\left[\check{L}_{\text{oper}}(\lambda)\right]_{2,2}&=&\frac{\partial_\lambda\check{L}_{1,2}(\lambda)}{\check{L}_{1,2}(\lambda)}
\eea  
and the fact that $\check{\Psi}_{1,j}(\lambda)=\left[\check{\Psi}_{\text{oper}}(\lambda)\right]_{1,j}$ for all $j \in \{1,2\}$. We refer to \autoref{AppendixAsymptPsi} for details.
\end{proof}
We will now explicitly construct the form of the auxiliary matrix in order to solve the compatibility condition and obtain the Hamiltonian structure. 

\subsection{The auxiliary matrix $\check{A}_{\text{oper},\boldsymbol{\alpha}}(\lambda)$}
Using compatibility equations, one may easily obtain two of the entries of $\check{A}_{\text{oper},\boldsymbol{\alpha}}(\lambda)$. Indeed, since $\check{L}_{\text{oper}}(\lambda)$ is a companion-like matrix, compatibility equations imply that
\begin{align}
    \label{TrivialEntriesA}\left[\check{A}_{\text{oper},\boldsymbol{\alpha}}(\lambda)\right]_{2,1}=& \partial_{\lambda} \left[\check{A}_{\text{oper},\boldsymbol{\alpha}}(\lambda)\right]_{1,1}+\left[\check{A}_{\text{oper},\boldsymbol{\alpha}}(\lambda)\right]_{1,2}\left[\check{L}_{\text{oper}}(\lambda)\right]_{2,1}, \\
\left[\check{A}_{\text{oper},\boldsymbol{\alpha}}(\lambda)\right]_{2,2}=& \partial_{\lambda} \left[\check{A}_{\text{oper},\boldsymbol{\alpha}}(\lambda)\right]_{1,2}+\left[\check{A}_{\text{oper},\boldsymbol{\alpha}}(\lambda)\right]_{1,1}+\left[\check{A}_{\text{oper},\boldsymbol{\alpha}}(\lambda)\right]_{1,2}\left[\check{L}_{\text{oper}}(\lambda)\right]_{2,2}, \nonumber
\end{align}
so that only the first line of $\check{A}_{\text{oper},\boldsymbol{\alpha}}(\lambda)$ remains unknown at this stage. The other two entries of the compatibility equation lead to
\begin{align}
     \label{Compat}\mathcal{L}_{\boldsymbol{\alpha}}\left[\left[\check{L}_{\text{oper}}(\lambda)\right]_{2,1}\right]=& \frac{\partial^2 \left[\check{A}_{\text{oper},\boldsymbol{\alpha}}(\lambda)\right]_{1,1}}{\partial \lambda^2} +2 \left[\check{L}_{\text{oper}}(\lambda)\right]_{2,1}\, \partial_\lambda \left[\check{A}_{\text{oper},\boldsymbol{\alpha}}(\lambda)\right]_{1,2}\cr&+ \left[\check{A}_{\text{oper},\boldsymbol{\alpha}}(\lambda)\right]_{1,2} \, \partial_\lambda \left[\check{L}_{\text{oper}}(\lambda)\right]_{2,1}- \left[\check{L}_{\text{oper}}(\lambda)\right]_{2,2}\, \partial_{\lambda} \left[\check{A}_{\text{oper},\boldsymbol{\alpha}}(\lambda)\right]_{1,1},\cr
\mathcal{L}_{\boldsymbol{\alpha}}\left[\left[\check{L}_{\text{oper}}(\lambda)\right]_{2,2}\right]=&\frac{\partial^2 \left[\check{A}_{\text{oper},\boldsymbol{\alpha}}(\lambda)\right]_{1,2}}{\partial \lambda^2} +2 \partial_\lambda \left[\check{A}_{\text{oper},\boldsymbol{\alpha}}(\lambda)\right]_{1,1} \\
&+ \left[\check{L}_{\text{oper}}(\lambda)\right]_{2,2}\,\partial_\lambda \left[\check{A}_{\text{oper},\boldsymbol{\alpha}}(\lambda)\right]_{1,2}+ \left[\check{A}_{\text{oper},\boldsymbol{\alpha}}(\lambda)\right]_{1,2}\,\partial_{\lambda}\left[\check{L}_{\text{oper}}(\lambda)\right]_{2,2}. \nonumber
\end{align}
Then, let us obtain the asymptotic expansions at the poles of the first line of $\check{A}_{\mathrm{oper},\boldsymbol{\alpha}}(\lambda)$.
\begin{proposition}\label{PropAsymptoticExpansionA12} The asymptotic expansions of entry $\left[\check{A}_{\mathrm{oper},\boldsymbol{\alpha}}(\lambda)\right]_{1,2}$ are given by
\begin{align}
    \left[\check{A}_{\mathrm{oper},\boldsymbol{\alpha}}(\lambda)\right]_{1,2}\,\,\overset{\lambda\to \infty}{=} \,\,\,\sum_{i=0}^{2d-1} \frac{\nu^{(\boldsymbol{\alpha})}_{\infty,i}}{\lambda^i} +O\left(\lambda^{-2d}\right), \qquad
\left[\check{A}_{\mathrm{oper},\boldsymbol{\alpha}}(\lambda)\right]_{1,2} \, \,\overset{\lambda\to 0}{=} \,\,O(\lambda).
\end{align}
Coefficients $\left(\nu^{(\boldsymbol{\alpha})}_{\infty,k}\right)_{0\leq k\leq 2d-1}$ are determined by 
\beq \label{RelationNuAlphaInfty} M_\infty(\mathbf{t})\begin{pmatrix} \nu^{(\boldsymbol{\alpha})}_{\infty,0}\\ \vdots \\ \nu^{(\boldsymbol{\alpha})}_{\infty,2d-1}\end{pmatrix}=\begin{pmatrix}\frac{\alpha_{\infty,2d}}{2d}\\ \vdots \\ \frac{\alpha_{\infty,1}}{1}\end{pmatrix},\eeq
where the matrix $M_\infty(\mathbf{t})$ is given in \autoref{PropHamOper}. 

\end{proposition}

\begin{proof}The proof is given in Appendix \ref{AppendixExpansionA}.
\end{proof}

\begin{remark}Note that the coefficients $\left(\nu^{(\boldsymbol{\alpha})}_{\infty,k}\right)_{0\leq k\leq 2d-1}$ are the same as in \autoref{Sec22}. These quantities are thus independent of the gauge choice or the choice of Darboux coordinates.
\end{remark}

The previous proposition can be used to determine the general form of the entry $\left[\check{A}_{\text{oper},\boldsymbol{\alpha}}(\lambda)\right]_{1,2}$.

\begin{proposition}\label{PropA12Form} Entry $\left[\check{A}_{\mathrm{oper},\boldsymbol{\alpha}}(\lambda)\right]_{1,2}$ is given by
\beq \label{ExpressionA12} \left[\check{A}_{\mathrm{oper},\boldsymbol{\alpha}}(\lambda)\right]_{1,2}=\nu^{(\boldsymbol{\alpha})}_{\infty,0} + \sum_{j=1}^{2d} \frac{\check{\mu}^{(\boldsymbol{\alpha})}_j}{\lambda-q_j}.\eeq
Coefficients $\left(\check{\mu}^{(\boldsymbol{\alpha})}_j\right)_{1\leq j\leq 2d}$ are determined by the $(2d)\times (2d)$ linear system
\beq \label{RelationNuMuMatrixForm} \begin{pmatrix}
\frac{1}{q_1}&\frac{1}{q_2}&\dots&\dots& \frac{1}{q_{2d}}\\
1&1 &\dots &\dots &1\\
q_1& q_2&\dots &\dots& q_{2d}\\
\vdots & & & & \vdots\\
\vdots & & & & \vdots\\
q_1^{2d-2}& q_2^{2d-2} &\dots & \dots& q_{2d}^{2d-2}
\end{pmatrix}\begin{pmatrix} \check{\mu}^{(\boldsymbol{\alpha})}_1\\ \vdots\\\vdots\\ \check{\mu}^{(\boldsymbol{\alpha})}_{2d}\end{pmatrix}\coloneqq \mathbf{V}(\mathbf{q}) \begin{pmatrix} \check{\mu}^{(\boldsymbol{\alpha})}_1\\ \vdots\\\vdots\\ \check{\mu}^{(\boldsymbol{\alpha})}_{2d}\end{pmatrix}=\boldsymbol{\nu}^{(\boldsymbol{\alpha})}\coloneqq\begin{pmatrix}\nu^{(\boldsymbol{\alpha})}_{\infty,0}\\ \nu^{(\boldsymbol{\alpha})}_{\infty,1}
\\\vdots \\\vdots \\ \nu^{(\boldsymbol{\alpha})}_{\infty,2d-1}\end{pmatrix}. \eeq
\end{proposition}
\begin{proof}The proof is given in Appendix \ref{ProofEntryA12}.
\end{proof}
Finally, one obtains the general form of entry $\left[\check{A}_{\text{oper},\boldsymbol{\alpha}}(\lambda)\right]_{1,1}$.

\begin{proposition}\label{Propcalpha}The entry $\left[\check{A}_{\mathrm{oper},\boldsymbol{\alpha}}(\lambda)\right]_{1,1}$ is given by
\begin{align}
     \label{ExpressionA11} \left[\check{A}_{\mathrm{oper},\boldsymbol{\alpha}}(\lambda)\right]_{1,1}=
     \sum_{j=1}^{2d}\frac{\check{\rho}^{(\boldsymbol{\alpha})}_j}{\lambda-q_j} \qquad \text{with} \qquad \forall\, j\in \llbracket 1,2d\rrbracket \,:\, \check{\rho}^{(\boldsymbol{\alpha})}_j=-p_j\check{\mu}^{(\boldsymbol{\alpha})}_j.
\end{align}
\end{proposition}
\begin{proof}The proof is given in Appendix \ref{AppendixA11}.
\end{proof}

\begin{remark}Note that the matrix $\left[\check{A}_{\text{oper},\boldsymbol{\alpha}}(\lambda)\right]$ has the same form as $\left[\td{A}_{\text{oper},\boldsymbol{\alpha}}(\lambda)\right]$ except that the Darboux coordinates $\mathbf{q}$ are used in $\mathbf{V}(\mathbf{q})$ instead of the Darboux coordinates $\mathbf{u}$. All other formulas and the general structure, including $\left(\nu_{\infty,k}^{(\boldsymbol{\alpha})}\right)_{0\leq k\leq 2d-1}$ are identical.
\end{remark}

\subsection{Compatibility equations and Hamiltonians for $(\mathbf{q},\mathbf{p})$}
The previous sections provide the general form of the Lax matrices through  \autoref{PropLopersym}, \autoref{PropAsymptoticExpansionA12}, \autoref{PropA12Form}, \autoref{Propcalpha} and \eqref{TrivialEntriesA}. We will see below that inserting this previous knowledge into the compatibility equations \eqref{Compat} provides the evolutions of the Darboux coordinates $(\mathbf{q},\mathbf{p})$. The first step is to look at order $(\lambda-q_j)^{-2}$ in $\mathcal{L}_{\boldsymbol{\alpha}}\left[\left[\check{L}_{\text{oper}}(\lambda)\right]_{2,2}\right]$. We obtain, for all $j\in \llbracket 1, 2d\rrbracket$:
\beq \label{Lqj}\mathcal{L}_{\boldsymbol{\alpha}}[q_j]=2\check{\mu}^{(\boldsymbol{\alpha})}_jp_j
- \nu^{(\boldsymbol{\alpha})}_{\infty,0} - \sum_{1\leq i\neq j\leq 2d}\frac{\check{\mu}^{(\boldsymbol{\alpha})}_j+\check{\mu}^{(\boldsymbol{\alpha})}_i}{q_j-q_i}.\eeq
Note in particular that because $\left[\check{L}_{\text{oper}}(\lambda)\right]_{2,2}$ is regular at $\lambda=0$, we do not have a term $\frac{\check{\mu}^{(\boldsymbol{\alpha})}_j}{q_j}$.

The second step is then to determine the coefficients $\left(\check{H}_{\infty,j}\right)_{0\leq j\leq 2d-2}$ and $\check{H}_{0,1}$ that remain unknown in $\left[\check{L}_{\text{oper}}(\lambda)\right]_{2,1}$. To achieve this, we look at order $(\lambda-q_j)^{-2}$ in $\mathcal{L}_{\boldsymbol{\alpha}}\left[\left[\check{L}_{\text{oper}}(\lambda)\right]_{2,1}\right]$ using \eqref{Compat}. We obtain

\bea - p_j \mathcal{L}_{\boldsymbol{\alpha}}[q_j]&=&-2 \check{\mu}^{(\boldsymbol{\alpha})}_j\left(
-\check{P}_2(q_j) +\sum_{k=0}^{2d-2}\check{H}_{\infty,k}q_j^k+ \frac{\check{H}_{0,1}}{q_j} -\sum_{1\leq i\neq j\leq 2d} \frac{p_i}{q_j-q_i}\right)\cr
&&+p_j\left(\nu^{(\boldsymbol{\alpha})}_{\infty,0}+\sum_{1\leq i\neq j\leq 2d} \frac{\check{\mu}^{(\boldsymbol{\alpha})}_i}{q_j-q_i}\right)- \check{\mu}^{(\boldsymbol{\alpha})}_j p_j\left(
\sum_{1\leq i\neq j\leq 2d} \frac{1}{q_j-q_i}\right).
\eea
Inserting \eqref{Lqj} provides:
\beq \label{DefCi}\forall j\in \llbracket 1,2d\rrbracket\,:\, \sum_{k=0}^{2d-2}\check{H}_{\infty,k}q_j^k+\frac{\check{H}_{0,1}}{q_j}=p_j^2
+\check{P}_2(q_j)+ \sum_{1\leq i\neq j\leq 2d}\frac{p_i-p_j}{q_j-q_i},
\eeq
where it is obvious that the r.h.s. is independent of the deformation vector $\boldsymbol{\alpha}$. The last relation can be rewritten into a matrix form.

\begin{lemma}\label{PropDefCi2} We have 
\beq \label{DefCi2}
\begin{pmatrix}\frac{1}{q_1}& 1&q_1& \dots &q_1^{2d-2}\\
\frac{1}{q_2}& 1&q_2& \dots &q_2^{2d-2}\\
\vdots& \vdots&\vdots&&\vdots\\
\frac{1}{q_{2d}}& 1&q_{2d}& \dots &q_{2d}^{2d-2}\end{pmatrix}
\begin{pmatrix}\check{H}_{0,1}\\ \check{H}_{\infty,0}\\ \vdots\\ \check{H}_{\infty,2d-2}\end{pmatrix}=\begin{pmatrix} p_1^2 +\check{P}_2(q_1)+\underset{1\leq i\neq 1\leq 2d}{\sum}\frac{p_i-p_1}{q_1-q_i}\\
p_2^2 +\check{P}_2(q_2)+\underset{1\leq i\neq 2\leq 2d}{\sum}\frac{p_i-p_2}{q_2-q_i}\\
\vdots\\
p_{2d}^2+\check{P}_2(q_{2d})+\underset{1\leq i\neq 2d\leq 2d}{\sum}\frac{p_i-p_{2d}}{q_{2d}-q_i}
\end{pmatrix}
\eeq
\end{lemma}

Finally, in order to obtain the evolution equation for $p_j$ we look at order $(\lambda-q_j)^{-1}$ of the entry $\mathcal{L}_{\boldsymbol{\alpha}}\left[\left[\check{L}_{\text{oper}}(\lambda)\right]_{2,1}\right]$. We get, for all $j\in \llbracket 1,2d\rrbracket$,
\beq \label{Lpj} \mathcal{L}_{\boldsymbol{\alpha}}[p_j]= \sum_{1\leq i\neq j\leq 2d}\frac{(\check{\mu}^{(\boldsymbol{\alpha})}_i+\check{\mu}^{(\boldsymbol{\alpha})}_j)(p_i-p_j)}{(q_j-q_i)^2} +\check{\mu}^{(\boldsymbol{\alpha})}_j\left(-\check{P}_2'(q_j)+\sum_{k=1}^{2d-2}k\check{H}_{\infty,k}q_j^{k-1}-\frac{\check{H}_{0,1}}{q_j^2}\right).
\eeq
Thus, we have obtained the general evolutions for $\left(p_j,q_j\right)_{1\leq j\leq 2d}$ through \eqref{Lqj} and \eqref{Lpj}. These evolutions are Hamiltonian because of the following theorem.

\begin{theorem}[Hamiltonians for $(\mathbf{q},\mathbf{p})$] \label{HamTheorem} The Hamiltonian for $(\mathbf{q},\mathbf{p})$ is given by
\beq \label{DefHamqp} \mathrm{Ham}^{(\boldsymbol{\alpha})}(\mathbf{q},\mathbf{p})\coloneqq\sum_{k=0}^{2d-2} \nu_{\infty,k+1}^{\boldsymbol{(\alpha)}}\check{H}_{\infty,k}+\nu_{\infty,0}^{\boldsymbol{(\alpha)}}\left(\check{H}_{0,1}-\sum_{j=1}^{2d} p_j\right)
\eeq
in the sense that 
\beq \forall\, j\in \llbracket 1,2d\rrbracket \,:\, \mathcal{L}_{\boldsymbol{\alpha}}[q_j]=\frac{\partial \mathrm{Ham}^{(\boldsymbol{\alpha})}(\mathbf{q},\mathbf{p})}{\partial p_j}\, \mathrm{ and }\, \mathcal{L}_{\boldsymbol{\alpha}}[p_j]=-\frac{\partial \mathrm{Ham}^{(\boldsymbol{\alpha})}(\mathbf{q},\mathbf{p})}{\partial q_j}.\eeq
Quantities involved in the Hamiltonians are defined by \autoref{PropAsymptoticExpansionA12} and \autoref{PropDefCi2}.
\end{theorem}

\begin{proof}The proof is given in \autoref{AppendixProofHamiltonian}.\end{proof}

\begin{remark}
The Hamiltonian may also be rewritten as
\beq
\text{Ham}^{(\boldsymbol{\alpha})}(\mathbf{q},\mathbf{p})=-\frac{1}{2}\displaystyle{\sum_{\substack{(i,j)\in \llbracket 1,2d\rrbracket^2 \\ i\neq j }}} \frac{(\check{\mu}^{(\boldsymbol{\alpha})}_i+\check{\mu}^{(\boldsymbol{\alpha})}_j)(p_i-p_j)}{q_i-q_j} - \nu^{(\boldsymbol{\alpha})}_{\infty,0}\sum_{j=1}^{2d}  p_j+\sum_{j=1}^{2d}\check{\mu}^{(\boldsymbol{\alpha})}_j\left(p_j^2 +\check{P}_2(q_j)\right).\eeq
\end{remark}

\begin{remark}\label{HamiltonianStrucutre}It is important to notice that the formulas for the Hamiltonians are almost identical in the two oper gauges. They are obtained as the same time-dependent linear combinations of the $\left(\td{H}_{\infty,k}\right)_{0\leq k\leq 2d-2} \cup \td{H}_{0,1}$ (resp. $\left(\check{H}_{\infty,k}\right)_{0\leq k\leq 2d-2} \cup \check{H}_{0,1}$) where these coefficients are almost identical up to the change of Darboux coordinates $(\mathbf{u},\mathbf{v})\leftrightarrow (\mathbf{p},\mathbf{q})$ in \autoref{PropHamOper} and \autoref{HamTheorem}. The only difference is that the r.h.s. of \eqref{HCoeffOperDarbouxCoord} giving $\left(\td{H}_{\infty,k}\right)_{0\leq k\leq 2d-2}$ and $\td{H}_{0,1}$ has a term in $\frac{v_i}{u_i}+t_{\infty,2d+1}u_i^{2d-1}$, which is missing in \autoref{PropDefCi2}. This is due to the normalization of $\td{L}(\lambda)$ that differs from $\check{L}(\lambda)$. More specifically, the fact that $\check{L}_{1,2}(\lambda)$ is chosen to have no pole at $\lambda=0$ (\autoref{ChoiceRepresentative}) removes the term $\frac{p_i}{q_i}$, while the fact that $\td{L}(\lambda)=\diag(-t_{\infty,2d+1},t_{\infty,2d+1})\lambda^{2d}+O(\lambda^{2d-1})$ adds the term $t_{\infty,2d+1}u_i^{2d-1}$.  
\end{remark}
We will now define a new set of coordinates analogous to the geometric Darboux coordinates but adapted to the symmetry through the normalization condition, this is the main construction of the next subsection. 

\subsection{Parameterization using symmetric geometric Darboux coordinates}
\sloppy{The symmetric geometric Darboux coordinates are denoted by $\left(\mathbf{Q}_\infty, \mathbf{P}_{\infty}\right)\coloneqq\left(Q_{\infty, 0},\dots, Q_{\infty,2d-1}, P_{\infty, 0},\dots, P_{\infty,2d-1}\right)$. They are also canonical Darboux coordinates and they are related to the symmetric oper Darboux coordinates by the time-independent symplectic (i.e., preserving the symplectic two-form $\Omega$) change of coordinates given by the following definition.}

\begin{definition}[Symmetric geometric Darboux coordinates]\label{DefNewcoor2} We define the symmetric geometric Darboux coordinates $\left(\mathbf{Q}_\infty, \mathbf{P}_{\infty}\right)\coloneqq\left(Q_{\infty, 0},\dots, Q_{\infty,2d-1}, P_{\infty, 0},\dots, P_{\infty,2d-1}\right)$ by 
 \bea \check{L}_{1,2}(\lambda)=-t_{\infty,2d+1}\underset{j=1}{\overset{2d}{\prod}}(\lambda-q_j)&=& -t_{\infty,2d+1}\left(\sum_{k=0}^{2d-1} Q_{\infty,k}\lambda^k +\lambda^{2d}\right),\\
\forall \, i\in \llbracket 1,2d\rrbracket\,\,:\,\, p_i&=&\sum_{k=0}^{2d-1} P_{\infty,k}\frac{\partial Q_{\infty,k}(q_1,\dots,q_{2d})}{\partial q_i}. \nonumber
\eea   
\end{definition}

The fact that this change of coordinates is symplectic is a consequence of Lemma $6.3$ of \cite{MarchalP1Hierarchy}. Consequently, the Darboux coordinates $\left(\mathbf{Q}_\infty, \mathbf{P}_{\infty}\right)$ are canonical and the corresponding Hamiltonian is obtained by taking the Hamiltonian for $\left(\mathbf{q},\mathbf{p}\right)$ and replacing the coordinates $\left(\mathbf{q},\mathbf{p}\right)$ by their expression in terms of $\left(\mathbf{Q}_\infty,\mathbf{P}_\infty\right)$. Moreover, we obtain explicit expression of $\check{L}(\lambda)$ using these Darboux coordinates:

\begin{theorem}[Expression of the Lax matrices $\check{L}(\lambda)
$ and $\check{A}_{\boldsymbol{\alpha}}(\lambda)$ in terms of $\left(\mathbf{Q}_\infty,\mathbf{P}_\infty\right)$]\label{TheoLaxMatricesQinfty} We have:
\bea \left[\check{L}(\lambda)\right]_{1,2}&=&-t_{\infty,2d+1}\left(\sum_{k=0}^{2d-1} Q_{\infty,k}\lambda^k +\lambda^{2d}\right),\cr
\left[\check{L}(\lambda)\right]_{1,1}&=&-\sum_{k=0}^{2 d-1} P_{\infty,2d-1-k}\lambda^k-\sum_{k=0}^{2d-2}\sum_{m=0}^{2d-2-k}P_{\infty,m}Q_{\infty,k+1+m}\lambda^k+ \frac{t_{0,0}}{\lambda} \cr
&&+\frac{t_{0,0}}{Q_{\infty,0}}\left(\lambda^{2d-1}+\sum_{k=1}^{2d-1} Q_{\infty,k}\lambda^{k-1}\right),\cr
\left[\check{L}(\lambda)\right]_{2,1}&=& \left[\frac{\underset{j= 2 d-1}{\overset{4d}{\sum}}\left(\underset{m=0}{\overset{4d-j}{\sum}} t_{\infty,2 d +1-m}t_{\infty,j+m-2d+1}\right) \lambda^{j} -\check{L}_{1,1}(\lambda)^2}{\check{L}_{1,2}(\lambda)}\right]_{\infty,+}\cr&&
-\frac{1}{\lambda}\Res_{\lambda\to \infty} \frac{\underset{j= 2 d-1}{\overset{4d}{\sum}}\left( \underset{m=0}{\overset{4d-j}{\sum}} t_{\infty,2 d +1-m}t_{\infty,j+m-2d+1} \right) \lambda^{j}  -\check{L}_{1,1}(\lambda)^2}{\check{L}_{1,2}(\lambda) }, \cr
\left[\check{L}(\lambda)\right]_{2,2}&=&-\left[\check{L}(\lambda)\right]_{1,1},
\eea
and 
\bea \left[\check{A}_{\boldsymbol{\alpha}}(\lambda)\right]_{1,2}&=&-t_{\infty,2d+1}\left(\sum_{i=0}^{2d-1
}\nu_{\infty,i}^{(\boldsymbol{\alpha})}\lambda^{2d-i} +\sum_{k=1}^{2d-1}\sum_{i=0}^{k-1
}\nu_{\infty,i}^{(\boldsymbol{\alpha})}Q_{\infty,k} \lambda^{k-i}\right),\cr
\left[\check{A}_{\boldsymbol{\alpha}}(\lambda)\right]_{1,1}&=&
\left[\check{L}_{1,1}(\lambda) \left(\sum_{i=0}^{2d-1} \nu^{(\boldsymbol{\alpha})}_{\infty,i}\lambda^{-i}\right)\right]_{\infty,+},\cr
\left[\check{A}_{\boldsymbol{\alpha}}(\lambda)\right]_{2,1}&=& -\frac{t_{0,0}\mathcal{L}_{\boldsymbol{\alpha}}[Q_{\infty,0}]}{t_{\infty,2d+1}(Q_{\infty,0})^2\lambda} + \left[\frac{\check{L}_{1,1}(\lambda)}{\check{L}_{1,2}(\lambda)}\left(\frac{\check{L}_{1,1}(\lambda)}{\check{L}_{1,2}(\lambda)}\left[\check{A}_{\boldsymbol{\alpha}}(\lambda)\right]_{1,2}-2\left[\check{A}_{\boldsymbol{\alpha}}(\lambda)\right]_{1,1}\right)\right]_{\infty,+} \cr&&
+\left[\frac{\check{L}_{1,1}(\lambda)}{\check{L}_{1,2}(\lambda)}\left(\frac{\check{L}_{1,1}(\lambda)}{\check{L}_{1,2}(\lambda)}\left[\check{A}_{\boldsymbol{\alpha}}(\lambda)\right]_{1,2}-2\left[\check{A}_{\boldsymbol{\alpha}}(\lambda)\right]_{1,1}\right)\right]_{0,-}\cr
&&+\frac{t_{0,0}(1-t_{0,0})\left(\nu_{\infty,2d-1}^{\boldsymbol{(\alpha)}} +\underset{k=1}{\overset{2d-1}{\sum}} \nu_{\infty,k-1}^{(\boldsymbol{\alpha})} Q_{\infty,k}\right) }{t_{\infty,2d+1}(Q_{\infty,0})^2\lambda}\cr
&&+\left[-\frac{\underset{k=2d}{\overset{4d}{\sum}}\underset{j=k}{\overset{4d}{\sum}}\underset{m=0}{\overset{4d-j}{\sum}} t_{\infty,2 d +1-m}t_{\infty,j+m-2d+1}\nu^{(\boldsymbol{\alpha})}_{\infty,j-k} \lambda^{k}
}{t_{\infty,2d+1}\left(\underset{k=0}{\overset{2d-1}{\sum}} Q_{\infty,k}\lambda^k +\lambda^{2d}\right)}\right]_{\infty,+},\cr
\left[\check{A}_{\boldsymbol{\alpha}}(\lambda)\right]_{2,2}&=&-\left[\check{A}_{\boldsymbol{\alpha}}(\lambda)\right]_{1,1},
\eea
with $\nu_{\infty,2d}^{(\boldsymbol{\alpha})}\coloneqq -\underset{\lambda\to \infty}{\Res}\lambda^{2d-1}\frac{[\check{A}_{\boldsymbol{\alpha}}(\lambda)]_{1,2}}{\check{L}_{1,2}(\lambda)}$.
\end{theorem}

\begin{proof}The proof is given in \autoref{AppendixLaxExpression}.  
\end{proof}

The explicit expression of the matrix $\check{L}(\lambda)$ in terms of the Darboux coordinates $\left(\mathbf{Q}_\infty,\mathbf{P}_{\infty}\right)$ and the fact that the change of coordinates $\left(\mathbf{q},\mathbf{p}\right)$ to $\left(\mathbf{Q}_\infty,\mathbf{P}_{\infty}\right)$ is time-independent and symplectic implies that we get explicit formulas for the Hamiltonian evolutions of $\left(\mathbf{Q}_\infty,\mathbf{P}_{\infty}\right)$. 

\begin{theorem}[Hamiltonian for $\left(\mathbf{Q}_\infty,\mathbf{P}_\infty\right)$] \label{HamTheorem2} The Hamiltonian for $\left(\mathbf{Q}_\infty,\mathbf{P}_\infty\right)$ is given by
\beq \label{DefHam} \mathrm{Ham}_{(\boldsymbol{\alpha})}(\mathbf{Q}_\infty,\mathbf{P}_\infty)\coloneqq\sum_{k=0}^{2d-2} \nu_{\infty,k+1}^{\boldsymbol{(\alpha)}}\check{H}_{\infty,k}+\nu_{\infty,0}^{\boldsymbol{(\alpha)}}\left(\check{H}_{0,1} +2d P_{\infty,2d-1} +\sum_{k=0}^{2d-2} (k+1)Q_{\infty,k+1}P_{\infty,k}\right),
\eeq
in the sense that 
\beq \forall\, k\in \llbracket 0,2d-1\rrbracket \,:\, \mathcal{L}_{\boldsymbol{\alpha}}[Q_{\infty,k}]=\frac{\partial \mathrm{Ham}_{(\boldsymbol{\alpha})}(\mathbf{Q}_\infty,\mathbf{P}_\infty)}{\partial P_{\infty,k}}\, \mathrm{ and }\, \mathcal{L}_{\boldsymbol{\alpha}}[P_{\infty,k}]=-\frac{\partial \mathrm{Ham}_{(\boldsymbol{\alpha})}(\mathbf{Q}_\infty,\mathbf{P}_\infty)}{\partial Q_{\infty,k}}.\eeq
Coefficients $\left(\nu_{\infty,j}^{\boldsymbol{(\alpha)}}\right)_{0\leq j\leq 2d-1}$ are determined by  \autoref{PropAsymptoticExpansionA12} while coefficients $\left(\check{H}_{\infty,k}\right)_{0\leq k\leq 2d-2} \cup \check{H}_{0,1}$ are given by
\bea  \label{DefHs}
\forall \, j\in \llbracket 0, 2d-2\rrbracket\,:\, \check{H}_{\infty,j}(\mathbf{Q}_{\infty},\mathbf{P}_{\infty};\mathbf{t},t_{\infty,0},t_{0,0})&= &-\Res_{\lambda\to \infty}\lambda^{-j-1} \left[ (\check{L}_{1,1})^2+\check{L}_{2,1}\check{L}_{1,2} +\check{L}_{1,2}\partial_\lambda\left(\frac{ \check{L}_{1,1}}{\check{L}_{1,2}}\right)\right], \nonumber \\
\check{H}_{0,1} (\mathbf{Q}_{\infty},\mathbf{P}_{\infty};\mathbf{t},t_{\infty,0},t_{0,0})&=&\Res_{\lambda\to 0}(\check{L}_{1,1})^2+\check{L}_{2,1}\check{L}_{1,2} +\check{L}_{1,2} \partial_\lambda\left(\frac{ \check{L}_{1,1}}{\check{L}_{1,2}}\right),
\eea
where the Lax matrix $\check{L}(\lambda)$ is given in terms of $\left(\mathbf{Q}_\infty,\mathbf{P}_\infty\right)$ by \autoref{TheoLaxMatricesQinfty}.
\end{theorem}

For completeness, we mention that \autoref{HamTheorem2} can be used to obtain the Hamiltonian in every time-direction:

\begin{proposition}
\label{PropHamReduced2}The Hamiltonian evolution can be rewritten as:
\beq \label{NewHamReduced2}\begin{pmatrix}\mathrm{Ham}_{{t_{\infty,1}}}(\mathbf{Q}_{\infty},\mathbf{P}_{\infty};\mathbf{t},t_{\infty,0},t_{0,0})\\ \vdots \\ (2d-1)\mathrm{Ham}_{{t_{\infty,2d-1}}}(\mathbf{Q}_{\infty},\mathbf{P}_{\infty};\mathbf{t},t_{\infty,0},t_{0,0})\\
(2d)\mathrm{Ham}_{{t_{\infty,2d}}}(\mathbf{Q}_{\infty},\mathbf{P}_{\infty};\mathbf{t},t_{\infty,0},t_{0,0})
\end{pmatrix}=\left(M_{\infty}(\mathbf{t})\right)^{-1}\begin{pmatrix}\check{H}_{\infty,2d-2}(\mathbf{Q}_{\infty},\mathbf{P}_{\infty};\mathbf{t},t_{\infty,0},t_{0,0})\\ \vdots\\ \check{H}_{\infty,0}(\mathbf{Q}_{\infty},\mathbf{P}_{\infty};\mathbf{t},t_{\infty,0},t_{0,0}) \\ \check{H}_{0,1}(\mathbf{Q}_{\infty},\mathbf{P}_{\infty};\mathbf{t},t_{\infty,0},t_{0,0})+\check{\delta}_{0,1}(\mathbf{Q}_\infty,\mathbf{P}_\infty)
\end{pmatrix},\eeq 
with
\beq \check{\delta}_{0,1}(\mathbf{Q}_\infty,\mathbf{P}_\infty)\coloneqq2d P_{\infty,2d-1} +\sum_{k=0}^{2d-2} (k+1)Q_{\infty,k+1}P_{\infty,k}.\eeq
\end{proposition}

\begin{proof}The proof follows from the fact that
\beq \label{Aux}\text{Ham}_{(\boldsymbol{\alpha})}(\mathbf{Q}_\infty,\mathbf{P}_\infty;\mathbf{t},t_{\infty,0},t_{0,0})= \left(\nu^{(\boldsymbol{\alpha})}_{\infty,2d-1},\dots,\nu^{(\boldsymbol{\alpha})}_{\infty,0}\right)\begin{pmatrix}\check{H}_{\infty,2d-2}(\mathbf{Q}_\infty,\mathbf{P}_\infty;\mathbf{t},t_{\infty,0},t_{0,0})\\ \vdots\\ \check{H}_{\infty,0}(\mathbf{Q}_\infty,\mathbf{P}_\infty;\mathbf{t},t_{\infty,0},t_{0,0})\\ \check{H}_{0,1}(\mathbf{Q}_\infty,\mathbf{P}_\infty;\mathbf{t},t_{\infty,0},t_{0,0})+\delta_{0,1}(\mathbf{Q}_\infty,\mathbf{P}_\infty) \end{pmatrix}. \eeq
\autoref{PropAsymptoticExpansionA12} gives:
\bea &&\left(\nu^{(\boldsymbol{\alpha})}_{\infty,0},\dots,\nu^{(\boldsymbol{\alpha})}_{\infty,2d-1}\right) (M_{\infty}(\mathbf{t}))^{\mathsf T}=\left(\frac{\alpha_{\infty,2d}}{2d},\dots, \frac{\alpha_{\infty,1}}{1}\right) \cr
&&\Leftrightarrow \,\,\left(\nu^{(\boldsymbol{\alpha})}_{\infty,2d-1},\dots,\nu^{(\boldsymbol{\alpha})}_{\infty,0}\right) J (M_\infty(\mathbf{t}))^{\mathsf T}=\left(\frac{\alpha_{\infty,2d}}{2d},\dots, \frac{\alpha_{\infty,1}}{1}\right), \eea
where $J$ is the $(2d)\times (2d)$ matrix with only $1$ on the anti-diagonal. Thus we have
\bea \left(\nu^{(\boldsymbol{\alpha})}_{\infty,2d-1},\dots,\nu^{(\boldsymbol{\alpha})}_{\infty,0}\right) &=&\left(\frac{\alpha_{\infty,2d}}{2d},\dots, \frac{\alpha_{\infty,1}}{1}\right)(M_\infty(\mathbf{t})^{-1})^{\mathsf T}J= \left(\frac{\alpha_{\infty,1}}{1} ,\dots,\frac{\alpha_{\infty,2d}}{2d} \right)
J(M_\infty(\mathbf{t})^{-1})^{\mathsf T}J\cr
&=&\left(\frac{\alpha_{\infty,1}}{1} ,\dots,\frac{\alpha_{\infty,2d}}{2d} \right)M_\infty(\mathbf{t})^{-1}, \eea
where we have used that for any lower-triangular Toeplitz matrix $T$, we have $JT^{\mathsf T}J=T$. Inserting this relation into \eqref{Aux} ends the proof.
\end{proof}

\section{The FN PII hierarchy as a self-similarity reduction of the mKdV hierarchy}\label{SecFNP2}
The FN PII hierarchy is obtained from the mKdV hierarchy via the \textit{self-similarity reduction} following the symmetries of the mKdV hierarchy \cite{Ablowitz1977ExactLO,Airault1979RationalSO,Flaschka1980}. Let us briefly recall the construction. 

\begin{definition}[The mKdV hierarchy] Define the following Lenard type recursion:
\begin{align}
     \forall\, m\geq0\,:\, \frac{\partial}{\partial x} R_{m+1}[v] =& \left( \frac{\partial^3}{\partial x^3} + 4(v_x - v^2) \frac{\partial}{\partial x} + 2(v_x-v^2)_x \right) R_m[v], \qquad  \text{with} \qquad  
    R_0[v] = \frac{1}{2}.
\end{align}
The mKdV hierarchy is defined as the set of PDEs given by
\begin{align}
    \frac{\partial}{\partial T_{m+1}} v + \frac{\partial}{\partial x} \left( \frac{\partial}{\partial x} +2v \right)R_m[v_x-v^2] = 0, \,\,\,\,\, \forall\, m \in \mathbb{N},
\end{align}
where $(T_m)_{m\geq 1}$ and $x$ are the parameters.
\end{definition}
Every equation in the hierarchy can be seen as a Hamiltonian flow and each single equation defines a symmetry for all the others, this is due to the commutation of the Hamiltonian flows. The solution space for the $m$-th mKdV equation denoted $v(x;T_1,T_2,...)$ corresponds to the additional condition $\frac{\partial v}{\partial T_{m+1} }=0$. Furthermore, this hierarchy has an additional set of symmetries, let us for this define the Virasoro infinitesimal generator 
\begin{align}
    \frac{d}{d s_m} \coloneqq  \sum_{l=0}^m (2l+1) T_{l+1} \frac{\partial}{\partial T_{l+1}}. \end{align}
The stationary solutions satisfy $\frac{d v}{d s_m} = 0$ and thus one obtains from the form of the hierarchy
\begin{align}
    \frac{d v }{ d s_m} = - \sum_{l=0}^m (2l+1) T_{l+1} \frac{\partial}{\partial x} \left( \frac{\partial}{\partial x} +2v \right)R_l[v_x-v^2] = 0,
\end{align}
which, after integration, provides
\begin{align}
    - \sum_{l=0}^m (2l+1) T_{l+1} \left( \frac{\partial}{\partial x} +2v \right)R_l[v_x-v^2] = \alpha_m,
\end{align}
where $\alpha_m$ is the integration constant in this setup. From the $m=0$ equation of the hierarchy, set $T_1=-x$ so that the above equation is nothing but an ODE in the variable $x$ dependent on some additional parameters. These parameters can be absorbed due to the additional symmetry reduction
\begin{align}
    v(x,T_{m+1} ) =& \frac{u(z)}{[(2m+1) T_{m+1}]^{\frac{1}{2m+1}}} , \qquad z = \frac{x}{\left[(2m+1) T_{m+1}\right]^{\frac{1}{2m+1}}}, \cr
    R_l[v_x -v^2] =& \frac{1}{\left[(2m+1) T_{m+1}\right]^{\frac{2l}{2m+1}}} \text{Len}_l[u_z-u^2], \cr
    t_0 = & -z, \qquad t_l = \frac{(2l+1) T_{m+1}}{\left[(2m+1) T_{m+1}\right]^{\frac{2l+1}{2m+1}}},\,\, \forall\, l \in \llbracket 1, m\rrbracket \,\,,\,\,\, t_m = 1.
 \end{align}
In this way, one obtains the FN PII hierarchy defined as follows:
\begin{definition}[The FN PII hierarchy] The FN PII hierarchy is defined as the set of nonlinear ODEs given by 
\begin{align}
    \left( \frac{d}{d z} + 2  u \right) \text{Len}_n[u_z-u^2] + \sum_{l=1}^{m-1}t_l  \left( \frac{d}{d z} + 2  u \right) \text{Len}_l[u_z-u^2] = z u + \alpha_m, 
 \end{align}
 with the operator $\text{Len}_m$ satisfying the Lenard recursion
 \begin{align}
      \text{Len}_0[u_z-u^2] =& \frac{1}{2} \qquad  \text{and} \qquad
    \frac{d}{d z} \text{Len}_{m+1} = \left( \frac{d^3}{d z^3} + 4(u_z - u^2) \frac{d}{d z} + 2(u_z-u^2)_z \right) \text{Len}_m,
 \end{align} 
 and $\alpha_m$ is a constant.
\end{definition}

In \cite{Flaschka1980}, the FN PII equation was recovered from the compatibility equations of linear differential equations involving the FN PII Lax matrix via the above reduction. Following a proposal of \cite{https://doi.org/10.1002/sapm1974534249}, the construction has been generalized to the full hierarchy in \cite{articleCJM,Kudryashov} for specific values of the times and then by \cite{mazzocco2007hamiltonian} for generic values. This last work provides the isomonodromic formulation of the FN PII hierarchy summarized below.

\begin{definition} \label{DefIsomondoromicProb}
    Let $d\geq 1$. The isomonodromic deformation problem of the $d$-th member of the FN PII hierarchy is given by 
    \bea
        \frac{\partial}{\partial \lambda} \Psi_{\text{FN}}(\lambda) &=& L^{(d)}_{\text{FN}}(\lambda) \Psi_{\text{FN}}(\lambda), \qquad \frac{\partial}{\partial z} \Psi_{\text{FN}}(\lambda) = \begin{pmatrix}
            - \lambda &u \\ u& \lambda    \end{pmatrix} \Psi_{\text{FN}}(\lambda)  \cr
         \text{and } \,  \frac{\partial}{\partial t_k} \Psi_{\text{FN}}(\lambda) &=& \frac{1}{2 k +1}\left( M^{(k)}(\lambda) - \begin{pmatrix}
             0  & (\partial_z  + 2 u ) \text{Len}_k \\
             (\partial_z  + 2 u ) \text{Len}_k  & 0
         \end{pmatrix} \right) \Psi_{\text{FN}}(\lambda)\cr
         &\coloneqq&A_{FN,k}^{(d)}(\lambda)\Psi_{\text{FN}}(\lambda),\,\, \forall\, k\in \llbracket 1,d-1\rrbracket,
    \eea
    with the Lax matrix 
    \beq L_{\text{FN}}^{(d)}(\lambda)\coloneqq\frac{1}{\lambda}\left[\begin{pmatrix}-\lambda z& -\alpha_d\\ -\alpha_d& \lambda z \end{pmatrix}+M^{(d)}(\lambda)+\sum_{i=1}^{d-1} t_i M^{(i)}(\lambda)\right]\eeq
and with $M^{(l)}(\lambda)$ a polynomial in $\lambda$ given by
\beq \forall\, l\in \llbracket 1,d\rrbracket\,:\,  M^{(l)}(\lambda)\coloneqq\begin{pmatrix}\underset{j=1}{\overset{2l+1}{\sum}} A_j^{(l)} \lambda^j & \underset{j=1}{\overset{2l}{\sum}} B_j^{(l)} \lambda^j  \\ \underset{j=1}{\overset{2l}{\sum}} C_j^{(l)}\lambda^j & -\underset{j=1}{\overset{2l+1}{\sum}} A_j^{(l)} \lambda^j \end{pmatrix}
\eeq 
with the coefficients 
\begin{align}
    A^{(l)}_{2l+1} \coloneqq& 4^l, \qquad A_{2k} = 0,\qquad  \forall\,k \in \llbracket 0,l\rrbracket,  \nonumber \\
    A^{(l)}_{2k+1} \coloneqq& \frac{4^{k+1}}{2} \left[ \text{Len}_{l-k} (u_z-u^2) - \frac{d}{d z } \left( \frac{d}{d z }  + 2 u \right) \text{Len}_{l-k-1} [u_z-u^2]   \right], \qquad \forall\,k \in \llbracket 0,l-1\rrbracket,\nonumber \\
     B^{(l)}_{2k+1} \coloneqq& \frac{4^{k+1}}{2} \frac{d}{d z } \left( \frac{d}{d z }  + 2 u \right) \text{Len}_{l-k-1}[u_z-u^2], \qquad \forall\,k \in \llbracket 0,l-1\rrbracket, \nonumber \\ 
      B^{(l)}_{2k} \coloneqq& -4^{k} \left( \frac{d}{d z }  + 2 u \right) \text{Len}_{l-k}[u_z-u^2], \qquad \forall\,k\in \llbracket 1, l\rrbracket, \nonumber \\
      C^{(l)}_{2k+1} \coloneqq& - B^{(l)}_{2k+1}, \,\,\,\forall\, k\in \llbracket0 ,l-1\rrbracket, \qquad \qquad  C^{(l)}_{2k} \coloneqq B^{(l)}_{2k}, \qquad \forall\, k \in \llbracket 0,l\rrbracket.
\end{align}
\end{definition}

In the previous definition, the deformation parameters are $(z=t_0,t_1,\dots,t_{d-1})$ while the spectral parameter is $\lambda$ and $t_d=1$. The Lax matrix is $L_{\text{FN}}^{(d)}$. The compatibility of the above three equations provides the FN PII hierarchy and their consistency has been checked in \cite{mazzocco2007hamiltonian}. We will not reproduce the computations here and will state the main observation behind the above definition: the linear problem of \autoref{DefIsomondoromicProb} is the isomonodromic problem of the FN PII hierarchy. It is rather useful to regroup the coefficients of the matrices into the following form:
\beq L_{\text{FN}}^{(d)}(\lambda)\coloneqq\begin{pmatrix}\underset{k=0}{\overset{d}{\sum}} a_{2k+1}^{(d)}\lambda^{2k} & \underset{k=0}{\overset{d}{\sum}}  b_{2k}^{(d)}\lambda^{2k-1}+ \underset{k=0}{\overset{d-1}{\sum}}  b_{2k+1}^{(d)}\lambda^{2k}\\\underset{k=0}{\overset{d}{\sum}}  b_{2k}^{(d)}\lambda^{2k-1}- \underset{k=0}{\overset{d-1}{\sum}}  b_{2k+1}^{(d)}\lambda^{2k}& -\underset{k=0}{\overset{d}{\sum}}  a_{2k+1}^{(d)}\lambda^{2k}
    \end{pmatrix} \,.
\eeq
The correspondence between $\left(A_j^{(l)},B_j^{(l)},C_j^{(l)}\right)_{j,l}$ and $(a_{2k+1},b_{2k},b_{2k+1})_{0\leq k\leq d}$ is given by
\begin{align}
a_1^{(d)}=&\sum_{l=1}^d t_l A_{1}^{(l)} -z, \qquad  b_0^{(d)}=-\alpha_d, \qquad a_{2k+1}^{(d)}=\sum_{l=1}^d t_l A_{2k+1}^{(l)} \,\,,\,\, \forall \, k\in \llbracket 1,d\rrbracket, \cr
b_{2k+1}^{(d)}=&\sum_{l=1}^d t_l B_{2k+1}^{(l)} \,\,,\,\, \forall \, k\in \llbracket 1,d-1\rrbracket, \quad \quad b_{2k}^{(d)}=\sum_{l=1}^d t_l B_{2k}^{(l)} \,\,,\,\, \forall \, k\in \llbracket 1,d\rrbracket.
\end{align} 
The isomonodromic reformulation of the hierarchy allows one to attach a linear differential system whose isomonodromic deformations provide the members of the hierarchy. However, in order to have explicit expressions, one needs to introduce a set of Darboux coordinates. In \cite{mazzocco2007hamiltonian}, the construction of the Darboux coordinates is based on considering the problem as a vector field on the coadjoint orbits of a twisted loop algebra. Using the Poisson structure on the dual loop algebra one obtains a degenerate Poisson-Lie bracket whose symplectic leaves are the co-adjoint orbits of the elements. This idea of the Poisson structure on dual loop algebras is known and well documented, see for instance \cite{Babelon_Bernard_Talon_2003,HarnadRouthier1994}. This is the main construction of the isospectral approach to integrable hierarchies. Following this, one obtains the Casimir elements as the times $t_1,\dots,t_{d}$.  One important observation of this approach is that the dimension of the coadjoint orbits is $2d$, this is of central importance when discussing the definition of the coordinates used in \cite{mazzocco2007hamiltonian} where the choice was made to compensate this miss-match. Note that this is due to a grading on the loop algebra manifested by a symmetry on the entries of the Lax matrices. Indeed, in this setting, the Lax matrix and the wave matrix are constrained by the symmetry relations:
\begin{align}
     \Psi_{\text{FN}}(-\lambda)= \sigma_1 \Psi_{\text{FN}}(\lambda)\sigma_1, \qquad 
L_{\text{FN}}^{(d)}(-\lambda)=-\sigma_1 L_{\text{FN}}^{(d)}(\lambda)\sigma_1, \qquad A_{\text{FN},k}^{(d)}(-\lambda)=\sigma_1 A_{\text{FN},k}^{(d)}(\lambda)\sigma_1,
\end{align}
where we denote $(\sigma_1,\sigma_2,\sigma_3)$ the three Pauli matrices:
\begin{align}
    \label{DefPauliMatrices} \sigma_1\coloneqq\begin{pmatrix} 0&1\\ 1&0\end{pmatrix}\,,\, \qquad\sigma_2\coloneqq\begin{pmatrix}
    0&-i\\ i&0
\end{pmatrix}\,,\, \qquad \sigma_3\coloneqq\begin{pmatrix} 1&0\\0&-1\end{pmatrix}.
\end{align}
In particular this symmetry condition implies that the exponential behavior of each column of the wave function $\Psi(\lambda)$ at each pole cannot be a single exponential since the symmetry mixes entry $(1,2)$ with entry $(2,1)$. This indicates that the current gauge necessarily differs from the one used in \autoref{SecP4hierarchy} where this property is always verified and is a central element to derive the Hamiltonians (see \cite{marchal2024hamiltonianrepresentationisomonodromicdeformations} for details about the relation between asymptotics of $\td{\Psi}(\lambda)$ and the Lax matrices). In \cite{alameddine2026symmetryreductionpainleveiv}, the authors tried to enforce on $\td{\Psi}(\lambda)$ the previous symmetry condition and to identify Darboux coordinates canonically but the identification of the Darboux coordinates was complicated. However, upon the proper identification of the Darboux coordinates, the authors could match the underlying Hamiltonian evolutions. As we see below, fixing the normalization issue will allow for a direct correspondence of the Darboux coordinates. In order to achieve a gauge where each column of the wave matrix has pure exponential behavior at each pole, we perform the following gauge transformation:
\begin{definition}[Transformation for the FN PII hierarchy]\label{DefgaugeP2} We define
\beq \check{\Psi}_{\text{FN}}(\lambda)\coloneqq S\Psi_{\text{FN}}(\lambda)S^{-1}. 
\eeq    
\end{definition}
Note that we have $S\sigma_1 S^{-1}=-\sigma_3$ and we have:
\begin{align}
    \check{L}_{\text{FN}}^{(d)}(\lambda)=SL_{\text{FN}}^{(d)}(\lambda) S^{-1} \,, \qquad \check{A}_{\text{FN},k}^{(d)}(\lambda)=SA_{\text{FN},k}^{(d)}(\lambda) S^{-1}.
\end{align} 
Therefore the symmetry condition is equivalent to 
\begin{align}
\check{\Psi}_{\text{FN}}(-\lambda)= \sigma_3 \check{\Psi}_{\text{FN}}(\lambda) \sigma_3, \qquad
 \check{L}_{\text{FN}}^{(d)}(-\lambda)=-\sigma_3 \check{L}_{\text{FN}}^{(d)}(\lambda) \sigma_3, \qquad
 \check{A}_{\text{FN},k}^{(d)}(-\lambda)= \sigma_3  \check{A}_{\text{FN},k}^{(d)}(\lambda) \sigma_3.    
\end{align}
This is equivalent in coordinates to 
\bea
    \left[\check{L}_{\text{FN}}^{(d)}(-\lambda)\right]_{1,1}&=&- \left[\check{L}_{\text{FN}}^{(d)}(\lambda)\right]_{1,1}, \qquad
 \left[\check{L}_{\text{FN}}^{(d)}(-\lambda)\right]_{2,2}=- \left[\check{L}_{\text{FN}}^{(d)}(\lambda)\right]_{2,2}, \qquad
 \left[\check{L}_{\text{FN}}^{(d)}(-\lambda)\right]_{1,2}= \left[\check{L}_{\text{FN}}^{(d)}(\lambda)\right]_{1,2}, \cr
 \left[\check{L}_{\text{FN}}^{(d)}(-\lambda)\right]_{2,1}&=& \left[\check{L}_{\text{FN}}^{(d)}(\lambda)\right]_{2,1}, \qquad
\left[\check{A}_{\text{FN},k}^{(d)}(-\lambda)\right]_{1,1}=\left[\check{A}_{\text{FN},k}^{(d)}(\lambda)\right]_{1,1}, \qquad
\left[\check{A}_{\text{FN},k}^{(d)}(-\lambda)\right]_{2,2}=\left[\check{A}_{\text{FN},k}^{(d)}(\lambda)\right]_{2,2}, \cr
\left[\check{A}_{\text{FN},k}^{(d)}(-\lambda)\right]_{1,2}&=&-\left[\check{A}_{\text{FN},k}^{(d)}(\lambda)\right]_{1,2}, \qquad
\left[\check{A}_{\text{FN},k}^{(d)}(-\lambda)\right]_{2,1}=-\left[\check{A}_{\text{FN},k}^{(d)}(\lambda)\right]_{2,1}.
\eea
More precisely we have:
\small{\bea \check{L}_{\text{FN}}^{(d)}(\lambda)&=&\begin{pmatrix} -\frac{1}{2}\left([L_{\text{FN}}^{(d)}(\lambda)]_{1,2}+[L_{\text{FN}}^{(d)}(\lambda)]_{2,1}\right) & \frac{1}{2}\left([L_{\text{FN}}^{(d)}(\lambda)]_{1,2}-[L_{\text{FN}}^{(d)}(\lambda)]_{2,1}\right)+[L_{\text{FN}}^{(d)}(\lambda)]_{1,1} \\  -\frac{1}{2}\left([L_{\text{FN}}^{(d)}(\lambda)]_{1,2}-[L_{\text{FN}}^{(d)}(\lambda)]_{2,1}\right)+[L_{\text{FN}}^{(d)}(\lambda)]_{1,1}& \frac{1}{2}\left([L_{\text{FN}}^{(d)}(\lambda)]_{1,2}+[L_{\text{FN}}^{(d)}(\lambda)]_{2,1}\right)\end{pmatrix}\cr
&=&\begin{pmatrix} -\underset{k=0}{\overset{d}{\sum}}b_{2k}^{(d)} \lambda^{2k-1} & \underset{k=0}{\overset{d}{\sum}}a_{2k+1}^{(d)} \lambda^{2k}+ \underset{k=0}{\overset{d-1}{\sum}}b_{2k+1}^{(d)} \lambda^{2k} \\ \underset{k=0}{\overset{d}{\sum}}a_{2k+1}^{(d)} \lambda^{2k}- \underset{k=0}{\overset{d-1}{\sum}}b_{2k+1}^{(d)} \lambda^{2k}& \underset{k=0}{\overset{d}{\sum}}b_{2k}^{(d)} \lambda^{2k-1}\end{pmatrix}.
\eea}
\normalsize{Note} that the normalization at infinity is now given by:
\beq \check{L}_{\text{FN}}^{(d)}(\lambda)\overset{\lambda\to \infty}{=} -t_{2d+1} \begin{pmatrix} 0&1\\ 1&0\end{pmatrix} \lambda^{2d} +O(\lambda^{2d-1}),\eeq
and that the symmetry condition $\check{\Psi}_{\text{FN}}(-\lambda)= \sigma_3 \check{\Psi}_{\text{FN}}(\lambda) \sigma_3$ is compatible with the fact that each column of $\check{\Psi}_{\text{FN}}(\lambda)$ is a pure exponential at each pole. 

\medskip

In \cite{mazzocco2007hamiltonian}, the authors introduced a parameterization of the Lax matrix in terms of coordinates $(Q_k,P_k)_{1\leq k\leq d}$ that they proved to be canonical with respect to the Poisson structure. 

\begin{definition}[Darboux coordinates for the FN PII hierarchy (Section $5$ of \cite{mazzocco2007hamiltonian})]\label{DefDarbouxFNP2} Let $\left(q_i^{\text{FN}}\right)_{1\leq i\leq 2d}$ be the zeros of $\left[\check{L}_{\text{FN}}^{(d)}(\lambda)\right]_{1,2}$ such that $q^{\text{FN}}_{i+d}=-q^{\text{FN}}_i$ for all $1\leq i\leq d$. Their dual partners are $\left(p^{\text{FN}}_i\right)_{1\leq i\leq 2d}=\left[\check{L}_{\text{FN}}^{(d)}(q_i)\right]_{1,1}$ and satisfy $p^{\text{FN}}_{i+d}=-p^{\text{FN}}_i$ for all $1\leq i\leq d$. From these oper Darboux coordinates $\left(q^{\text{FN}}_i,p^{\text{FN}}_i\right)_{1\leq i\leq d}$, we define $\left(Q^{\text{FN}}_k,P^{\text{FN}}_k\right)_{1\leq k\leq d}$ by:
\bea Q^{\text{FN}}_k&\coloneqq&\Pi^{\text{FN}}_{2k}=\sum_{1\leq i_1<\dots<i_{2k}\leq 2d} q^{\text{FN}}_{i_1}\dots q^{\text{FN}}_{i_{2k}} \,\,,\,\, \forall \, k\in \llbracket 1, d\rrbracket,\cr
-2\left(p^{\text{FN}}_i+\frac{b_0}{q^{\text{FN}}_i}\right)&=&\sum_{k=1}^{d} P^{\text{FN}}_{k} \frac{\partial Q^{\text{FN}}_{k}}{\partial q^{\text{FN}}_i}  \,\,,\,\,  \forall \, i\in \llbracket 1, 2d\rrbracket.
\eea
Alternatively, we have for all $k\in \llbracket 1,d\rrbracket$:
\begin{align}
    P^{\text{FN}}_k\coloneqq\sum_{j=1}^{d}\frac{1}{2j}b_{2j}^{(d)} \frac{\partial S^{\text{FN}}_{2j}}{\partial \Pi^{\text{FN}}_{2k}},\, \qquad\text{ where } \qquad S^{\text{FN}}_{j}=\sum_{i=1}^{2d} \left(q^{\text{FN}}_i\right)^j \,\,,\,\, \forall \,j\in \llbracket 1,d\rrbracket.
\end{align}  
\end{definition}

Let us note that we have swapped $Q_k \leftrightarrow P_k$ in the definition given in \cite{mazzocco2007hamiltonian}. This choice is motivated for a better identification with the formalism of \autoref{SecP4hierarchy} and so that the symplectic two-form $\Omega=\underset{k=1}{\overset{d}{\sum}} dQ^{\text{FN}}_k\wedge dP^{\text{FN}}_k$ match with no sign difference with the formalism obtained from the Painlev\'{e} IV hierarchy. In particular, we have that $\left(q^{\text{FN}}_i,p^{\text{FN}}_i\right)_{1\leq i\leq d}$ and $\left(Q^{\text{FN}}_k,P^{\text{FN}}_k\right)_{1\leq k\leq d}$ are canonical with respect to the Kostant-Kirillov Poisson structure in the sense that:
\bea \left\{q^{\text{FN}}_i,q^{\text{FN}}_j\right\}&=&0\,\,,\,\, \left\{p^{\text{FN}}_i,p^{\text{FN}}_j\right\}=0 \,\,,\,\, \left\{q^{\text{FN}}_i,p^{\text{FN}}_j\right\}=\delta_{i,j}\,\,,\,\, \forall \,(i,j)\in \llbracket 1,d\rrbracket^2\cr
\left\{Q^{\text{FN}}_i,Q^{\text{FN}}_j\right\}&=&0\,\,,\,\, \left\{P^{\text{FN}}_i,P^{\text{FN}}_j\right\}=0 \,\,,\,\, \left\{Q^{\text{FN}}_i,P^{\text{FN}}_j\right\}=\delta_{i,j}\,\,,\,\, \forall \,(i,j)\in \llbracket 1,d\rrbracket^2
\eea
Note also that we have by definition:
\beq \left[\check{L}_{\text{FN}}^{(d)}(\lambda)\right]_{1,2}=-t_{\infty,2d+1}\prod_{i=1}^d (\lambda^2- q_i^2)=-t_{\infty,2d+1}\prod_{i=1}^{2d} (\lambda- q_i) = -t_{\infty,2d+1}\left(\lambda^{2d} +\sum_{j=0}^{d-1} Q^{\text{FN}}_{d-j} \lambda^{2j}\right)
\eeq
and 
\bea \left[\check{L}_{\text{FN}}^{(d)}(\lambda)\right]_{1,1}&=&\frac{1}{\lambda}\left(-b_0^{(d)}+ \sum_{i=1}^{d} (p_iq_i+b_0^{(d)})  \frac{\underset{1\leq j\neq i \leq d}{\prod}(\lambda^2-q_j^2)}{\underset{1\leq j\neq i \leq d}{\prod}(q_i^2-q_j^2)}\right)\cr
&=&-\frac{b_0^{(d)}}{\lambda}-\sum_{k=1}^{d}\left(P^{\text{FN}}_{k}+\sum_{m=0}^{d-1-k} Q^{\text{FN}}_{d-k-m}P^{\text{FN}}_{d-m}\right)\lambda^{2k-1}.
\eea
Let us now turn our attention to the symmetry reduction and the identification with this formalism. 

\section{Symmetry reduction from the PIV hierarchy to the FN PII hierarchy} \label{secSymmetry}

\subsection{Symmetry reduction}
We will discuss in this section the symmetry reduction that we impose on the horizontal sections and the Lax matrices in the respective normalization that we established in \autoref{SecP4hierarchy} and \autoref{SecFNP2}. This is given in the following definition.

\begin{definition}[Symmetry condition on the PIV hierarchy]\label{DefSymmetry} The symmetry is given by
\begin{align}
\td{\Psi}(-\lambda)= \sigma_1 \td{\Psi}(\lambda)\sigma_1, \qquad
\td{L}(-\lambda)= -\sigma_1 \td{L}(\lambda)\sigma_1, \qquad
\td{A}_{\boldsymbol{\alpha}}(-\lambda)= \sigma_1 \td{A}_{\boldsymbol{\alpha}}(\lambda)\sigma_1.
\end{align}
Translated to the new normalization, one has the diagonal $\mathbb{Z}_2$-involution
\beq \check{\Psi}(-\lambda)= \sigma_3 \check{\Psi}(\lambda)\sigma_3.\eeq
It lifts to an involution at the level of the Lax matrices given by 
\begin{align} \check{L}(-\lambda) = -\sigma_3 \check{L}(\lambda)\sigma_3,  \qquad 
\check{A}_{\boldsymbol{\alpha}}(-\lambda) =  \sigma_3 \check{A}(\lambda)\sigma_3.
\end{align}
\end{definition}

The first symmetry is the one studied in \cite{alameddine2026symmetryreductionpainleveiv}. In particular, it reproduces the symmetry corresponding to the FN PII hierarchy described in \autoref{SecFNP2}. It turns out that the symmetry is more transparent in the $\check{\Psi}(\lambda)$ normalization, because $\sigma_3$ is diagonal and thus the symmetry is equivalent to parity conditions on the entries of the Lax matrices:
\begin{align}
\left[\check{L}(-\lambda)\right]_{1,1}=&- \left[\check{L}(\lambda)\right]_{1,1}, \qquad
\left[\check{L}(-\lambda)\right]_{2,2}=- \left[\check{L}(\lambda)\right]_{2,2}, \qquad
\left[\check{L}(-\lambda)\right]_{1,2}= \left[\check{L}(\lambda)\right]_{1,2}, \nonumber \\
\left[\check{L}(-\lambda)\right]_{2,1}=& \left[\check{L}(\lambda)\right]_{2,1}, \qquad
\left[\check{A}_{\boldsymbol{\alpha}}(-\lambda)\right]_{1,1}=\left[\check{A}_{\boldsymbol{\alpha}}(\lambda)\right]_{1,1}, \qquad
\left[\check{A}_{\boldsymbol{\alpha}}(-\lambda)\right]_{2,2}=\left[\check{A}_{\boldsymbol{\alpha}}(\lambda)\right]_{2,2}, \nonumber \\ 
\left[\check{A}_{\boldsymbol{\alpha}}(-\lambda)\right]_{1,2}=&-\left[\check{A}_{\boldsymbol{\alpha}}(\lambda)\right]_{1,2}, \qquad
\left[\check{A}_{\boldsymbol{\alpha}}(-\lambda)\right]_{2,1}=-\left[\check{A}_{\boldsymbol{\alpha}}(\lambda)\right]_{2,1}.
\end{align}

\subsection{Explicit Hamiltonian representation of the FN PII hierarchy}
The previous symmetry implies that half of the symmetric geometric Darboux coordinates and half of the times of the PIV hierarchy vanish. More precisely, one has the following result. 
\begin{theorem}[Reduced Lax and Hamiltonian for the PIV hierarchy]\label{TheoReducedLaxHamSym} Under the symmetry of \autoref{DefSymmetry}, we have:
\begin{align}
    t_{\infty,2k}=&0 \,\,,\,\, \forall \, k\in \llbracket 0, d\rrbracket , \qquad Q_{\infty,2k+1}=0 \,\,,\,\, \forall \, k\in \llbracket 0, d-1\rrbracket, \nonumber  \\
P_{\infty,2k+1}=&0 \,\,,\,\, \forall \, k\in \llbracket 0, d-1\rrbracket \qquad
\nu_{2k}^{(\boldsymbol{\alpha})}=0 \,\,,\,\, \forall \, k\in \llbracket 0, d\rrbracket, \nonumber \\
\check{H}_{0,1}=&0, \qquad
\check{H}_{\infty,2k+1}= 0 \,\,,\,\, \forall \, k\in \llbracket 0, d-2\rrbracket.
\end{align}
Consequently the evolutions of the reduced Darboux coordinates $\left(Q_{\infty,2k},P_{\infty,2k}\right)_{0\leq k\leq d-1}$ are Hamiltonians. Moreover, the Hamiltonians are given by
\beq \label{DefHamRed2} \mathrm{Ham}^{\mathrm{red}}_{(\boldsymbol{\alpha})}(Q_{\infty,0},\dots,Q_{\infty,2d-2},P_{\infty,0},\dots,P_{\infty,2d-2};t_{\infty,1},\dots,t_{\infty,2d+1},t_{0,0})=\sum_{k=0}^{d-1} \nu_{\infty,2k+1}^{\boldsymbol{(\alpha)}}\check{H}^{\mathrm{red}}_{\infty,2k},
\eeq
with
\small{\beq \label{RelationNuAlphaInftyRed} M^{\mathrm{red}}_\infty\begin{pmatrix} \nu^{(\boldsymbol{\alpha})}_{\infty,1}\\\nu^{(\boldsymbol{\alpha})}_{\infty,3}\\ \vdots \\ \nu^{(\boldsymbol{\alpha})}_{\infty,2d-1}\end{pmatrix}=\begin{pmatrix}\frac{\alpha_{\infty,2d-1}}{2d-1}\\\frac{\alpha_{\infty,2d-3}}{2d-3}\\ \vdots \\ \frac{\alpha_{\infty,1}}{1}\end{pmatrix} \,\text{ , } \,M^{\mathrm{red}}_{\infty}(t_{\infty,3},\dots,t_{\infty,2d+1})\coloneqq\begin{pmatrix}
   t_{\infty,2d+1}&0&\dots &\dots&0 \\
t_{\infty,2d-1} &t_{\infty,2d+1} &\ddots  & &0\\
\vdots &\ddots&\ddots&\ddots&0\\
\vdots &&\ddots&\ddots&0\\
t_{\infty,3} &\dots&\dots&t_{\infty,2d-1}& t_{\infty,2d+1}\\ 
\end{pmatrix},\eeq}
\normalsize{or} equivalently
\small{
\bea &&\label{NewHamReduced4}\begin{pmatrix}\mathrm{Ham}^{\mathrm{red}}_{{t_{\infty,1}}}(Q_{\infty,0},\dots,Q_{\infty,2d-2},P_{\infty,0},\dots,P_{\infty,2d-2};t_{\infty,1},\dots,t_{\infty,2d+1},t_{0,0})\\ \vdots \\ (2d-3)\mathrm{Ham}^{\mathrm{red}}_{{t_{\infty,2d-3}}}(Q_{\infty,0},\dots,Q_{\infty,2d-2},P_{\infty,0},\dots,P_{\infty,2d-2};t_{\infty,1},\dots,t_{\infty,2d+1},t_{0,0})\\
(2d-1)\mathrm{Ham}^{\mathrm{red}}_{{t_{\infty,2d-1}}}(Q_{\infty,0},\dots,Q_{\infty,2d-2},P_{\infty,0},\dots,P_{\infty,2d-2};t_{\infty,1},\dots,t_{\infty,2d+1},t_{0,0})
\end{pmatrix}=\cr
&&\left(M^{\mathrm{red}}_{\infty}(t_{\infty,3},\dots,t_{\infty,2d+1})\right)^{-1}\begin{pmatrix}\check{H}^{\mathrm{red}}_{\infty,2d-2}(Q_{\infty,0},\dots,Q_{\infty,2d-2},P_{\infty,0},\dots,P_{\infty,2d-2};t_{\infty,1},\dots,t_{\infty,2d+1},t_{0,0})\\ \vdots\\ \check{H}^{\mathrm{red}}_{\infty,2}(Q_{\infty,0},\dots,Q_{\infty,2d-2},P_{\infty,0},\dots,P_{\infty,2d-2};t_{\infty,1},\dots,t_{\infty,2d+1},t_{0,0}) \\ \check{H}^{\mathrm{red}}_{\infty,0}(Q_{\infty,0},\dots,Q_{\infty,2d-2},P_{\infty,0},\dots,P_{\infty,2d-2};t_{\infty,1},\dots,t_{\infty,2d+1},t_{0,0}) 
\end{pmatrix}\cr&&\eea}
\normalsize{and}
\bea  \label{DefHsReduced}
&&\forall \, j\in \llbracket 0, d-1\rrbracket\,:\, \check{H}^{\mathrm{red}}_{\infty,2j}(Q_{\infty,0},\dots,Q_{\infty,2d-2},P_{\infty,0},\dots,P_{\infty,2d-2};t_{\infty,1},\dots,t_{\infty,2d+1},t_{0,0})=\cr&&-\Res_{\lambda\to \infty}\lambda^{-2j-1} \left[ \left(\check{L}^{\mathrm{red}}_{1,1}(\lambda)\right)^2+\check{L}^{\mathrm{red}}_{2,1}(\lambda)\check{L}^{\mathrm{red}}_{1,2}(\lambda) +\check{L}^{\mathrm{red}}_{1,2}(\lambda)\,\partial_\lambda\left(\frac{ \check{L}^{\mathrm{red}}_{1,1}(\lambda)}{\check{L}^{\mathrm{red}}_{1,2}(\lambda)}\right)\right].
\eea
The Lax matrices of \autoref{TheoLaxMatricesQinfty} reduce to
\bea \left[\check{L}^{\mathrm{red}}(\lambda)\right]_{1,2}&=&-t_{\infty,2d+1}\left(\sum_{k=0}^{d-1} Q_{\infty,2k}\lambda^{2k} +\lambda^{2d}\right),\cr
\left[\check{L}^{\mathrm{red}}(\lambda)\right]_{1,1}&=&-\sum_{k=1}^{d} P_{\infty,2 d-2k}\lambda^{2k-1}-\sum_{k=1}^{d-1}\sum_{m=0}^{d-1-k}P_{\infty,2m}Q_{\infty,2k+2m}\lambda^{2k-1}+ \frac{t_{0,0}}{\lambda} \cr
&&+\frac{t_{0,0}}{Q_{\infty,0}}\left(\lambda^{2d-1}+\sum_{k=1}^{d-1} Q_{\infty,2k}\lambda^{2k-1}\right),\cr
[\check{L}^{\mathrm{red}}(\lambda)]_{2,1}&=& \left[\frac{\underset{j= d}{\overset{2d}{\sum}}\left(\underset{m=0}{\overset{2d-j}{\sum}} t_{\infty,2 d +1-2m}t_{\infty,2j+2m-2d+1}\right) \lambda^{2j} -\check{L}_{1,1}(\lambda)^2}{\check{L}_{1,2}(\lambda)}\right]_{\infty,+},\cr
\left[\check{L}^{\mathrm{red}}(\lambda)\right]_{2,2}&=&-\left[\check{L}^{\mathrm{red}}(\lambda)\right]_{1,1},
\eea
and 
\bea \left[\check{A}^{\mathrm{red}}_{\boldsymbol{\alpha}}(\lambda)\right]_{1,2}&=&-t_{\infty,2d+1}\left(\sum_{i=0}^{d-1}\nu_{\infty,2i+1}^{(\boldsymbol{\alpha})}\lambda^{2d-2i-1} +
\sum_{k=1}^{d-1}\sum_{i=0}^{k-1}\nu_{\infty,2i+1}^{(\boldsymbol{\alpha})}Q_{\infty,2k} \lambda^{2k-2i-1}\right),\cr
\left[\check{A}^{\mathrm{red}}_{\boldsymbol{\alpha}}(\lambda)\right]_{1,1}&=& \left[\check{L}^{\mathrm{red}}_{1,1}(\lambda) \left(\sum_{i=0}^{d-1} \nu^{(\boldsymbol{\alpha})}_{\infty,2i+1}\lambda^{-2i-1}\right)\right]_{\infty,+},\cr
\left[\check{A}^{\mathrm{red}}_{\boldsymbol{\alpha}}(\lambda)\right]_{2,1}&=& -\frac{t_{0,0}\mathcal{L}_{\boldsymbol{\alpha}}[Q_{\infty,0}]}{t_{\infty,2d+1}(Q_{\infty,0})^2\lambda} + \left[\frac{\check{L}^{\mathrm{red}}_{1,1}(\lambda)}{\check{L}^{\mathrm{red}}_{1,2}(\lambda)}\left(\frac{\check{L}^{\mathrm{red}}_{1,1}(\lambda)}{\check{L}^{\mathrm{red}}_{1,2}(\lambda)}\left[\check{A}^{\mathrm{red}}_{\boldsymbol{\alpha}}(\lambda)\right]_{1,2}-2\left[\check{A}^{\mathrm{red}}_{\boldsymbol{\alpha}}(\lambda)\right]_{1,1}\right)\right]_{\infty,+} \cr&&
+\left[\frac{\check{L}^{\mathrm{red}}_{1,1}(\lambda)}{\check{L}^{\mathrm{red}}_{1,2}(\lambda)}\left(\frac{\check{L}^{\mathrm{red}}_{1,1}(\lambda)}{\check{L}^{\mathrm{red}}_{1,2}(\lambda)}\left[\check{A}^{\mathrm{red}}_{\boldsymbol{\alpha}}(\lambda)\right]_{1,2}-2\left[\check{A}^{\mathrm{red}}_{\boldsymbol{\alpha}}(\lambda)\right]_{1,1}\right)\right]_{0,-}\cr
&&+\frac{t_{0,0}(1-t_{0,0})\left(\nu_{\infty,2d-1}^{\boldsymbol{(\alpha)}} +\underset{k=1}{\overset{d-1}{\sum}} \nu_{\infty,2k-1}^{(\boldsymbol{\alpha})} Q_{\infty,2k}\right) }{t_{\infty,2d+1}(Q_{\infty,0})^2\lambda}\cr
&&+\left[\frac{\underset{k=d}{\overset{2d-1}{\sum}}\underset{j=k+1}{\overset{2d}{\sum}}\underset{m=0}{\overset{2d-j}{\sum}} t_{\infty,2 d +1-2m}t_{\infty,2j+2m-2d+1}\nu^{(\boldsymbol{\alpha})}_{\infty,2j-2k-1} \lambda^{2k+1}
}{t_{\infty,2d+1}\left(\underset{k=0}{\overset{d-1}{\sum}} Q_{\infty,2k}\lambda^{2k} +\lambda^{2d}\right)}\right]_{\infty,+},\cr
\left[\check{A}^{\mathrm{red}}_{\boldsymbol{\alpha}}(\lambda)\right]_{2,2}&=&-\left[\check{A}^{\mathrm{red}}_{\boldsymbol{\alpha}}(\lambda)\right]_{1,1}.
\eea
\end{theorem}

The previous theorem provides an explicit Lax pair representation of the FN PII hierarchy and explicit formulas for the Hamiltonians. Moreover, one can match entries $\left(\left[\check{L}^{\text{red}}(\lambda)\right]_{1,1},\left[\check{L}^{\text{red}}(\lambda)\right]_{1,2}\right)$ with $\left(\left[\check{L}_{\text{FN}}^{(d)}(\lambda)\right]_{1,1},\left[\check{L}_{\text{FN}}^{(d)}(\lambda)\right]_{1,2}\right)$ of \autoref{SecFNP2} to get

\begin{proposition}[Identification of the FN PII hierarchy with the reduced PIV hierarchy after symmetry]\label{PropIdentification}The Lax matrices $\check{L}^{\mathrm{red}}(\lambda)$ and $\check{L}_{\mathrm{FN}}^{(d)}(\lambda)$ match under the identification $t_{\infty,2d+1}=-4^{d}$ and
\begin{align}
    \alpha_d=&t_{0,0}, \qquad
z=t_{\infty,1}, \qquad
-4^{k}t_k= t_{\infty,2k+1}\,\,\,,\,\, \forall\, k\in \llbracket 1,d-1\rrbracket, \nonumber \\
Q_{d-k}^{\mathrm{FN}}=&Q_{\infty,2k} \,\,\,,\,\, \forall\, k\in \llbracket 0,d-1\rrbracket,   \qquad
P_d^{\mathrm{FN}}=P_{\infty,0}-\frac{t_{0,0}}{Q_{\infty,0}}, \nonumber \\
P_{d-k}^{\mathrm{FN}}=&P_{\infty,2k} \,\,,\,\, \forall\, k\in \llbracket 1,d-1\rrbracket.
\end{align}
 
Thus, it is trivial to obtain the Hamiltonian evolutions for $\left(Q_1^{\mathrm{FN}},\dots,Q_d^{\mathrm{FN}}, P_1^{\mathrm{FN}}, \dots, P_d^{\mathrm{FN}}\right)$ from \autoref{TheoReducedLaxHamSym}.\end{proposition}

Note that only $P_d^{\text{FN}}$ has a shift by $-\frac{t_{0,0}}{Q_{\infty,0}}$ which does not modify the symplectic two-form. Thus, up to ordering and the trivial shift by  $-\frac{t_{0,0}}{Q_{\infty,0}}$, our construction of Darboux coordinates from the symmetry reduction of the PIV hierarchy matches with Darboux coordinates proposed in \cite{mazzocco2007hamiltonian}. However, the main advantage of our construction is that it inserts the mKdV hierarchy in the standard isomonodromic deformation framework and also provides explicit formulas for the Hamiltonian and the Lax system in \autoref{TheoReducedLaxHamSym}. This completes the problem of finding the full Hamiltonian structure for the hierarchy. We will illustrate this on the first two members of the hierarchy in the next section.

\section{Examples for $d=1$ and $d=2$} \label{secExamples}
We will present in this section the first two members of the hierarchy as a direct application of the general theory developed in the article. 
\subsection{The case $d=1$}
After symmetry, the case $d=1$ corresponds to the Flaschka-Newell Lax pair for the PII equation. For completeness, we provide the complete computations starting from the Painlev\'{e} IV hierarchy before applying the symmetry. Applying \autoref{TheoLaxMatricesQinfty}, one has the Lax matrix before symmetry:
\begin{align}
    \left[\check{L}(\lambda)\right]_{1,2} = & - t_{\infty,3} \left( Q_{\infty,0} + Q_{\infty,1} \lambda + \lambda^2 \right),  \nonumber \\
    \left[\check{L}(\lambda)\right]_{1,1} = & \left(\frac{t_{0,0}}{Q_{\infty,0}}- P_{\infty,0}\right) \lambda - P_{\infty,1} - Q_{\infty,1} P_{\infty,0}+ \frac{t_{0,0}Q_{\infty,1}}{Q_{\infty,0}} + \frac{t_{0,0}}{\lambda},  \nonumber \\
    \left[\check{L}(\lambda)\right]_{2,1} = & -t_{\infty,3} \lambda^2 + \left( t_{\infty,3} Q_{\infty,1}-2 t_{\infty,2}  \right) \lambda + \frac{\left( \frac{t_{0,0}}{Q_{\infty,0}} - P_{\infty,0} \right)^2 - (t_{\infty,2})^2    }{t_{\infty,3}} 
     - 2 t_{\infty,1} + 2 t_{\infty,2} Q_{\infty,1} \nonumber \\
     & - \left(  (Q_{\infty,1})^2 - Q_{\infty,0} \right) t_{\infty,3} + \frac{1}{\lambda} \bigg( Q_{\infty,0} \left( -2 Q_{\infty,1} t_{\infty,3} + 2 t_{\infty,2} \right) +  t_{\infty,3}(Q_{\infty,1})^3 - 2 t_{\infty,2} (Q_{\infty,1})^2   \nonumber  \\
     & + \frac{Q_{\infty,1}}{t_{\infty,3}} \left( (P_{\infty,0})^2 + 2 t_{\infty,1} t_{\infty,3} + (t_{\infty,2})^2 \right) -2 t_{\infty,0} -2 \frac{t_{\infty,2} t_{\infty,1} - P_{\infty,0}P_{\infty,1} }{t_{\infty,3}}  \nonumber \\  
     &-\frac{2 t_{0,0}}{Q_{\infty,0} t_{\infty,3}} \left( P_{\infty,0} Q_{\infty,1} + P_{\infty,1} \right) + \frac{(t_{0,0})^2 Q_{\infty,1}}{t_{\infty,3} (Q_{\infty,0})^2}   \bigg),  \nonumber   \\
    \left[\check{L}(\lambda)\right]_{2,2} =&  - \left[\check{L}(\lambda)\right]_{1,1}.  
\end{align}
The auxiliary matrix $\check{A}_{\boldsymbol{\alpha}}(\lambda)$ is given by 
\small{\begin{align}
     \left[\check{A}_{\boldsymbol{\alpha}}(\lambda)\right]_{1,2}=& - t_{\infty,3} \nu_{\infty,0}^{(\boldsymbol{\alpha})} \lambda^2 -t_{\infty,3} \left( \nu_{\infty,1}^{(\boldsymbol{\alpha})} +     \nu_{\infty,0}^{(\boldsymbol{\alpha})}  Q_{\infty,1} \right) \lambda, 
     \nonumber \\
     \left[\check{A}_{\boldsymbol{\alpha}}(\lambda)\right]_{1,1}=& \nu_{\infty,0}^{(\boldsymbol{\alpha})} \left( \frac{t_{0,0}}{Q_{\infty,0}} - P_{\infty,0}  \right)  \lambda + \nu_{\infty,1}^{(\boldsymbol{\alpha})} \left( \frac{t_{0,0}}{Q_{\infty,0}} - P_{\infty,0} \right) +  \nu_{\infty,0}^{(\boldsymbol{\alpha})} \left( - P_{\infty,1}  - Q_{\infty,1} P_{\infty,0} + \frac{t_{0,0} Q_{\infty,1}}{Q_{\infty,0}} \right),   \nonumber \\
     \left[\check{A}_{\boldsymbol{\alpha}}(\lambda)\right]_{2,1}=& - t_{\infty,3} \nu_{\infty,0}^{(\boldsymbol{\alpha})} \lambda^2 + \left( t_{\infty,3}\nu_{\infty,0}^{(\boldsymbol{\alpha})} Q_{\infty,1}   -  2t_{\infty,2} \nu_{\infty,0}^{(\boldsymbol{\alpha})}  -  t_{\infty,3} \nu_{\infty,1}^{(\boldsymbol{\alpha})} \right) \lambda + \frac{1}{t_{\infty,3}} \left( \frac{t_{0,0}}{Q_{\infty,0}} - P_{\infty,0}  \right)^2 \nu_{\infty,0}^{(\boldsymbol{\alpha})}  \nonumber \\
     & -  \frac{(t_{\infty,2})^2 \nu_{\infty,0}^{(\boldsymbol{\alpha})}}{t_{\infty,3}} - 2 t_{\infty,1} \nu_{\infty,0}^{(\boldsymbol{\alpha})}  -  2 t_{\infty,2} \nu_{\infty,1}^{(\boldsymbol{\alpha})} + 2t_{\infty,3} \left( Q_{\infty,0} \nu_{\infty,0}^{(\boldsymbol{\alpha})}  + Q_{\infty,1} \nu_{\infty,1}^{(\boldsymbol{\alpha})}  \right)  \nonumber   \\ 
      & -t_{\infty,3} \nu^{(\boldsymbol{\alpha})} _{\infty,0}(Q_{\infty,1})^2 +2t_{\infty,2} \nu^{(\boldsymbol{\alpha})} _{\infty,0}Q_{\infty,1}, 
      \nonumber   \\
     \left[\check{A}_{\boldsymbol{\alpha}}(\lambda)\right]_{2,2}=& - \left[\check{A}_{\boldsymbol{\alpha}}(\lambda)\right]_{1,1}, 
 \end{align}}
 \normalsize{where} $\nu_{\infty,0}^{(\boldsymbol{\alpha})}=\frac{1}{2t_{\infty,3}}\alpha_{\infty,2} $ and $\nu_{\infty,1}^{(\boldsymbol{\alpha})}=\frac{1}{t_{\infty,3}}\alpha_{\infty,1}-\frac{t_{\infty,2}}{2(t_{\infty,3})^2}\alpha_{\infty,2}$. The spectral invariants before the symmetry are provided by \autoref{HamTheorem2} and are given as follows
 \begin{align}
     \check{H}_{\infty,0} = &   -  (t_{\infty,3})^2 (Q_{\infty,0})^2 + \left( 3 (Q_{\infty,1})^2 (t_{\infty,3})^2 - 4 Q_{\infty,1} t_{\infty,2} t_{\infty,3} - (P_{\infty,0})^2 + 2 t_{\infty,1} t_{\infty,3} + (t_{\infty,2})^2  \right) Q_{\infty,0}    \nonumber     \\
     & - (Q_{\infty,1})^4 (t_{\infty,3})^2 +2 (Q_{\infty,1})^3 t_{\infty,2} t_{\infty,3} - \left( 2 t_{\infty,1} t_{\infty,3} +  (t_{\infty,2})^2 \right) (Q_{\infty,1})^2 \nonumber \\
      & + \left( 2 t_{\infty,0} t_{\infty,3} +  2 t_{\infty,1} t_{\infty,2}  \right) Q_{\infty,1} + (P_{\infty,1})^2 + P_{\infty,0} - \frac{t_{0,0} \left( t_{0,0}-1\right)}{Q_{\infty,0}}, \nonumber   \\ 
    \check{H}_{0,1}  = & 2\left( Q_{\infty,1} (t_{\infty,3})^2 - t_{\infty,2} t_{\infty,3}  \right) (Q_{\infty,0})^2 + \big( - (Q_{\infty,1})^3 (t_{\infty,3})^2 + 2 (Q_{\infty,1})^2 t_{\infty,2} t_{\infty,3}  \\
    & - \left((P_{\infty,0})^2  + 2 t_{\infty,1} t_{\infty,3} + (t_{\infty,2})^2 \right) Q_{\infty,1} - 2P_{\infty,0} P_{\infty,1}  + 2 t_{\infty,0} t_{\infty,3} +  2 t_{\infty,1} t_{\infty,2} \big) Q_{\infty,0}  \nonumber   \\ &+ \frac{t_{0,0} \left( t_{0,0}-1\right) Q_{\infty,1}}{Q_{\infty,0}}. \nonumber
   \end{align}
The general Hamiltonian for a general deformation vector is given by
\begin{align}
    \text{Ham}_{(\boldsymbol{\alpha})}(\mathbf{Q}_\infty,\mathbf{P}_\infty) = \nu_{\infty,1}^{(\boldsymbol{\alpha})} \check{H}_{\infty,0}+ \nu_{\infty,0}^{(\boldsymbol{\alpha})} \left( \check{H}_{0,1} + 2 P_{\infty,1} + Q_{\infty,1} P_{\infty,0} \right).
\end{align}

The symmetry condition in the case $d=1$ is equivalent to the following conditions: 
\begin{align}
t_{\infty,0} = t_{\infty,2} = 0=\alpha_{\infty,0}=\alpha_{\infty,2}, \qquad Q_{\infty,1}=0, \qquad P_{\infty,1}=0, \qquad \check{H}_{0,1}=0,  \qquad \nu_{\infty,0}^{(\boldsymbol{\alpha})}=\nu_{\infty,2}^{(\boldsymbol{\alpha})}= 0\,.
 \end{align}
Applying the symmetry on the Lax matrices provides the reduced Lax matrix $\check{L}^{\text{red}}(\lambda)$:
\begin{align}
    \left[\check{L}^{\text{red}}(\lambda)\right]_{1,1} = &    \frac{t_{0,0}}{\lambda} + \left(\frac{t_{0,0}}{Q_{\infty,0}}- P_{\infty,0} \right)  \lambda,  \qquad
    \left[\check{L}^{\text{red}}(\lambda)\right]_{1,2} = - t_{\infty,3} \left( Q_{\infty,0} + \lambda^2 \right),  \\
    \left[\check{L}^{\text{red}}(\lambda)\right]_{2,1} = & -t_{\infty,3} \lambda^2 + \frac{\left( \frac{t_{0,0}}{Q_{\infty,0}} - P_{\infty,0} \right)^2     }{t_{\infty,3}} 
     - 2 t_{\infty,1}  +  Q_{\infty,0} t_{\infty,3}, \qquad  \left[\check{L}^{\text{red}}(\lambda)\right]_{2,2} =  - \left[\check{L}^{\text{red}}(\lambda)\right]_{1,1}.  \nonumber
\end{align}
The reduced auxiliary matrix $\check{A}^{\text{red}}_{\boldsymbol{\alpha}}(\lambda)$ is given by 
\small{\begin{align}
  \left[\check{A}^{\text{red}}_{\boldsymbol{\alpha}}(\lambda)\right]_{1,1}=& \nu_{\infty,1}^{(\boldsymbol{\alpha})}   \left( \frac{t_{0,0}}{Q_{\infty,0} } - P_{\infty,0} \right) , \qquad \left[\check{A}^{\text{red}}_{\boldsymbol{\alpha}}(\lambda)\right]_{1,2} = - t_{\infty,3} \nu_{\infty,1}^{(\boldsymbol{\alpha})}  \lambda, \nonumber \\ 
  \left[\check{A}^{\text{red}}_{\boldsymbol{\alpha}}(\lambda)\right]_{2,1}=& - t_{\infty,3} \nu_{\infty,1}^{(\boldsymbol{\alpha})}  \lambda,  
  \qquad \left[\check{A}^{\text{red}}_{\boldsymbol{\alpha}}(\lambda)\right]_{2,2}=  - \left[\check{A}^{\text{red}}_{\boldsymbol{\alpha}}(\lambda)\right]_{1,1}.
\end{align}}
\normalsize{We} also have $\nu_{\infty,1}^{(\boldsymbol{\alpha})} = \frac{\alpha_{\infty,1}}{t_{\infty,3}}$, and thus a direct computation gives the Hamiltonian
\small{\beq
    \text{Ham}^{\text{red}}_{t_{\infty,1}}(\mathbf{Q}_\infty,\mathbf{P}_\infty) = \frac{1}{t_{\infty,3}} \check{H}^{\text{red}}_{\infty,0}=\frac{1}{t_{\infty,3}}\left(P_{\infty,0}-  (t_{\infty,3})^2 (Q_{\infty,0})^2   -  (P_{\infty,0})^2 Q_{\infty,0} +  2 t_{\infty,1} t_{\infty,3} Q_{\infty,0} -  \frac{t_{0,0} (1- t_{0,0})}{Q_{\infty,0}}\right).
\eeq}
\normalsize{Note} that the evolution equations are
\bea \frac{d}{d t_{\infty,1}}Q_{\infty,0}&=&\frac{1}{t_{\infty,3}}\left(1-2Q_{\infty,0}P_{\infty,0}\right),\cr
\frac{d}{d t_{\infty,1}}P_{\infty,0}&=&2t_{\infty,3} Q_{\infty,0}-2t_{\infty,1}+\frac{1}{t_{\infty,3}}\left((P_{\infty,0})^2+\frac{t_{0,0}(t_{0,0}-1)}{(Q_{\infty,0})^2}\right).
\eea
Let us define 
\beq y\coloneqq P_{\infty,0}- \frac{t_{0,0}}{Q_{\infty,0}}=P_d^{\text{FN}}.\eeq
Then we have
\beq \frac{d^2 y}{d t_{\infty,1}^{\,\,\,\,\,\,\,\,\,\,2}}= \frac{2}{(t_{\infty,3})^2}y^3-\frac{4t_{\infty,1}}{t_{\infty,3}}y -4t_{0,0},\eeq
so that if we define $T\coloneqq\left(-\frac{4}{t_{\infty,3}}\right)^{\frac{1}{3}} t_{\infty,1}$ and $u\coloneqq\frac{1}{4}\left(-\frac{4}{t_{\infty,3}}\right)^{\frac{2}{3}}y$, $u(T)$ satisfies the standard Painlev\'{e} II equation:
\beq \frac{d^2 u}{d T^2}=2u^3+T u-t_{0,0}.\eeq
Note that when $t_{\infty,3}=-4$, we have $4u=y=P_d^{\text{FN}}$ and $T=t_{\infty,1}$ recovering the FN PII Hamiltonian and equation, this is in agreement with Example $7.5$ of \cite{mazzocco2007hamiltonian}.

\subsection{The case $d=2$}
In this section, we will provide the Lax matrices for $d=2$ after the symmetry. Details of the computations before symmetry including Maple files  can be found at 
\href{https://math.univ-lyon1.fr/~marchal/AdditionalRessources/index.html}{O. M. website}. For $d=2$, the symmetry is equivalent to:
\begin{align}
 t_{\infty,0} = & t_{\infty,2} = t_{\infty,4} = 0 , \qquad  Q_{\infty,1} =  Q_{\infty,3} = 0, \qquad P_{\infty,1} =  P_{\infty,3} = 0, \nonumber \\
 &\nu_{\infty,0}^{(\boldsymbol{\alpha})}  =  \nu_{\infty,2}^{(\boldsymbol{\alpha})} = 0, \qquad \check{H}_{0,1} = \check{H}_{\infty,1} = 0.
\end{align}
The Lax matrix after reduction is given by
\begin{align}
    \left[\check{L}^{\text{red}}(\lambda)\right]_{1,1} = & \left(  \frac{t_{0,0}}{Q_{\infty,0}}-  P_{\infty,0}  \right)\lambda^3  + \left(  \frac{t_{0,0} Q_{\infty,2}}{Q_{\infty,0}} - P_{\infty,2} -P_{\infty,0} Q_{\infty,2}\right) \lambda + \frac{t_{0,0}}{\lambda}, \nonumber \\
    \left[\check{L}^{\text{red}}(\lambda)\right]_{1,2} = & - t_{\infty,5} \left( \lambda^4 + Q_{\infty,2} \lambda^2  + Q_{\infty,0} \right),  \nonumber \\
    \left[\check{L}^{\text{red}}(\lambda)\right]_{2,1} = &  - t_{\infty,5} \lambda^4 + \left( \frac{  \left( \frac{t_{0,0}}{Q_{\infty,0}} - P_{\infty,0} \right)^2  }{ t_{\infty,5}} - 2 t_{\infty,3}  + t_{\infty,5} Q_{\infty,2}  \right) \lambda^2 - 2 t_{\infty,1} - \frac{ (t_{\infty,3})^2}{ t_{\infty,5}} + 2 t_{\infty,3} Q_{\infty,2} \nonumber \\&
    -t_{\infty,5}(Q_{\infty,2})^2+t_{\infty,5}Q_{\infty,0}+\frac{Q_{\infty,2}}{t_{\infty,5}}\left(P_{\infty,0}-\frac{t_{0,0}}{Q_{\infty,0}}\right)^2+\frac{2P_{\infty,2}}{t_{\infty,5}}\left(P_{\infty,0}-\frac{t_{0,0}}{Q_{\infty,0}}\right),
     \nonumber \\
     \left[\check{L}^{\text{red}}(\lambda)\right]_{2,2} = & -  \left[\check{L}^{\text{red}}(\lambda)\right]_{1,1}.  
\end{align}
The reduced auxiliary matrix $\check{A}^{\text{red}}_{\boldsymbol{\alpha}}(\lambda)$ is given by 
\begin{align}
      \left[\check{A}^{\text{red}}_{\boldsymbol{\alpha}}(\lambda)\right]_{1,1}=& \nu_{\infty,1}^{(\boldsymbol{\alpha})}  \left( \frac{t_{0,0}}{Q_{\infty,0}} - P_{\infty,0} \right) \lambda^2 + \nu_{\infty,3}^{(\boldsymbol{\alpha})}  \left( \frac{t_{0,0}}{Q_{\infty,0}} - P_{\infty,0} \right) -  \nu_{\infty,1}^{(\boldsymbol{\alpha})} \left( P_{\infty,2} + P_{\infty,0} Q_{\infty,2} - \frac{ t_{0,0} Q_{\infty,2} }{Q_{\infty,0}}   \right), \nonumber \\
       \left[\check{A}^{\text{red}}_{\boldsymbol{\alpha}}(\lambda)\right]_{1,2}=& -t_{\infty,5}\left( \nu_{\infty,1}^{(\boldsymbol{\alpha})}  \lambda^3 + \left( \nu_{\infty,1}^{(\boldsymbol{\alpha})}  Q_{\infty,2} + \nu_{\infty,3}^{(\boldsymbol{\alpha})}  \right) \lambda\right),  \nonumber \\
        \left[\check{A}^{\text{red}}_{\boldsymbol{\alpha}}(\lambda)\right]_{2,1}=& - t_{\infty,5} \nu_{\infty,1}^{(\boldsymbol{\alpha})}  \lambda^3  +  \left(  \frac{1}{t_{\infty,5}}\left(  \frac{t_{0,0}}{Q_{\infty,0}} - P_{\infty,0} \right)^2 \nu_{\infty,1}^{(\boldsymbol{\alpha})} +t_{\infty,5} \nu_{\infty,1}^{(\boldsymbol{\alpha})}  Q_{\infty,2} - 2 t_{\infty,3}\nu_{\infty,1}^{(\boldsymbol{\alpha})}-t_{\infty,5} \nu_{\infty,3}^{(\boldsymbol{\alpha})} \right) \lambda,   \nonumber \\
         \left[\check{A}^{\text{red}}_{\boldsymbol{\alpha}}(\lambda)\right]_{2,2}=&  - \left[\check{A}^{\text{red}}_{\boldsymbol{\alpha}}(\lambda)\right]_{1,1}. 
 \end{align}
It is also straightforward to check that the reduced vectors take the form 
\begin{align}
    \nu_{\infty,1}^{(\boldsymbol{\alpha})} = \frac{\alpha_{\infty,3}}{3 t_{\infty,5}},  \qquad \nu_{\infty,3}^{(\boldsymbol{\alpha})} = \frac{\alpha_{\infty,1}}{t_{\infty,5}} - \frac{t_{\infty,3}   \alpha_{\infty,3} }{3(t_{\infty,5})^2},
\end{align}
and the Hamiltonians are
\bea \text{Ham}^{\text{red}}_{t_{\infty,1}}&=&\frac{t_{0,0}(t_{0,0}-1)}{t_{\infty,5}Q_{\infty,0}}+t_{\infty,5}(Q_{\infty,2})^3-2t_{\infty,3}(Q_{\infty,2})^2+\left(\frac{(t_{\infty,3})^2}{t_{\infty,5}}+2t_{\infty,1}\right)Q_{\infty,2}\cr
&&+\frac{1}{t_{\infty,5}}(P_{\infty,2})^2+\frac{1}{t_{\infty,5}}P_{\infty,0}+\left(2t_{\infty,3}-2t_{\infty,5}Q_{\infty,2}-\frac{(P_{\infty,0})^2}{t_{\infty,5}}\right)Q_{\infty,0},\cr
\text{Ham}^{\text{red}}_{t_{\infty,3}}&=&\frac{t_{0,0}(t_{0,0}-1)(t_{\infty,5}Q_{\infty,2}-t_{\infty,3})}{3(t_{\infty,5})^2Q_{\infty,0}} -\frac{t_{\infty,5}}{3}(Q_{\infty,0})^2-\frac{Q_{\infty,0}(P_{\infty,0})^2}{3(t_{\infty,5})^2}(t_{\infty,5}Q_{\infty,2}-t_{\infty,3})\cr&&
-\frac{t_{\infty,3}}{3(t_{\infty,5})^2}(P_{\infty,2})^2-\frac{t_{\infty,3}}{3}(Q_{\infty,2})^3+\frac{(Q_{\infty,2})^2}{3t_{\infty,5}}\left(2(t_{\infty,3})^2+(t_{\infty,5})^2Q_{\infty,0}\right)\cr&&
+\frac{Q_{\infty,0}}{3t_{\infty,5}}\left(2t_{\infty,1}t_{\infty,5}-(t_{\infty,3})^2-2P_{\infty,0}P_{\infty,2}\right)+\frac{P_{\infty,2}}{t_{\infty,5}}+\frac{P_{\infty,0}}{3(t_{\infty,5})^2}\left(t_{\infty,5}Q_{\infty,2}-t_{\infty,3}\right)\cr&&
-\frac{t_{\infty,3}Q_{\infty,2}\left(2t_{\infty,1}t_{\infty,5}+(t_{\infty,3})^2\right)}{3(t_{\infty,5})^2}.
\eea

\section{Conclusion and outlook} \label{secConclusion}
In this article, we considered the symmetry reduction of the even PIV hierarchy to the FN PII hierarchy. Building on the symmetry reduction of \cite{alameddine2026symmetryreductionpainleveiv}, we have diagonalized the symmetry representation enforced on the Lax matrices and constructed a set of Darboux coordinates that is compatible with the symmetry. This led to the full Hamiltonian description of the FN PII hierarchy with explicit formulas for the coordinate-dependent Lax matrices. One of the main results of the article is that when dealing with a symmetry, it is crucial to select a suitable trivialization for the Hamiltonian structure so that the maximum number of Darboux coordinates vanish by symmetry. This ensures that the reduced Hamiltonians can be extracted directly from the Hamiltonians before symmetry. The existence of such compatible Darboux coordinates was left open in \cite{alameddine2026symmetryreductionpainleveiv} and we have proven both their existence and the efficiency of the strategy in the present work. This enables us to obtain explicit formulas for the Hamiltonians for all times (and not only $z=t_{\infty,1}$ as in \cite{mazzocco2007hamiltonian}) which was an important missing point for the mKdV hierarchy. In particular, the general form of the Hamiltonians (a lower triangular Toeplitz matrix of times multiplied by some spectral invariants) is preserved for the FN PII hierarchy.

\medskip

Another important consequence of the present article is that it shows that the general results of \cite{marchal2024hamiltonianrepresentationisomonodromicdeformations} can be adapted when dealing with additional symmetries. In this case, one may need to select a compatible representative in the orbit that ensures an easier reduction after the symmetry. For example, the even/odd parity in this work underlies the main idea behind the conjugation action by the matrix $S$ of \autoref{DefGauges}, and one realizes in \autoref{secSymmetry} how this is an advantage when reducing the Hamiltonian structure. Our method is applicable beyond the present setup and we list several possibilities:
\begin{itemize}
    \item In this paper, the symmetry is of the form $\Psi(-\lambda)=\sigma \Psi(\lambda) \sigma^{-1}$. Since $\lambda\to -\lambda$ is an involution, it necessarily implies that $\sigma^2=I_2$. Thus $\sigma$ is always diagonalizable and thus, defining $S$ a matrix of eigenvectors of $\sigma$ we may always go into the gauge where $\sigma=\text{diag}(\alpha_1,\alpha_2)$. Since the symmetry is invariant under the multiplication of $\sigma$ by a constant, and because the eigenvalues satisfy $\mu^2=1$, we may always choose a gauge where  $\sigma=\text{diag}(-1,1)=\sigma_3$ without loss of generality. Thus, the present paper addresses the only possible case for symmetries of the form $\Psi(-\lambda)=\sigma \Psi(\lambda) \sigma^{-1}$. Such a symmetry puts constraints on the pole structure that must be invariant under $\lambda\to -\lambda$, and one can verify that the only possible poles are $\{0,\infty\}$. The present paper deals with the case of a regular pole at zero and an irregular pole at infinity, but it is not the only possible case. In this direction, the most general case with untwisted poles corresponds to two irregular poles at zero and infinity which deserves further investigation. In this setup, the only remaining work is to adapt the results of \autoref{SecP4hierarchy} to the case of an irregular pole at zero using \cite{marchal2024hamiltonianrepresentationisomonodromicdeformations}. This includes several models arising in mathematical physics such as the Sine-Gordon model \cite{Newell1985}, the Derivative NLS model (symmetry reduction of the Kaup–Newell hierarchy \cite{Pashaev:2017wqv}). 
    \item Another possible case is to keep the symmetry $\Psi(-\lambda)=\sigma \Psi(\lambda) \sigma^{-1}$ but to have twisted poles at zero or infinity. In this case, one could adapt the results of \cite{MarchalP1Hierarchy} regarding the Painlev\'{e} I hierarchy (one twisted pole at infinity) to get the Hamiltonian description before symmetry when two poles are involved.
    \item The symmetry considered here is of the form $\Psi(-\lambda)=\sigma \Psi(\lambda) \sigma^{-1}$. A natural generalization is to consider action of the dihedral group \cite{Lombardo}, i.e., $\Psi(\omega_N\lambda)=\sigma \Psi(\lambda) \sigma^{-1}$ with $\omega_N=e^{\frac{2i\pi}{N}}$. In this case, one gets $\sigma^N=I_2$ and the only possible poles are again only $\{0,\infty\}$. In this case, $\sigma$ is diagonalizable and one should get to the gauge where $\sigma=\diag(\rho,1)$ with $\rho$ a $N$th root of unity. We believe that the present strategy can be applied to this case.
    \item The present paper is limited to $\mathfrak{sl}_2(\mathbb{C})$ because explicit formulas for the Hamiltonian and Lax matrices are only available in this dimension from \cite{marchal2024hamiltonianrepresentationisomonodromicdeformations}. However, examples in $\mathfrak{sl}_3(\mathbb{C})$ have been dealt with a similar method in \cite{Alameddine:2024xjs}. In higher dimensions, the symmetry $\Psi(-\lambda)=\sigma \Psi(\lambda) \sigma^{-1}$ offers more flexibility regarding the eigenvalues of $\sigma$ since one can choose the number of $-1$ eigenvalues. Provided that explicit formulas can be obtained for $\mathfrak{sl}_n(\mathbb{C})$, we believe that the present strategy can be applied. For example, this could be applied to study higher-dimensional Lax representation of the reductions of the mKdV hierarchy.
\end{itemize}

\appendix
\renewcommand{\subsectionautorefname}{Appendix}
\renewcommand{\sectionautorefname}{Appendix}
\renewcommand{\theequation}{\thesection-\arabic{equation}}

\section{Expression of the auxiliary matrices $\td{A}_{\text{oper}, \boldsymbol{\alpha}}(\lambda)$ and $\td{A}_{\boldsymbol{\alpha}}(\lambda)$}
\subsection{Expression of $\td{A}_{\text{oper}, \boldsymbol{\alpha}}(\lambda)$ in the oper Darboux coordinates}\label{AppendixAuxiliaryMatrixOperGauge}
From \cite{marchal2024hamiltonianrepresentationisomonodromicdeformations}, the entries of the matrix $\td{A}_{\text{oper},\boldsymbol{\alpha}}(\lambda)$ are given by
\bea \left[\td{A}_{\text{oper}, \boldsymbol{\alpha}}(\lambda)\right]_{1,2}&=&\nu_{\infty,0}^{(\boldsymbol{\alpha})}+\sum_{j=1}^{2d} \frac{\td{\mu}_{j}^{(\boldsymbol{\alpha})}}{\lambda-u_j},\cr
\left[\td{A}_{\text{oper}, \boldsymbol{\alpha}}(\lambda)\right]_{1,1}&=&\frac{1}{2\omega}\mathcal{L}_{\boldsymbol{\alpha}}[\omega]- \sum_{j=1}^{2d}\frac{v_j \td{\mu}_j^{(\boldsymbol{\alpha})}}{\lambda-u_j},\cr
\left[\td{A}_{\text{oper}, \boldsymbol{\alpha}}(\lambda)\right]_{2,1}&=& \partial_{\lambda} \left[\td{A}_{\boldsymbol{\alpha}, \text{oper}}(\lambda)\right]_{1,1}+\left[\td{A}_{\text{oper}, \boldsymbol{\alpha}}(\lambda)\right]_{1,2}\left[\td{L}_{\text{oper}}(\lambda)\right]_{2,1},\cr
\left[\td{A}_{\text{oper}, \boldsymbol{\alpha}}(\lambda)\right]_{2,2}&=& \partial_{\lambda} \left[\td{A}_{\text{oper}, \boldsymbol{\alpha}}(\lambda)\right]_{1,2}+\left[\td{A}_{\text{oper}, \boldsymbol{\alpha}}(\lambda)\right]_{1,1}+\left[\td{A}_{\text{oper}, \boldsymbol{\alpha}}(\lambda)\right]_{1,2}\left[\td{L}_{\text{oper}}(\lambda)\right]_{2,2},\cr&&
\eea
where the coefficients $\left(\td{\mu}_1^{(\boldsymbol{\alpha})},\dots, \td{\mu}_{2d}^{(\boldsymbol{\alpha})}\right)$ are given by
\beq \begin{pmatrix}
\frac{1}{u_1}&\frac{1}{u_2}& \dots&\dots& \frac{1}{u_{2d}}\\
1&1 &\dots &\dots &1\\
u_1& u_2&\dots &\dots& u_{2d}\\
\vdots & & & & \vdots\\
\vdots & & & & \vdots\\
u_1^{2d-2}& u_2^{2d-2} &\dots & \dots& u_{2d}^{2d-2}\\
\end{pmatrix}\begin{pmatrix} \td{\mu}^{(\boldsymbol{\alpha})}_1\\ \vdots\\\vdots\\ \td{\mu}^{(\boldsymbol{\alpha})}_{2d}\end{pmatrix}=\begin{pmatrix} \nu^{(\boldsymbol{\alpha})}_{\infty,0}\\\nu^{(\boldsymbol{\alpha})}_{\infty,1}\\ \vdots\\ \nu^{(\boldsymbol{\alpha})}_{\infty,2d-1}\end{pmatrix}.
\eeq

\subsection{Expression of $\td{A}_{\boldsymbol{\alpha}}(\lambda)$ in the geometric Darboux coordinates}\label{AppendixAuxiliaryMatrixGeometricDarboux}
Entries of the matrix $\td{A}_{\boldsymbol{\alpha}}(\lambda)$ are given by

\begin{align}
    \left[\td{A}_{\boldsymbol{\alpha}}(\lambda)\right]_{1,2}=&
\omega\left(\sum_{i=0}^{2d-1}\nu_{\infty,i}^{(\boldsymbol{\alpha})}\lambda^{2d-1-i} +\sum_{k=0}^{2d-2}\sum_{i=0}^{k}\nu_{\infty,i}^{(\boldsymbol{\alpha})}U_{\infty,k} \lambda^{k-i}\right), \\
\left[\td{A}_{\boldsymbol{\alpha}}(\lambda)\right]_{1,1}=&\frac{1}{2\omega}\mathcal{L}_{\boldsymbol{\alpha}}[\omega]+t_{\infty,2d+1}\left(\nu_{\infty,0}^{(\boldsymbol{\alpha})}U_{0,1}+ \sum_{j=1}^{2d-1}\nu_{\infty,j}^{(\boldsymbol{\alpha})}U_{\infty,j-1}\right)+ \left[\td{L}_{1,1}(\lambda) \left(\sum_{i=0}^{2d-1} \nu_{\infty,i}\lambda^{-i}\right)\right]_{\infty,+}. \nonumber
\end{align}

One also has $ \left[\td{A}_{\boldsymbol{\alpha}}(\lambda)\right]_{2,2}=-\left[\td{A}_{\boldsymbol{\alpha}}(\lambda)\right]_{1,1}$ and finally
\bea \left[\td{A}_{\boldsymbol{\alpha}}(\lambda)\right]_{2,1}&=& -\frac{t_{\infty,2d+1}\mathcal{L}_{\boldsymbol{\alpha}}[\omega]}{\omega^2}\lambda -\frac{(t_{\infty, 2d}-t_{\infty,2d+1}U_{\infty,2d-2})\mathcal{L}_{\boldsymbol{\alpha}}[\omega]}{\omega^2}+ \frac{\alpha_{\infty,2d}-t_{\infty,2d+1}\mathcal{L}_{\boldsymbol{\alpha}}[U_{\infty,2d-2}]}{\omega} \cr
&&+\left[
\frac{ \underset{i=0}{\overset{2d+1}{\sum}}\underset{j=2d-1}{\overset{4d}{\sum}}\underset{m=0}{\overset{4d-j}{\sum}}t_{\infty,2d+1-m}t_{\infty,j+m-2d+1}\nu_{\infty,i}^{(\boldsymbol{\alpha}) }\lambda^{j-i}  }{\td{L}_{1,2}(\lambda)}\right]_{\infty,+} -\frac{1}{\omega}\nu_{\infty,0}^{(\boldsymbol{\alpha}) }  \cr
&&+ \left[\frac{\td{L}_{1,1}(\lambda)}{\td{L}_{1,2}(\lambda)}\left(\td{L}_{1,1}(\lambda)\frac{\left[\td{A}_{\boldsymbol{\alpha}}(\lambda)\right]_{1,2}}{\td{L}_{1,2}(\lambda)}-2\left[\td{A}_{\boldsymbol{\alpha}}(\lambda)\right]_{1,1}\right)\right]_{\infty,+},
\eea
with $\nu_{\infty,2d}^{(\boldsymbol{\alpha})}\coloneqq -\underset{j=1}{\overset{2d-1}{\sum}}\nu_{\infty,j}^{(\boldsymbol{\alpha})} U_{\infty,j-1} -\nu_{\infty,0}^{(\boldsymbol{\alpha})} U_{0,1}$ and 
$\nu_{\infty,2d+1}^{(\boldsymbol{\alpha})}\coloneqq
-\underset{j=2}{\overset{2d}{\sum}}\nu_{\infty,j}^{(\boldsymbol{\alpha})}U_{\infty,j-2} -\nu_{\infty,1}^{(\boldsymbol{\alpha})}U_{0,1}$.

\section{Asymptotics of the wave matrix $\check{\Psi}(\lambda)$ and proof of \autoref{PropLopersym}}\label{AppendixAsymptPsi}
Using the fact that $\det \check{\Psi}(\lambda)=\det \td{\Psi}(\lambda)$ is constant and can be set to $1$ by trivial normalization, we have:
\beq \check{L}_{1,2}(\lambda)=\check{\Psi}_{1,1}(\lambda) \partial_\lambda \check{\Psi}_{1,2}(\lambda) -\check{\Psi}_{1,2}(\lambda)\partial_\lambda \check{\Psi}_{1,1}(\lambda).\eeq
Thus, since $\check{L}_{1,2}(\lambda)$ is regular at $\lambda=0$ from \autoref{ChoiceRepresentative}, we have that if we denote
\bea \check{\Psi}_{1,1}(\lambda)&\overset{\lambda\to 0}{=}&\exp\left(\alpha_{1,1}\ln(\lambda)+ A_{0,0}+O\left(\lambda\right)\right),\cr
\check{\Psi}_{1,2}(\lambda)&\overset{\lambda\to 0}{=}&\exp\left(\alpha_{1,2}\ln(\lambda)+ B_{0,0}+O\left(\lambda\right)\right), \eea
then at $\lambda=0$ we have
\bea O(1)&=&\exp\left(\alpha_{1,1}\ln(\lambda)+ A_{0,0}+O\left(\lambda\right)\right)\exp\left(\alpha_{1,2}\ln(\lambda)+ B_{0,0}+O\left(\lambda\right)\right)\left(\frac{\alpha_{1,2}}{\lambda} +O(1)\right) \cr&&-\exp\left(\alpha_{1,2}\ln(\lambda)+ B_{0,0}+O\left(\lambda\right)\right)\exp\left(\alpha_{1,1}\ln(\lambda)+ A_{0,0}+O\left(\lambda\right)\right)\left(\frac{\alpha_{1,1}}{\lambda} +O(1)\right).
\eea
The r.h.s. is thus
\beq (\alpha_{1,2}-\alpha_{1,1}) \lambda^{\alpha_{1,1}+\alpha_{1,2} -1}e^{A_{0,0}+B_{0,0}} + O(\lambda^{\alpha_{1,1}+\alpha_{1,2}}).\eeq
In order to obtain $O(\lambda^0)$, the two possible solutions are either $\alpha_{1,2}=\alpha_{1,1}$ which is excluded because the two solutions of the quantum curve would be the same or $\alpha_{1,1}+\alpha_{1,2}=1$, i.e., $\alpha_{1,2}=1-\alpha_{1,1}$. From the fact that the first column is unchanged from $\td{\Psi}$, we get that
\bea \label{PsiAsymptotics}\check{\Psi}_{1,1}(\lambda)&\overset{\lambda\to 0}{=}&\exp\left(t_{0,0}\ln(\lambda)+ A_{0,0}+O\left(\lambda\right)\right),\cr
\check{\Psi}_{1,2}(\lambda)&\overset{\lambda\to 0}{=}&\exp\left((1-t_{0,0})\ln(\lambda)+ B_{0,0}+O\left(\lambda\right)\right).
\eea
A similar argument can be written at infinity. Starting from
\bea 
 \check{\Psi}_{1,1}(\lambda)&\overset{\lambda\to \infty}{=}&\exp\left(-\sum_{k=1}^{2d+1} \frac{t_{\infty,k}}{k}\lambda^k +\alpha_{1,1}\ln \lambda+ A_{\infty,0} +O\left(\lambda^{-1}\right)\right),\cr
 \check{\Psi}_{1,2}(\lambda)&\overset{\lambda\to \infty}{=}&\exp\left(\sum_{k=1}^{2d+1} \frac{t_{\infty,k}}{k}\lambda^k +\alpha_{1,2}\ln \lambda+ B_{\infty,0} +O\left(\lambda^{-1}\right)\right),
 \eea 
 we get that
 \beq t_{\infty,2d+1}\lambda^{2d}+O(\lambda^{2d-1})= \check{\Psi}_{1,1}(\lambda) \partial_\lambda \check{\Psi}_{1,2}(\lambda) -\check{\Psi}_{1,2}(\lambda)\partial_\lambda \check{\Psi}_{1,1}(\lambda).\eeq
 The r.h.s. is given by
 \bea &&\left(\sum_{k=1}^{2d+1} t_{\infty,k}k \lambda^{k-1} +\frac{\alpha_{1,2}}{\lambda} +O(\lambda^{-2})\right)\exp\left((\alpha_{1,1}+\alpha_{1,2})\ln \lambda+ A_{\infty,0}+B_{\infty,0} +O\left(\lambda^{-1}\right)\right)\cr
 &&- \left(-\sum_{k=1}^{2d+1} t_{\infty,k}k \lambda^{k-1} +\frac{\alpha_{1,1}}{\lambda} +O(\lambda^{-2})\right)\exp\left((\alpha_{1,1}+\alpha_{1,2})\ln \lambda+ A_{\infty,0}+B_{\infty,0} +O\left(\lambda^{-1}\right)\right).\cr&&
 \eea
Thus we must have $\alpha_{1,1}+\alpha_{1,2}=0$ and because of the monodromy, we must have $\alpha_{1,1}=-t_{\infty,0}$ so that 
\bea \label{PsiAsymptotics2}
 \check{\Psi}_{1,1}(\lambda)&\overset{\lambda\to \infty}{=}&\exp\left(-\sum_{k=1}^{2d+1} \frac{t_{\infty,k}}{k}\lambda^k -t_{\infty,0}\ln \lambda+ A_{\infty,0} +O\left(\lambda^{-1}\right)\right),\cr
 \check{\Psi}_{1,2}(\lambda)&\overset{\lambda\to \infty}{=}&\exp\left(\sum_{k=1}^{2d+1} \frac{t_{\infty,k}}{k}\lambda^k +t_{\infty,0}\ln \lambda+ B_{\infty,0} +O\left(\lambda^{-1}\right)\right).
 \eea

 Let us now compute the asymptotic expansion of $[\check{L}_{\text{oper}}(\lambda)]_{2,1}$ at each pole. Since $\check{L}_{\text{oper}}(\lambda)$ is the companion matrix attached to $\check{\Psi}(\lambda)$, it is a straightforward computation to show that
\bea\label{EntriesInTermsOfYi}
\left[\check{L}_{\text{oper}}(\lambda)\right]_{2,2}&=& \frac{\partial_\lambda \check{W}(\lambda)}{\check{W}(\lambda)},\cr
\left[\check{L}_{\text{oper}}(\lambda)\right]_{2,1}&=&- \check{Y}_1(\lambda) \,  \check{Y}_2(\lambda) +  \frac{\check{Y}_2(\lambda) \, \partial_\lambda \check{Y}_1(\lambda) - \check{Y}_1(\lambda) \, \partial_\lambda \check{Y}_2(\lambda) }{\check{Y}_2(\lambda)-\check{Y}_1(\lambda)},
\eea
where we have defined $\check{Y}_j(\lambda)\coloneqq \frac{1}{\check{\Psi}_{1,j}(\lambda)} \frac{\partial \check{\Psi}_{1,j}(\lambda)}{\partial \lambda}$ for $j\in\{1,2\}$ and 
\beq \check{W}(\lambda)\coloneqq\check{\Psi}_{1,1}(\lambda) \partial_\lambda \check{\Psi}_{1,2}(\lambda) -\check{\Psi}_{1,2}(\lambda) \partial_\lambda\check{\Psi}_{1,1}(\lambda)=\det \check{\Psi}_{\text{oper}}(\lambda).\eeq
Using the relation between $\left[\check{L}_{\text{oper}}(\lambda)\right]_{2,2}$ and the Wronskian, we get:
\beq \left[\check{L}_{\text{oper}}(\lambda)\right]_{2,2}=\sum_{j=1}^{2d} \frac{1}{\lambda-q_j}.\eeq
Then, using \eqref{PsiAsymptotics} and \eqref{PsiAsymptotics2} we have:
\bea\label{Y1Y2} \check{Y}_1(\lambda)&\overset{\lambda \to \infty}{=}&-\sum_{k=1}^{2d+1} t_{\infty,k}\lambda^{k-1} -\frac{t_{\infty,0}}{\lambda}+ O\left(\lambda^{-2}\right), \cr
\check{Y}_2(\lambda)&\overset{\lambda \to \infty}{=}&\sum_{k=1}^{2d+1} t_{\infty,k}\lambda^{k-1} +\frac{t_{\infty,0}}{\lambda} + O\left(\lambda^{-2}\right), \cr
\check{Y}_1(\lambda)&\overset{\lambda \to 0}{=}& \frac{t_{0,0}}{\lambda}+ O\left(\lambda^{-2}\right),\cr
\check{Y}_2(\lambda)&\overset{\lambda \to 0}{=}& \frac{1-t_{0,0}}{\lambda}+ O\left(\lambda^{-2}\right).
\eea
Thus, we get:
\bea \frac{\check{Y}_2(\lambda) \, \partial_\lambda \check{Y}_1(\lambda) - \check{Y}_1(\lambda) \, \partial_\lambda \check{Y}_2(\lambda) }{\check{Y}_2(\lambda)-\check{Y}_1(\lambda)}&\overset{\lambda \to \infty}{=}&O\left(\lambda^{2d-2}\right),\cr
\frac{\check{Y}_2(\lambda) \, \partial_\lambda \check{Y}_1(\lambda) - \check{Y}_1(\lambda) \, \partial_\lambda \check{Y}_2(\lambda) }{\check{Y}_2(\lambda)-\check{Y}_1(\lambda)}&\overset{\lambda \to 0}{=}&O\left(\lambda^{-1}\right),
\eea
from which we immediately obtain using \eqref{EntriesInTermsOfYi} that
\bea
\left[\check{L}_{\text{oper}}(\lambda)\right]_{2,1}&\overset{\lambda \to \infty}{=}&\underset{j= 2 d-1}{\overset{4d}{\sum}}\left(\underset{m=0}{\overset{4d-j}{\sum}} t_{\infty,2 d +1-m}t_{\infty,j+m-2d+1}\right) \lambda^{j} + O\left(\lambda^{2d-2}\right),\cr
\left[\check{L}_{\text{oper}}(\lambda)\right]_{2,1}&\overset{\lambda \to 0}{=}&-\frac{t_{0,0}(1-t_{0,0})}{\lambda^2} + O\left(\lambda^{-1}\right).
\eea

\section{Expression of the auxiliary matrix $\check{A}_{\text{oper},\boldsymbol{\alpha}}(\lambda)$}

\subsection{Asymptotics of $\left[\check{A}_{\text{oper},\boldsymbol{\alpha}}(\lambda)\right]_{1,2}$ at zero and infinity}\label{AppendixExpansionA}
We shall observe that the first line of $\check{A}_{\text{oper},\boldsymbol{\alpha}}(\lambda)$ is given by
\bea \label{ExpressionA12A11}
\left[\check{A}_{\text{oper},\boldsymbol{\alpha}}(\lambda)\right]_{1,2}&=&\frac{\check{W}_{\boldsymbol{\alpha}}(\lambda)}{\check{W}(\lambda)}= \frac{\check{Z}_{\boldsymbol{\alpha},2}(\lambda)-\check{Z}_{\boldsymbol{\alpha},1}(\lambda) }{\check{Y}_2(\lambda)-\check{Y}_1(\lambda)},\cr
\left[\check{A}_{\text{oper},\boldsymbol{\alpha}}(\lambda)\right]_{1,1}&=& \frac{\check{Z}_{\boldsymbol{\alpha},1}(\lambda) \check{Y}_2(\lambda) -\check{Z}_{\boldsymbol{\alpha},2}(\lambda) \check{Y}_1(\lambda)}{\check{Y}_2(\lambda) -\check{Y}_1(\lambda)},
\eea
where we have defined
\bea 
\check{Z}_{\boldsymbol{\alpha},i}(\lambda)&\coloneqq& \frac{\mathcal{L}_{\boldsymbol{\alpha}}[\check{\Psi}_{1,j}(\lambda)]}{\check{\Psi}_{1,j}(\lambda)} \,\,,\,\, \forall \, j\in\llbracket 1,2\rrbracket,\cr
\check{W}_{\boldsymbol{\alpha}}(\lambda)&=& \mathcal{L}_{\boldsymbol{\alpha}}[\check{\Psi}_{1,2}(\lambda)] \check{\Psi}_{1,1}(\lambda) - \mathcal{L}_{\boldsymbol{\alpha}}[\check{\Psi}_{1,1}(\lambda)] \check{\Psi}_{1,2}(\lambda).
\eea 
Using \eqref{PsiAsymptotics} and \eqref{PsiAsymptotics2} and the fact that $t_{\infty,2d+1}$ is fixed, we get that $\mathcal{L}_{\boldsymbol{\alpha}}[\check{\Psi}_{1,j}(\lambda)]$ has the following local expansions:
\bea \label{LPsi}  
\mathcal{L}_{\boldsymbol{\alpha}}[\check{\Psi}_{1,1}(\lambda)]&\overset{\lambda\to \infty}{=}&-\left(\sum_{k=1}^{2d} \frac{\alpha_{\infty,k}}{k}\lambda^k +O(1)\right)\exp\left(-\sum_{k=1}^{2d+1} \frac{t_{\infty,k}}{k}\lambda^k -t_{\infty,0}\ln \lambda+ A_{\infty,0} +O\left(\lambda^{-1}\right)\right),\cr
\mathcal{L}_{\boldsymbol{\alpha}}[\check{\Psi}_{1,2}(\lambda)]&\overset{\lambda\to \infty}{=}&\left(\sum_{k=1}^{2d} \frac{\alpha_{\infty,k}}{k}\lambda^k +O(1)\right)\exp\left(\sum_{k=1}^{2d+1} \frac{t_{\infty,k}}{k}\lambda^k +t_{\infty,0}\ln \lambda+ B_{\infty,0} +O\left(\lambda^{-1}\right)\right),\cr
\mathcal{L}_{\boldsymbol{\alpha}}[\check{\Psi}_{1,1}(\lambda)]&\overset{\lambda\to 0}{=}&\left(O(1)\right)\exp\left(t_{0,0}\ln \lambda+ A_{0,0}+O\left(\lambda\right)\right),\cr
\mathcal{L}_{\boldsymbol{\alpha}}[\check{\Psi}_{1,2}(\lambda)]&\overset{\lambda\to 0}{=}&\left(O(1)\right)\exp\left((1-t_{0,0})\ln \lambda+ B_{0,0}+O\left(\lambda\right)\right).
\eea
Thus, from \eqref{ExpressionA12A11} and the asymptotic expansions \eqref{Y1Y2} and \eqref{LPsi}, we deduce that
\bea\label{Asympttt} \left[\check{A}_{\text{oper},\boldsymbol{\alpha}}(\lambda)\right]_{1,2}&\overset{\lambda\to \infty}{=}&\sum_{i=0}^{2d-1} \frac{\nu^{(\boldsymbol{\alpha})}_{\infty,i}}{\lambda^i} +O\left(\lambda^{-2d}\right),\cr
\left[\check{A}_{\text{oper},\boldsymbol{\alpha}}(\lambda)\right]_{1,2}&\overset{\lambda\to 0}{=}& O(\lambda),
\eea
where $\left(\nu^{(\boldsymbol{\alpha})}_{\infty,i}\right)_{0\leq i\leq 2d-1}$ are defined recursively by
\bea \label{Defnuinftyk}
\nu^{(\boldsymbol{\alpha})}_{\infty,0}&=&\frac{\alpha_{\infty,2d}}{(2d)t_{\infty,2d+1}},\cr
\nu^{(\boldsymbol{\alpha})}_{\infty, 2d-k}
&=&\frac{1}{t_{\infty,2d+1}}\left(\frac{\alpha_{\infty,k}}{k}-\sum_{i=0}^{2d-1-k}t_{\infty,k+i+1}\nu^{(\boldsymbol{\alpha})}_{\infty,i}\right)\,\,,\forall \, k\in \llbracket 1,2d-1\rrbracket.
\eea
In particular, equations \eqref{Defnuinftyk}  may be rewritten in a matrix form giving Proposition \ref{PropAsymptoticExpansionA12}.

\subsection{Proof of \autoref{PropA12Form}}\label{ProofEntryA12}
Since $\left[\check{A}_{\text{oper},\boldsymbol{\alpha}}(\lambda)\right]_{1,2}$ is a rational function of $\lambda$ with only possible poles at infinity, zero or at $\left(q_k\right)_{1\leq k\leq 2d}$, the asymptotic expansion at each pole provided by \autoref{PropAsymptoticExpansionA12} implies that
\beq \left[\check{A}_{\text{oper},\boldsymbol{\alpha}}(\lambda)\right]_{1,2}=\nu^{(\boldsymbol{\alpha})}_{\infty,0} + \sum_{j=1}^{2d} \frac{\check{\mu}^{(\boldsymbol{\alpha})}_j}{\lambda-q_j}.\eeq
One then observes that the coefficients $\left(\check{\mu}^{(\boldsymbol{\alpha})}_j\right)_{1\leq j\leq 2d}$ are related to $\left(\nu^{(\boldsymbol{\alpha})}_{\infty,i}\right)_{0\leq i\leq 2d-1}$ through \eqref{Defnuinftyk}. We thus get
\beq \left[\check{A}_{\text{oper},\boldsymbol{\alpha}}(\lambda)\right]_{1,2}=\nu^{(\boldsymbol{\alpha})}_{\infty,0} + \sum_{j=1}^{2d} \frac{\check{\mu}^{(\boldsymbol{\alpha})}_j}{\lambda-q_j}
\overset{\lambda\to \infty}{=}\nu^{(\boldsymbol{\alpha})}_{\infty,0}+\sum_{k=1}^{\infty} \sum_{j=1}^{2d}\frac{\check{\mu}^{(\boldsymbol{\alpha})}_j q_j^{k-1}}{\lambda^{k}},
\eeq
so that
\beq \label{RelationNuMu} \forall\, k\in \llbracket 1, 2d-1\rrbracket\,:\, \sum_{j=1}^{2d} \check{\mu}^{(\boldsymbol{\alpha})}_j q_j^{k-1}=\nu^{(\boldsymbol{\alpha})}_{\infty,k}.
\eeq
The former relations may be rewritten using a $(2d-1)\times (2d)$ matrix $V_\infty(\mathbf{q})$ corresponding to the last $2d-1$ lines of the matrix $\mathbf{V}(\mathbf{q})$ as given in \autoref{PropA12Form}. The first line is given by the expansion at $\lambda=0$ provided by \eqref{Asympttt}:
\beq \label{Eqbis}\left[\check{A}_{\text{oper},\boldsymbol{\alpha}}(\lambda)\right]_{1,2}\overset{\lambda\to 0}{=}O(\lambda)=\nu^{(\boldsymbol{\alpha})}_{\infty,0} - \sum_{j=1}^{2d} \frac{\check{\mu}^{(\boldsymbol{\alpha})}_j}{q_j}\,\,\Rightarrow\,\, \nu^{(\boldsymbol{\alpha})}_{\infty,0}=\sum_{j=1}^{2d} \frac{\check{\mu}^{(\boldsymbol{\alpha})}_j}{q_j}.\eeq

\subsection{Proof of \autoref{Propcalpha}}\label{AppendixA11}
Starting from \eqref{ExpressionA12A11}, we may perform the same kind of computations for $\left[\check{A}_{\text{oper},\boldsymbol{\alpha}}(\lambda)\right]_{1,1}$ as in the previous section. By symmetry, we get that
\beq \left[\check{A}_{\text{oper},\boldsymbol{\alpha}}(\lambda)\right]_{1,1}=
\sum_{j=1}^{2d}\frac{\check{\rho}^{(\boldsymbol{\alpha})}_j}{\lambda-q_j},\eeq
where the coefficients $\left(\check{\rho}^{(\boldsymbol{\alpha})}_j\right)_{1\leq j\leq 2d}$ are obtained by looking at order $(\lambda-q_j)^{-3}$ of $\mathcal{L}_{\boldsymbol{\alpha}}\left[\left[\check{L}_{\text{oper}}(\lambda)\right]_{2,1}\right]$  in \eqref{Compat} giving
\beq \forall\, j\in \llbracket 1,2d\rrbracket\,:\, \check{\rho}_j^{(\boldsymbol{\alpha})}=-p_j\check{\mu_j}^{(\boldsymbol{\alpha})}.\eeq

\section{Proof of Theorem \ref{HamTheorem}}\label{AppendixProofHamiltonian}
\subsection{Preliminary results}
We start with the following lemma:

\begin{lemma}\label{PropsumC} For all $j\in \llbracket 1, 2d\rrbracket$:
\beq
\sum_{k=0}^{2d-2}\sum_{i=1}^{2d}\check{H}_{\infty,k}q_i^k\partial_{q_j} \check{\mu}^{\boldsymbol{(\alpha)}}_i+\sum_{i=1}^{2d}\frac{\check{H}_{0,1}}{q_i}\partial_{q_j}\check{\mu}^{\boldsymbol{(\alpha)}}_i=-\check{\mu}^{\boldsymbol{(\alpha)}}_j\left(\sum_{k=0}^{2d-2}k\check{H}_{\infty,k} q_j^{k-1}-\frac{\check{H}_{0,1}}{q_j}\right).
\eeq
\end{lemma}

\begin{proof}The proof follows from the expression relating $\left(\nu^{(\boldsymbol{\alpha})}_{\infty,k}\right)_{0\leq k\leq 2d-1}$ and $\left(\check{\mu}_j^{(\boldsymbol{\alpha})}\right)_{1\leq j\leq 2d}$ given by \eqref{RelationNuMu} and \eqref{Eqbis}. Taking the derivative with respect to $q_j$ and using the fact that the $\left(\nu^{(\boldsymbol{\alpha})}_{\infty,k}\right)_{0\leq k\leq 2d-1}$ are independent of $q_j$ gives:
\beq  \forall\, k\in \llbracket 0, 2d-2\rrbracket\,:\, \sum_{i=1}^{2d} (\partial_{q_j}\check{\mu}^{(\boldsymbol{\alpha})}_i) q_i^{k}=
-k \check{\mu}^{(\boldsymbol{\alpha})}_j q_j^{k-1}\,\,,\,\,
\sum_{i=1}^{2d}\frac{\partial_{q_j}\check{\mu}^{(\boldsymbol{\alpha})}_i}{q_i}=\frac{\check{\mu}^{(\boldsymbol{\alpha})}_j}{q_j^2}.
\eeq
Thus
\beq \sum_{k=0}^{2d-2}\sum_{i=1}^{2d}\check{H}_{\infty,k}q_i^k\partial_{q_j} \check{\mu}^{\boldsymbol{(\alpha)}}_i+\sum_{i=1}^{2d} \frac{\check{H}_{0,1}}{q_i}\partial_{q_j}\check{\mu}^{\boldsymbol{(\alpha)}}_i=-\check{\mu}^{(\boldsymbol{\alpha})}_j \left(\sum_{k=0}^{2d-2} k\check{H}_{\infty,k}q_j^{k-1}-  \frac{\check{H}_{0,1}}{q_j^2 }\right).
\eeq
so that the lemma is proved.
\end{proof}

We may now provide an alternative expression for $\mathcal{L}_{\boldsymbol{\alpha}}[p_j]$:

\begin{proposition}\label{PropLpjbis} Let $j\in \llbracket 1,2d\rrbracket$, we have an alternative expression for $\mathcal{L}_{\boldsymbol{\alpha}}[p_j]$:
\bea \label{Lpjbis} \mathcal{L}_{\boldsymbol{\alpha}}[p_j]&=& \sum_{1\leq i\neq j\leq 2d}\frac{(\check{\mu}^{(\boldsymbol{\alpha})}_i+\check{\mu}^{(\boldsymbol{\alpha})}_j)(p_i-p_j)}{(q_j-q_i)^2} +\frac{1}{2}\displaystyle{\sum_{\substack{(r,s)\in \llbracket 1,2d\rrbracket^2 \\ r\neq s }}} \frac{(p_s-p_r)(\partial_{q_j}\check{\mu}^{\boldsymbol{(\alpha)}}_r+\partial_{q_j}\check{\mu}^{\boldsymbol{(\alpha)}}_s)}{q_s-q_r}\cr
&&-\check{\mu}^{(\boldsymbol{\alpha})}_j\check{P}_2'(q_j)-\sum_{r=1}^{2d} (\partial_{q_j} \check{\mu}^{(\boldsymbol{\alpha})}_r)\left(\check{P}_2(q_r)+ p_r^2\right).
\eea
\end{proposition}

\begin{proof}Using Lemma \ref{PropsumC} we get that the expression \eqref{Lpj} for $\mathcal{L}[p_j]$ becomes:
\beq \mathcal{L}_{\boldsymbol{\alpha}}[p_j]= \sum_{1\leq i\neq j\leq 2d}\frac{(\check{\mu}^{(\boldsymbol{\alpha})}_i+\check{\mu}^{(\boldsymbol{\alpha})}_j)(p_i-p_j)}{(q_j-q_i)^2}-\check{\mu}^{(\boldsymbol{\alpha})}_j\check{P}_2'(q_j)-\sum_{i=1}^{2d} (\partial_{q_j}\check{\mu}^{\boldsymbol{(\alpha)}}_i)\left(\sum_{k=0}^{2d-2}\check{H}_{\infty,k}q_i^k +\sum_{i=1}^{2d} \frac{\check{H}_{0,1}}{q_i}\right).
\eeq
We now use \eqref{DefCi} to get
\bea \mathcal{L}_{\boldsymbol{\alpha}}[p_j]&=& \sum_{1\leq i\neq j\leq 2d}\frac{(\check{\mu}^{(\boldsymbol{\alpha})}_i+\check{\mu}^{(\boldsymbol{\alpha})}_j)(p_i-p_j)}{(q_j-q_i)^2} -\check{\mu}^{(\boldsymbol{\alpha})}_j\check{P}_2'(q_j) \cr&&
-\sum_{i=1}^{2d} (\partial_{q_j}\check{\mu}^{\boldsymbol{(\alpha)}}_i)\Big[p_i^2+\check{P}_2(q_i)+ \sum_{1\leq r\neq i\leq 2d}\frac{p_r-p_i}{q_i-q_r}\Big]\cr
&=& \sum_{1\leq i\neq j\leq 2d}\frac{(\check{\mu}^{(\boldsymbol{\alpha})}_i+\check{\mu}^{(\boldsymbol{\alpha})}_j)(p_i-p_j)}{(q_j-q_i)^2} -\check{\mu}^{(\boldsymbol{\alpha})}_j\check{P}_2'(q_j)\cr&&
-\sum_{i=1}^{2d} (\partial_{q_j}\check{\mu}^{\boldsymbol{(\alpha)}}_i)\left(p_i^2+\check{P}_2(q_i)\right)+\sum_{i=1}^{2d} (\partial_{q_j}\check{\mu}^{\boldsymbol{(\alpha)}}_i)\sum_{1\leq r\neq i\leq 2d}\frac{p_r-p_i}{q_r-q_i}.
\eea
The last sums may be split into a symmetric and anti-symmetric case: $\partial_{q_j}\check{\mu}^{\boldsymbol{(\alpha)}}_i= \frac{1}{2}(\partial_{q_j}\check{\mu}^{\boldsymbol{(\alpha)}}_i-\partial_{q_j}\check{\mu}^{\boldsymbol{(\alpha)}}_i)+ \frac{1}{2}(\partial_{q_j}\check{\mu}^{\boldsymbol{(\alpha)}}_i+\partial_{q_j}\check{\mu}^{\boldsymbol{(\alpha)}}_i)$. The term involving $\partial_{q_j}\check{\mu}^{\boldsymbol{(\alpha)}}_i-\partial_{q_j}\check{\mu}^{\boldsymbol{(\alpha)}}_i$ is trivially zero because the sum is anti-symmetric so that we end up with 
\bea \mathcal{L}_{\boldsymbol{\alpha}}[p_j]&=& \sum_{1\leq i\neq j\leq 2d}\frac{(\check{\mu}^{(\boldsymbol{\alpha})}_i+\check{\mu}^{(\boldsymbol{\alpha})}_j)(p_i-p_j)}{(q_j-q_i)^2} -\check{\mu}^{(\boldsymbol{\alpha})}_j\check{P}_2'(q_j)-\sum_{i=1}^{2d} (\partial_{q_j}\check{\mu}^{\boldsymbol{(\alpha)}}_i)\left(p_i^2+\check{P}_2(q_i)\right)\cr
&&+\frac{1}{2}\sum_{i=1}^{2d}\sum_{1\leq r\neq i\leq 2d} \frac{(p_r-p_i)(\partial_{q_j}\check{\mu}^{\boldsymbol{(\alpha)}}_i+\partial_{q_j}\check{\mu}^{\boldsymbol{(\alpha)}}_r)}{q_r-q_i},
\eea
ending the proof of \autoref{PropLpjbis}.
\end{proof}

\subsection{Proof of \autoref{HamTheorem}}

We may now proceed to the proof of Theorem \ref{HamTheorem}. We recall that the Hamiltonian is given by:

\beq\label{HamComputation}\text{Ham}^{(\boldsymbol{\alpha})}(\mathbf{q},\mathbf{p})=-\frac{1}{2}\displaystyle{\sum_{\substack{(i,j)\in \llbracket 1,2d\rrbracket^2 \\ i\neq j }}} \frac{(\check{\mu}^{(\boldsymbol{\alpha})}_i+\check{\mu}^{(\boldsymbol{\alpha})}_j)(p_i-p_j)}{q_i-q_j} - \nu^{(\boldsymbol{\alpha})}_{\infty,0}\sum_{j=1}^{2d}  p_j +\sum_{j=1}^{2d}\check{\mu}^{(\boldsymbol{\alpha})}_j\left(p_j^2+\check{P}_2(q_j)\right).
\eeq

It is a straightforward computation from \eqref{HamComputation} and from the fact that the $\left(\nu^{(\boldsymbol{\alpha})}_{\infty,k}\right)_{0\leq k\leq 2d-1}$ are independent of $q_j$ to get that
\bea -\frac{\partial \text{Ham}^{(\boldsymbol{\alpha})}(\mathbf{q},\mathbf{p})}{\partial q_j}&=&\displaystyle{\sum_{\substack{i\in \llbracket 1,2d\rrbracket \\ i\neq j }}} \frac{(\check{\mu}^{(\boldsymbol{\alpha})}_i+\check{\mu}^{(\boldsymbol{\alpha})}_j)(p_i-p_j)}{(q_i-q_j)^2}+\frac{1}{2}\displaystyle{\sum_{\substack{(r,s)\in \llbracket 1,2d\rrbracket^2 \\ r\neq s }}} \frac{(\partial_{q_j}\check{\mu}^{(\boldsymbol{\alpha})}_r+\partial_{q_j}\check{\mu}^{(\boldsymbol{\alpha})}_s)(p_r-p_s)}{q_r-q_s}\cr
&&-\sum_{i=1}^{2d}\partial_{q_j}(\check{\mu}^{(\boldsymbol{\alpha})}_i)\left(p_i^2+\check{P}_2(q_i)\right)-\check{\mu}^{(\boldsymbol{\alpha})}_j\check{P}_2'(q_j)\cr
&\overset{Prop. \ref{PropLpjbis}}{=}&\mathcal{L}_{\boldsymbol{\alpha}}[p_j].
\eea
Similarly a straightforward computation using the fact that $\left(\nu^{(\boldsymbol{\alpha})}_{\infty,k}\right)_{0\leq k\leq2d-1}$ are independent of $p_j$ gives:
\beq \frac{\partial \text{Ham}^{(\boldsymbol{\alpha})}(\mathbf{q},\mathbf{p})}{\partial p_j}=-\displaystyle{\sum_{\substack{i\in \llbracket 1,2d\rrbracket \\ i\neq j }}} \frac{\check{\mu}^{(\boldsymbol{\alpha})}_i+\check{\mu}^{(\boldsymbol{\alpha})}_j}{q_j-q_i} - \nu^{(\boldsymbol{\alpha})}_{\infty,0}  +2p_j\check{\mu}^{(\boldsymbol{\alpha})}_j,
\eeq
which is exactly $\mathcal{L}_{\boldsymbol{\alpha}}[q_j]$ given by \eqref{Lqj}.

\medskip

The last step is to verify that from \autoref{PropA12Form} and \autoref{PropDefCi2}:
\small{\bea &&\sum_{j=1}^g \check{\mu}_j^{\boldsymbol{(\alpha)}}\left(p_j^2+\check{P}_2(q_j)+ \underset{1\leq i\neq j\leq 2d}{\sum}\frac{p_i-p_j}{q_j-q_i}\right) =\begin{pmatrix}\check{\mu}_1^{\boldsymbol{(\alpha)}},\dots,\check{\mu}_{2d} ^{\boldsymbol{(\alpha)}}\end{pmatrix} \begin{pmatrix} p_1^2+\check{P}_2(q_1)+ \underset{1\leq i\neq 1\leq 2d}{\sum}\frac{p_i-p_1}{q_1-q_i}\\
p_2^2+\check{P}_2(q_2)+ \underset{1\leq i\neq 2\leq 2d}{\sum}\frac{p_i-p_2}{q_2-q_i}\\
\vdots\\
p_{2d}^2+\check{P}_2(q_{2d})+ \underset{1\leq i\neq 2d\leq 2d}{\sum}\frac{p_i-p_{2d}}{q_{2d}-q_i}\end{pmatrix}\cr
&&=\sum_{k=0}^{2d-2} \nu_{\infty,k+1}^{\boldsymbol{(\alpha)}}\check{H}_{\infty,k}+\nu_{\infty,0}^{\boldsymbol{(\alpha)}} \check{H}_{0,1},
\eea}
\normalsize{so} that \eqref{HamComputation} becomes
\beq\text{Ham}^{(\boldsymbol{\alpha})}(\mathbf{q},\mathbf{p})=\sum_{k=0}^{2d-2} \nu_{\infty,k+1}^{\boldsymbol{(\alpha)}}\check{H}_{\infty,k}+\nu_{\infty,0}^{\boldsymbol{(\alpha)}}\left(\check{H}_{0,1}-\sum_{j=1}^{2d} p_j\right).
\eeq

\section{Proof of \autoref{TheoLaxMatricesQinfty}}\label{AppendixLaxExpression}
The proof of $\check{L}_{1,2}(\lambda)$ and $\check{L}_{1,1}(\lambda)$ is standard. Since $\det \check{L}(\lambda)=\det \td{L}(\lambda)$, the entry $\check{L}_{2,1}(\lambda)$ is constrained by the fact that 
\bea \label{Impot}\det \check{L}(\lambda)=-\check{L}_{1,1}(\lambda)^2-\check{L}_{1,2}(\lambda)\check{L}_{2,1}(\lambda)&\overset{\lambda\to \infty}{=}&-\underset{j= 2 d-1}{\overset{4d}{\sum}}\left( \underset{m=0}{\overset{4d-j}{\sum}} t_{\infty,2 d +1-m}t_{\infty,j+m-2d+1}\right) \lambda^{j}+ O(\lambda^{2d-2}),\cr
&\overset{\lambda\to 0}{=}& -\frac{(t_{0,0})^2}{\lambda^2}+O(\lambda),
\eea
giving the expression for $\check{L}_{2,1}(\lambda)$ up to the term in $\lambda^{-1}$. Let us compute this term that we shall denote $\gamma$ in this proof. By definition of the determinant, we have:
\beq \check{L}_{2,1}(\lambda)=\frac{-\check{L}_{1,1}(\lambda)^2-\det \check{L}(\lambda)}{\check{L}_{1,2}(\lambda)}\,\Rightarrow \, \gamma= \Res_{\lambda\to 0}\frac{-\check{L}_{1,1}(\lambda)^2-\det \check{L}(\lambda)}{\check{L}_{1,2}(\lambda)}=\Res_{\lambda\to 0}\frac{-\check{L}_{1,1}(\lambda)^2+\frac{(t_{0,0})^2}{\lambda^2}+O(\lambda^{-1})}{\check{L}_{1,2}(\lambda)}. \eeq
Unfortunately, the numerator is not precise enough to derive the expression of the residue. On the contrary, we can use the same idea at infinity:
\bea \gamma&=& -\Res_{\lambda\to \infty} \frac{-\check{L}_{1,1}(\lambda)^2-\det \check{L}(\lambda)}{\check{L}_{1,2}(\lambda)}\cr
&\overset{\eqref{Impot}}{=}&-\Res_{\lambda\to \infty} \frac{-\check{L}_{1,1}(\lambda)^2+\underset{j= 2 d-1}{\overset{4d}{\sum}}\left( \underset{m=0}{\overset{4d-j}{\sum}} t_{\infty,2 d +1-m}t_{\infty,j+m-2d+1}\right) \lambda^{j}+ O(\lambda^{2d-2})}{\check{L}_{1,2}(\lambda) }\cr
&=&-\Res_{\lambda\to \infty} \left(\frac{-\check{L}_{1,1}(\lambda)^2+\underset{j= 2 d-1}{\overset{4d}{\sum}}\left( \underset{m=0}{\overset{4d-j}{\sum}} t_{\infty,2 d +1-m}t_{\infty,j+m-2d+1}\right) \lambda^{j}}{\check{L}_{1,2}(\lambda) } + O(\lambda^{-2})\right)\cr
&=&-\Res_{\lambda\to \infty} \frac{-\check{L}_{1,1}(\lambda)^2+\underset{j= 2 d-1}{\overset{4d}{\sum}}\left( \underset{m=0}{\overset{4d-j}{\sum}} t_{\infty,2 d +1-m}t_{\infty,j+m-2d+1}\right) \lambda^{j}}{\check{L}_{1,2}(\lambda) } ,
\eea
because we have $\check{L}_{1,2}(\lambda)=O(\lambda^{2d})$.

Regarding the auxiliary matrix, the gauge transformation is 
\beq \check{\Psi}_{\text{oper}}(\lambda)= \begin{pmatrix}1&0\\ \check{L}_{1,1}(\lambda)&\check{L}_{1,2}(\lambda) \end{pmatrix} \check{\Psi}(\lambda)\coloneqq G(\lambda) \check{\Psi}(\lambda),\eeq
so that
\beq \check{A}_{\boldsymbol{\alpha}}(\lambda)= G(\lambda)^{-1} \check{A}_{\text{oper}, \boldsymbol{\alpha}}(\lambda) G(\lambda)-  G(\lambda)^{-1}\mathcal{L}_{\boldsymbol{\alpha}}\left[ G(\lambda)\right], \eeq
and in particular
\bea \label{Super}\left[\check{A}_{\boldsymbol{\alpha}}(\lambda)\right]_{1,2}&=&\check{L}_{1,2}(\lambda) \left[\check{A}_{\text{oper}, \boldsymbol{\alpha}}(\lambda)\right]_{1,2},\cr
\left[\check{A}_{\boldsymbol{\alpha}}(\lambda)\right]_{1,1}&=&\left[\check{A}_{\text{oper}, \boldsymbol{\alpha}}(\lambda)\right]_{1,1}+\check{L}_{1,1}(\lambda) \left[\check{A}_{\text{oper}, \boldsymbol{\alpha}}(\lambda)\right]_{1,2},
\eea
so that the local expansions given by \autoref{PropAsymptoticExpansionA12} gives the expression of $[\check{A}_{\boldsymbol{\alpha}}(\lambda)]_{1,1}$ and $[\check{A}_{\boldsymbol{\alpha}}(\lambda)]_{1,2}$. For  $[\check{A}_{\boldsymbol{\alpha}}(\lambda)]_{2,1}$, we use that
\bea \label{tdA21} \left[\check{A}_{\boldsymbol{\alpha}}(\lambda)\right]_{2,1}&=&-\mathcal{L}_{\boldsymbol{\alpha}}\left[\frac{\check{L}_{1,1}(\lambda)}{\check{L}_{1,2}(\lambda)}\right]+\frac{\partial_\lambda\left[\check{A}_{\text{oper}, \boldsymbol{\alpha}}(\lambda)\right]_{1,1}}{\check{L}_{1,2}(\lambda)} +\left[\check{A}_{\text{oper}, \boldsymbol{\alpha}}(\lambda)\right]_{1,2}\frac{\left[\check{L}_{\text{oper}, \boldsymbol{\alpha}}(\lambda)\right]_{2,1}}{\check{L}_{1,2}(\lambda)}\cr
&&+\frac{\check{L}_{1,1}(\lambda)}{\check{L}_{1,2}(\lambda)}\left(\frac{\check{L}_{1,1}(\lambda)}{\check{L}_{1,2}(\lambda)}\left[\check{A}_{\boldsymbol{\alpha}}(\lambda)\right]_{1,2}-2\left[\check{A}_{\boldsymbol{\alpha}}(\lambda)\right]_{1,1}\right).
\eea
We have that
\bea -\mathcal{L}_{\boldsymbol{\alpha}}\left[\frac{\check{L}_{1,1}(\lambda)}{\check{L}_{1,2}(\lambda)}\right]&\overset{\lambda\to 0}{=}&\mathcal{L}_{\boldsymbol{\alpha}}\left[\frac{t_{0,0}}{t_{\infty,2d+1} Q_{\infty,0}\lambda} +O(1)\right]=-\frac{t_{0,0}\mathcal{L}_{\boldsymbol{\alpha}}[Q_{\infty,0}]}{t_{\infty,2d+1}(Q_{\infty,0})^2} \lambda^{-1}+O(1),\cr
&\overset{\lambda\to \infty}{=}&O(\lambda^{-2}).
\eea
The quantity $\frac{\partial_\lambda\left[\check{A}_{\text{oper}, \boldsymbol{\alpha}}(\lambda)\right]_{1,1}}{\check{L}_{1,2}(\lambda)} $ is regular at infinity and zero so it does not contribute. Let us now discuss the term $\left[\check{A}_{\text{oper}, \boldsymbol{\alpha}}(\lambda)\right]_{1,2}\frac{\left[\check{L}_{\text{oper}, \boldsymbol{\alpha}}(\lambda)\right]_{2,1}}{\check{L}_{1,2}(\lambda)}$. At $\lambda=0$, it is regular because $\left[\check{A}_{\text{oper}, \boldsymbol{\alpha}}(\lambda)\right]_{1,2}=O(\lambda)$ and $\left[\check{L}_{\text{oper}, \boldsymbol{\alpha}}(\lambda)\right]_{2,1} =O(\lambda^{-1})$. At infinity, we have:
\bea
&&\left[\check{A}_{\text{oper}, \boldsymbol{\alpha}}(\lambda)\right]_{1,2}\frac{\left[\check{L}_{\text{oper}, \boldsymbol{\alpha}}(\lambda)\right]_{2,1}}{\check{L}_{1,2}(\lambda)}\overset{\lambda\to \infty}{=}\cr&&
\left(\frac{-\underset{j= 2 d-1}{\overset{4d}{\sum}}\left(\underset{m=0}{\overset{4d-j}{\sum}} t_{\infty,2 d +1-m}t_{\infty,j+m-2d+1}\right) \lambda^{j} +O(\lambda^{2d-2})}{-t_{\infty,2d+1}\left(\underset{k=0}{\overset{2d-1}{\sum}} Q_{\infty,k}\lambda^k +\lambda^{2d}\right)} \right)\left(\sum_{i=0}^{2d} \frac{\nu^{(\boldsymbol{\alpha})}_{\infty,i}}{\lambda^i} +O\left(\lambda^{-2d-1}\right) \right)\cr
&\overset{\lambda\to \infty}{=}&
\left(\frac{\underset{j= 2 d}{\overset{4d}{\sum}}\left(\underset{m=0}{\overset{4d-j}{\sum}} t_{\infty,2 d +1-m}t_{\infty,j+m-2d+1}\right) \lambda^{j} }{t_{\infty,2d+1}\left(\underset{k=0}{\overset{2d-1}{\sum}} Q_{\infty,k}\lambda^k +\lambda^{2d}\right)} +O(\lambda^{-1})\right)\left(\sum_{i=0}^{2d} \frac{\nu^{(\boldsymbol{\alpha})}_{\infty,i}}{\lambda^i} +O\left(\lambda^{-2d-1}\right) \right)\cr
&&\overset{\lambda\to \infty}{=}\frac{\underset{j= 2 d}{\overset{4d}{\sum}}\underset{m=0}{\overset{4d-j}{\sum}}\underset{i=0}{\overset{2d}{\sum}} t_{\infty,2 d +1-m}t_{\infty,j+m-2d+1}\nu^{(\boldsymbol{\alpha})}_{\infty,i} \lambda^{j-i}}{t_{\infty,2d+1}\left(\underset{k=0}{\overset{2d-1}{\sum}} Q_{\infty,k}\lambda^k +\lambda^{2d}\right)} +O(\lambda^{-1})\cr
&&\overset{k=j-i}{=}\frac{\underset{k=2d}{\overset{4d}{\sum}}\underset{j=k}{\overset{4d}{\sum}}\underset{m=0}{\overset{4d-j}{\sum}} t_{\infty,2 d +1-m}t_{\infty,j+m-2d+1}\nu^{(\boldsymbol{\alpha})}_{\infty,j-k} \lambda^{k}
}{t_{\infty,2d+1}\left(\underset{k=0}{\overset{2d-1}{\sum}} Q_{\infty,k}\lambda^k +\lambda^{2d}\right)} +O(\lambda^{-1}),
\eea
with $\nu_{\infty,2d}^{(\boldsymbol{\alpha})}\coloneqq -\underset{\lambda\to \infty}{\Res}\lambda^{2d+1}\frac{\left[\check{A}_{\boldsymbol{\alpha}}(\lambda)\right]_{1,2}}{\check{L}_{1,2}(\lambda)}$ obtained as the coefficient $\lambda^{-2d}$ of $\left[\check{A}_{\text{oper}, \boldsymbol{\alpha}}(\lambda)\right]_{1,2}=\frac{\left[\check{A}_{\boldsymbol{\alpha}}(\lambda)\right]_{1,2}}{\check{L}_{1,2}(\lambda)}$. Now let us consider $\left[\check{A}_{\text{oper}, \boldsymbol{\alpha}}(\lambda)\right]_{1,2}\frac{\left[\check{L}_{\text{oper}, \boldsymbol{\alpha}}(\lambda)\right]_{2,1}}{\check{L}_{1,2}(\lambda)}$ at $\lambda\to 0$. It has a simple pole since $\left[\check{L}_{\text{oper}, \boldsymbol{\alpha}}(\lambda)\right]_{2,1}=-\frac{t_{0,0}(1-t_{0,0})}{\lambda^2}+O(\lambda^{-1})$ and $\check{L}_{1,2}(\lambda)$ is regular and $\left[\check{A}_{\text{oper}, \boldsymbol{\alpha}}(\lambda)\right]_{1,2}$ has a simple zero from \autoref{PropAsymptoticExpansionA12}. In order to evaluate the simple pole, we use \eqref{Super} to get:
\bea \left[\check{A}_{\text{oper}, \boldsymbol{\alpha}}(\lambda)\right]_{1,2}\frac{\left[\check{L}_{\text{oper}, \boldsymbol{\alpha}}(\lambda)\right]_{2,1}}{\check{L}_{1,2}(\lambda)}&=&\left[\check{A}_{\boldsymbol{\alpha}}(\lambda)\right]_{1,2}\frac{\left[\check{L}_{\text{oper}, \boldsymbol{\alpha}}(\lambda)\right]_{2,1}}{\check{L}_{1,2}(\lambda)^2}\cr
&=&\frac{t_{0,0}(1-t_{0,0})\left(\nu_{\infty,2d-1}^{\boldsymbol{(\alpha)}} +\underset{k=1}{\overset{2d-1}{\sum}} \nu_{\infty,k-1}^{(\boldsymbol{\alpha})} Q_{\infty,k}\right) }{t_{\infty,2d+1}(Q_{\infty,0})^2\lambda} +O(1).\eea

\bibliographystyle{plain}
\bibliography{Biblio}

\end{document}